\documentclass[11pt,letter]{article}
\usepackage{natbib}

\pdfoutput=1 

\usepackage{fullpage}
\usepackage{times}
\usepackage{soul}
\usepackage{url}
\usepackage[utf8]{inputenc}
\usepackage{color}
\usepackage{amsmath,amssymb,amsthm}
\usepackage{algorithm,empheq}
\usepackage{algpseudocode}
\algtext*{EndIf}
\algtext*{EndWhile}
\algtext*{EndFor}

\usepackage{libertine}

\usepackage[a4paper, total={170mm,257mm}, left=20mm, top=20mm]{geometry}

\newlength\mylen

\newcommand{\OPT}{\ensuremath{\mathrm{OPT}}}

\newcommand{\np}{\ensuremath{\mathrm{NP}}}

\newcommand{\MNSW}{\ensuremath{\mathsf{MNSW}}}
\newcommand{\MESW}{\ensuremath{\mathsf{MESW}}}
\newcommand{\MPSW}{\ensuremath{\mathsf{PMSW}}}

\newcommand{\PSW}{\operatorname{PSW}}
\newcommand{\ESW}{\operatorname{ESW}}
\newcommand{\NSW}{\operatorname{NSW}}

\newtheorem{theorem}{Theorem}

\newtheorem{proposition}{Proposition}
\newtheorem{lemma}{Lemma}

\newtheorem{remark}{Remark}
\newtheorem{claim}{Claim}
\newtheorem{corollary}{Corollary}

\title{Maximizing $p$-Mean Social Welfare in the High-Multiplicity Setting: Few Agent Types and Few Item Types}
\author{Trung Thanh Nguyen\thanks{DataOptLab, National Economics University, Hanoi 11616, Vietnam} \and Khaled Elbassioni\thanks{Khalifa University of Science and Technology, Abu Dhabi 127788, UAE}}
\date{\today}

\begin{document}

\maketitle

\begin{abstract}
The $p$-mean welfare objective unifies several classical social welfare criteria
for the allocation of indivisible goods. We study its maximization under
additive nonnegative utilities when item and agent multiplicities are
encoded in binary. For every fixed finite rational $p<1$, $p\neq0$,
we show that the problem is NP-hard with only two item types. Nash
welfare ($p=0$) and egalitarian welfare ($p=-\infty$) are NP-hard with
three item types. These results hold both for computing an optimal
allocation and for rational-threshold decision, and the utilities and
threshold can be required to be positive integers. We also show that maximizing \(p\)-mean social welfare is strongly NP-hard with one agent type and an unrestricted number of item types, for every fixed finite rational \(p<1\) and for \(p=-\infty\). A quantitative gap in this reduction
rules out an FPTAS in the latter setting unless
$\mathsf{P}=\mathsf{NP}$. On the positive side, for a fixed number of
item types and an arbitrary number of agent types, we give an FPTAS for
every fixed $p\in\mathbb Q\cup\{-\infty\}$. Its running time is
polynomial in the compact input length and in $1/\varepsilon$, and it
returns a compressed allocation. For a fixed number of agent types and
an unrestricted number of item types, we give a PTAS for every fixed
finite rational $p<1$, also in the fully compact model. We further
give explicit compact-model proofs of the classical exact allocation
algorithms for one item type, and for egalitarian welfare with two
item types. These results essentially settle the complexity and approximability of \(p\)-mean welfare maximization with few item types
and/or few agent types, leaving only the exact complexity of Nash welfare maximization with two item types unresolved in the small-item-type classification.
\end{abstract}

\newpage
\tableofcontents

\clearpage
\newpage

\section{Introduction}
Fair division studies how a set of resources can be allocated among agents with possibly conflicting preferences so that the resulting allocation is fair in a well-defined sense. Much of the modern literature, motivated by applications such as course allocation, inheritance division, spectrum and bandwidth allocation, and the assignment of computational resources, has focused on the setting of indivisible goods, where each item must be allocated in its entirety to a single agent \cite{lan-rot:b-2nd-edition:economics-and-computation-fair-division-of-indivisible-goods,bou-che-mau:b:handbook-comsoc-fair-allocation-of-indivisible-goods}. Prominent fairness notions include {\em envy-freeness} and its relaxations EF1 and EFX, {\em proportionality} and its relaxation PROP1, and the {\em maximin share} (MMS) guarantee; we refer to the survey \cite{AmanatidisABFLMVW23} and the references therein for formal definitions and an overview of the extensive literature on these notions. Another natural fairness objective is to maximize the \emph{egalitarian social welfare} (ESW), defined as the minimum utility attained by any agent, thereby directly promoting the welfare of the least well-off agent.

Alongside fairness, efficiency is a fundamental consideration in resource allocation. A desirable allocation should make effective use of the available resources and, ideally, be Pareto optimal or optimize an appropriate social welfare objective. Two of the most prominent objectives are the \emph{utilitarian social welfare} (USW), defined as the sum of the agents' utilities, and the \emph{Nash social welfare} (NSW), defined as the geometric mean of their utilities. While USW emphasizes aggregate efficiency and ESW focuses on the least well-off agent, NSW is celebrated for achieving a compelling balance between fairness and efficiency \cite{CaragiannisKMP016}. This interplay naturally motivates the study of welfare objectives that provide a unified framework encompassing these classical criteria. One such framework is given by the family of \emph{generalized (or $p$-mean) welfare} functions, formally defined in Section~\ref{sec:notation}. By varying $p\in[-\infty,1]$, this family smoothly interpolates between the three classical welfare objectives: $p=1$ gives USW divided by $n$ and hence the same maximization problem, $p=0$ gives NSW, and $p=-\infty$ gives ESW. Moreover, decreasing $p$ places progressively greater emphasis on agents receiving lower utility, thereby providing a natural parameter for controlling the trade-off between aggregate efficiency and distributive fairness. Consequently, $p$-mean welfare has attracted growing research interest over the last  decade~\cite{BarmanBKS20,ChaudhuryGM21,barman2020uniform,GargMV22,BarmanVerma2022,BarmanKhanMaiti2022,Viswanathan25,Procaccia_2024,Eckart_EC-2024,krishna2025pmean,huang2025long,Suksompong26}.

Exact $p$-mean welfare maximization is computationally intractable in general, motivating the study of restrictions on the heterogeneity of items and agents. In many applications, the population of agents and the catalogue of items consist of a modest number of distinct \emph{types}, each replicated many times. For example, course seats, shifts, machine time, and standardized resources can have many interchangeable copies, while agents sharing a role may have identical valuations. Such an instance can be described by the numbers of agent types $K$ and item types $q$, the per-type utility values, and the corresponding multiplicities. Because the total numbers $n$ of agents and $m$ of items can be exponential in the length of this description, running time polynomial in $n$ and $m$ need not be polynomial in the compact input length. This distinction motivates the study of parameterized and high-multiplicity fair division; see, e.g.,~\cite{Woeginger97,ngu-rot:j:how-to-decrease-envy,bli-bre-nie:c:complexity-of-efficient-and-envy-free-resource-allocation-few-agents,BarmanKV18a,Anari_2018,jan-maa:j:eptas-max-min,NguyenRothe2023AIJ,GargMV22,NER24}. We also note that a substantial body of work has studied fair division in the high-multiplicity setting with respect to fairness notions such as EF, EF1, EFX, and MMS; see,  e.g.~~\cite{azizchore23,Mahara23,Chaudhury24,DavidiukEtAl2026AAAI,GargMurhekarQin2024IJCAI,MaEtAl2025ISAAC,Prakash_EC25,Garg_EFX_25}.

We study how the numbers of item and agent types determine the complexity and approximability of $p$-mean welfare maximization in this compact model. Our principal contributions are as follows.

\begin{itemize}
    \item We establish $\np$-hardness with two item types for every fixed
    finite rational $p<1$, $p\neq0$, and with three item types for Nash
    and egalitarian welfare. The reductions use positive rational
    utilities and rational welfare thresholds; both can also be made
    positive integers. They prove hardness of exact allocation and of
    rational-threshold decision already for explicitly listed agents.
    We also establish
    strong $\np$-hardness with one agent type
    for every fixed finite rational $p<1$ and for $p=-\infty$, with an
    unrestricted number of item types. A quantitative gap in this
    reduction rules out an FPTAS in the latter setting unless
    $\mathsf{P}=\mathsf{NP}$.
    \item For a fixed number of item types and an arbitrary number of
    agent types, with both multiplicities encoded in binary, we give an
    FPTAS for every fixed $p\in\mathbb Q\cup\{-\infty\}$. The running
    time is polynomial in the compact input length and in
    $1/\varepsilon$, and the output is a compressed allocation.
    \item For a fixed number of agent types and an unrestricted number
    of item types, with both multiplicities encoded in binary, we give a
    PTAS for every fixed finite rational $p<1$.
\end{itemize}

For completeness, we also give explicit proofs in the compact model of
two exact allocation results following from classical methods. With
one item type, an optimal allocation is computable in polynomial time
for every fixed $p\in\mathbb Q\cup\{-\infty\}$. The nonzero finite
welfare parameters $p$ below one use separable concave resource allocation~\cite{GalilM79};
the remaining cases admit direct constructions. With two item types,
egalitarian welfare is exactly solvable in polynomial time by the
two-job-type methods of McCormick, Smallwood, and
Spieksma~\cite{McCormickSmallwoodSpieksma2001}. Both results allow an
arbitrary number of agent types and binary multiplicities.

Table~\ref{tab:small-type-complexity} summarizes the resulting complexity
of exact allocation. The
polynomial-time entries concern constructing an optimal allocation;
they do not automatically imply polynomial-time comparison of its
welfare with an arbitrary rational threshold. This distinction is
necessary even with one item type, as explained in
Remark~\ref{rem:one-item-type-threshold}.

\begin{table}[t]
\centering
\small
\begin{tabular}{c|c|c|c}
 $q$ & Fixed rational $p<1$, $p\neq0$ & Nash: $p=0$ & Egalitarian: $p=-\infty$\\
 \hline
 $1$ & Polynomial time~\cite{GalilM79} & Polynomial time [here] & Polynomial time~\cite{McCormickSmallwoodSpieksma2001}\\
 $2$ & $\mathsf{NP}$-hard [here] & Not settled here & Polynomial time~\cite{McCormickSmallwoodSpieksma2001}\\
 $3$ & $\mathsf{NP}$-hard [here] & $\mathsf{NP}$-hard [here] & $\mathsf{NP}$-hard [here]\\
 $4$ & $\mathsf{NP}$-hard [here] & $\mathsf{NP}$-hard [here] & $\mathsf{NP}$-hard [here]
\end{tabular}
\caption{Status established here for computing an exact optimal allocation,
with at most $q$ item types, an unrestricted number of agent types, and binary item and agent
multiplicities. The hardness entries also
hold for rational-threshold decision. ``Not settled here'' denotes the case not resolved by the results in this paper and thus remains open. References identify the
classical methods underlying the polynomial-time results; their compact
implementations are proved here. The marker [here] indicates a proof in
this paper, including the elementary one-item-type Nash result, which
is not claimed as new.}
\label{tab:small-type-complexity}
\end{table}

The polynomial-time entries follow from
Theorems~\ref{thm:one-item-type-exact} and
\ref{thm:two-item-types-egalitarian}; the hardness entries follow from
Theorem~\ref{np-hard-pmsw}. All fixed values of $q$ in the table admit
the compact FPTAS of Theorem~\ref{thm:fixed-item-types-ptas}.
The utilitarian case $p=1$ is polynomial-time solvable for every $q$
by assigning each item type to an agent who values it most.
Consequently, the item-type boundaries for exact allocation are sharp
for every fixed finite rational $p<1$, $p\neq0$, and for egalitarian
welfare: hardness starts at two and three types, respectively. Among
the cases in the table, only two-type Nash welfare remains unsettled
by our results.

\section{Related Work}

The study of fair and efficient allocations has yielded a substantial body
of algorithmic work. We express approximation guarantees as fractions of the
optimal welfare; an approximation loss of $\alpha\geq1$ corresponds to a
welfare guarantee of $1/\alpha$. For $\MNSW$, Nguyen and
Rothe~\cite{ngu-rot:j:how-to-decrease-envy} gave a $1/m$ welfare guarantee,
and Cole and Gkatzelis~\cite{col-gka:c:approximating-nash} obtained the first
constant-factor guarantee. Barman, Krishnamurthy, and
Vaish~\cite{bar-kri-vai:c:finding-fair-and-efficient-allocations} gave an
approximation loss of $1.45$, corresponding to a welfare guarantee of
$1/1.45\approx0.68966$. On the hardness side, Garg, Hoefer, and
Mehlhorn~\cite{GargHoeferMehlhorn2024} proved that achieving any fixed welfare
guarantee greater than $\sqrt{7/8}$ is $\np$-hard.

For general $\MESW$, Chakrabarty, Chuzhoy, and
Khanna~\cite{ChakrabartyCK09} gave an approximation loss of
$\widetilde O(m^\epsilon)$ in time $m^{O(1/\epsilon)}$ for
$\epsilon\geq8\log\log m/\log m$, where $m$ is the number of items.
In particular, every fixed $\epsilon>0$ gives a polynomial-time welfare
guarantee of $\widetilde\Omega(m^{-\epsilon})$, whereas taking $\epsilon$
near the stated lower bound yields a polylogarithmic approximation in
quasipolynomial time. Bateni, Charikar, and
Guruswami~\cite{bat-cha-gur:c:maxmin-allocation} obtained related guarantees
for special cases. Bez\'{a}kov\'{a} and
Dani~\cite{bez-dan:j:allocating-indivisible-goods} showed that achieving any
fixed welfare guarantee greater than $1/2$ is $\np$-hard.

For $p$-mean welfare, Barman et al.~\cite{BarmanBKS20} proved that, for
every fixed $p\in(0,1)$, computing a $1/(m-n+1)$-approximation is
$\np$-hard even under additive valuations. Their reduction sets $m=n$, so
the required factor is one and the result establishes hardness of exact
optimization. They also gave an $\Omega(1/n)$ welfare guarantee for all
$p\leq1$, including subadditive valuations. Chaudhury et
al.~\cite{ChaudhuryGM21} obtained a comparable welfare guarantee together
with approximate EFX.

A separate line of work studies instances with few agent or item types.
For a fixed number of machine types, Jansen and
Maack~\cite{jan-maa:j:eptas-max-min} gave an EPTAS for maximizing the minimum
load, corresponding to $\ESW$, extending earlier work for identical
agents~\cite{Woeginger97}. Their input lists jobs individually, so this
running-time guarantee should be distinguished from polynomial time in
binary item multiplicities. In the high-multiplicity setting, Nguyen,
Elbassioni, and Rothe~\cite{NER24} gave a PTAS for $\NSW$ with a fixed
number of agent types. They also established $\np$-hardness for $\ESW$ and
$\NSW$ with seven item types, using a modification of a scheduling
reduction of Kouteck\'{y} and Zink~\cite{KouteckyZ20}, and gave a PTAS for
$\ESW$ with two item types. Our hardness results in
Section~\ref{sec:hardness} reduce the item-type bound to three for both
$\NSW$ and $\ESW$, and establish hardness with two item types for every
fixed finite rational $p<1$, $p\neq0$.

Exact algorithms are available for some smaller type restrictions.
Galil and Megiddo~\cite{GalilM79} reduce separable concave integer
resource allocation to selection in implicit sorted arrays. This underlies
the nonzero finite-order case $p<1$ of the one-item-type algorithm in
Section~\ref{subsec:one-item-type-exact}; we account explicitly for
binary agent multiplicities and exact comparisons of fixed-degree
algebraic numbers. McCormick, Smallwood, and
Spieksma~\cite{McCormickSmallwoodSpieksma2001} give polynomial-time
algorithms for high-multiplicity scheduling with two job types,
including unrelated machines and compact machine classes in
Theorems~5.1--5.2, and covering variants in their Section~6. The
two-item-type egalitarian result in
Section~\ref{sec:two-item-types-egalitarian} is a consequence of these
methods. We provide a direct integer-hull proof and compressed
reconstruction for the welfare setting. Related exact allocation results
with both numbers of types fixed were established by Nguyen and
Rothe~\cite{NguyenRothe2023AIJ} using integer-cone methods;
Appendix~\ref{sec:exact-egalitarian-appendix} gives a direct application
of the underlying theorem of Goemans and Rothvo{\ss}~\cite{GoemansR20}.

Garg et al.~\cite{GargMV22} gave a PTAS for $p$-mean welfare with symmetric
agents having identical additive valuations. Our fixed-agent-type result
extends this setting to a fixed number of valuation types in the compact
input model. Their broader study also treats asymmetric Nash welfare;
arbitrary agent weights are outside the model considered here. For
identical subadditive valuations, Barman and
Sundaram~\cite{barman2020uniform} obtained constant-factor guarantees for
$p$-mean welfare in the demand-oracle model. Viswanathan and Zick~\cite{Viswanathan25} studied a general
framework for fair allocation under matroid-rank valuations, including
welfare objectives and connections to strategyproofness.

\section{Allocation Setting and Notation}
\label{sec:notation}

For every nonnegative integer $r$, write $[r]:=\{1,\ldots,r\}$, with
$[0]:=\varnothing$.  There are $n\geq1$ agents and $q\geq1$ item types.  Item
type $j\in[q]$ has multiplicity
$m_j\in\mathbb Z_{\geq0}$, and we put
$m:=\sum_{j=1}^q m_j,$ $m_{\max}:=\max_{j\in[q]} m_j$, and $\boldsymbol m:=(m_1,\ldots,m_q)$. 
We use the terms \emph{item} and \emph{good} interchangeably.
When agents are grouped into $K\geq1$ types, $A_1,\ldots,A_K$ form a partition of
$[n]$.  Type $k\in[K]$ contains $n_k:=|A_k|$ agents, and
$n=\sum_{k=1}^K n_k$, $\boldsymbol n:=(n_1,\ldots,n_K)$.

Utilities are additive and nonnegative.  In a setting with arbitrary agents,
$u_{i,j}\in\mathbb Q_{\geq0}$ denotes the value that agent $i$ assigns to one
item of type $j$.  For every count vector
$c=(c_1,\ldots,c_q)\in\mathbb R_{\geq0}^q$, define $u_i(c):=\sum_{j=1}^q u_{i,j}c_j.$ 
If $i\in A_k$, the item values depend only on the agent type, and we write
them as $u_{k,j}$.  In this case, define $u_k(c):=\sum_{j=1}^q u_{k,j}c_j.$ 
A bundle is an integral count vector
$x_i=(x_{i,1},\ldots,x_{i,q})\in\mathbb Z_{\geq0}^q$; its utility is
$u_i(x_i)$, or equivalently $u_k(x_i)$ when $i\in A_k$.
For a bundle assignment $x=(x_1,\ldots,x_n)$, we also write
$u_i(x):=u_i(x_i)$.  The assignment is an \emph{allocation} if $ \sum_{i=1}^n x_{i,j}=m_j$ for every $j\in[q].$ 
It is a \emph{partial allocation} if the preceding equalities are replaced by
inequalities.  Since all utilities are nonnegative, every partial allocation
can be completed without decreasing any agent's utility or the resulting
welfare.

For finite $p>0$, and also for $-\infty<p<0$ when every agent has positive utility, the
$p$-mean social welfare of a bundle assignment $x$ is
\[
 \PSW_p(x)
 :=
 \left(\frac1n\sum_{i=1}^n u_i(x_i)^p\right)^{1/p}.
\]
We refer to \(p\) interchangeably as the \emph{welfare parameter} or the \emph{order} of the generalized mean.
The limiting cases used in this paper are
\[
 \NSW(x)=\PSW_0(x):=\left(\prod_{i=1}^n u_i(x_i)\right)^{1/n}
 \qquad\text{and}\qquad
 \ESW(x)=\PSW_{-\infty}(x):=\min_{i\in[n]}u_i(x_i).
\]
For $p\leq0$, we use the standard convention that $\PSW_p(x)=0$ whenever
some agent has utility zero.  We denote the optimal welfare by ${\OPT}_p
 :=
 \max\{\PSW_p(x):x\text{ is an allocation}\}.$

For every finite $p$, define the transformed utility function
\[
 \phi_p(t)
 :=
 \begin{cases}
  t^p,&p>0,\\
  \log t,&p=0,\\
  -t^p,&p<0.
 \end{cases}
\]
For $p>0$, its domain is $[0,\infty)$ and $\phi_p(0):=0$; for $p\leq0$,
it is used only on $(0,\infty)$.  Throughout, $\log$ denotes the natural
logarithm.  For any bundle assignment on which every term is defined, put
\[
 \Phi_p(x):=\sum_{i=1}^n\phi_p(u_i(x_i)).
\]
For finite $p$, maximizing $\PSW_p$ is equivalent to maximizing $\Phi_p$: for $p>0$ this
follows from the monotonicity of
$z\mapsto z^{1/p}$, for $p=0$ from the monotonicity of the logarithm, and
for $p<0$ from the fact that $z\mapsto z^{1/p}$ is decreasing.  When
$p\leq0$, this equivalence is applied only to positive-utility bundle
assignments.

For finite $p\neq0$ and any bundle assignment $x=(x_1,\ldots,x_n)$ for which
the powers are defined, put
\[
 Q_p(x):=\sum_{i=1}^n u_i(x_i)^p.
\]
For positive-utility bundle assignments, put
\[
 \Lambda(x):=\sum_{i=1}^n\log u_i(x_i).
\]
Whenever the displayed quantities are defined, they satisfy
\[
 \PSW_p(x)=\left(\frac{Q_p(x)}n\right)^{1/p}
 \quad (p\in\mathbb R\setminus\{0\}),
 \qquad
 \PSW_0(x)=\exp\left(\frac{\Lambda(x)}n\right).
\]

The item multiplicities $m_j$ are encoded in binary throughout. All
algorithmic results use the fully compact representation: the input lists
the $K$ agent types, their utility vectors, and their multiplicities
$n_k$ in binary. Explicitly listed agents are included as the special
case $K=n$ and $n_k=1$ for every $k\in[K]$. In general, both $n$ and
$m$ may be exponential in the input length. A \emph{compressed allocation}
lists, for every agent type, each distinct bundle vector together with
the number of agents receiving that bundle. The algorithms in
Sections~\ref{sec:fixed-item-types} and~\ref{sec:constant-agent-types},
and the exact algorithm recalled in
Appendix~\ref{sec:exact-egalitarian-appendix}, return allocations in
this representation.

We study the problem $\MPSW$ of maximizing $\PSW_p$. By an exact
allocation algorithm, we mean an algorithm that returns an optimal
allocation in compressed form. The rational-threshold decision problem
asks whether $\OPT_p\geq\tau$ for a given rational $\tau$. Since
welfare values may be irrational, these two tasks require separate
complexity statements; see Remark~\ref{rem:one-item-type-threshold}.
A polynomial-time approximation scheme (PTAS) returns, for every fixed
$\varepsilon\in(0,1)$,
an allocation of welfare at least $(1-\varepsilon)\OPT_p$ in polynomial time.
An efficient PTAS (EPTAS) has a running time whose polynomial degree is
independent of $\varepsilon$, and a fully polynomial-time approximation
scheme (FPTAS) runs in time polynomial also in $1/\varepsilon$.
The welfare parameter $p$ is fixed in each algorithmic theorem, and constants
may depend on $p$.

When a rational accuracy parameter is part of the input, the running time
includes the time needed to read its encoding. We may first replace the
requested accuracy $\varepsilon_{\mathrm{req}}\in(0,1)$ by a dyadic power
$\varepsilon=2^{-h}$, with $h\in\mathbb Z_{\geq1}$ chosen so that
$\varepsilon_{\mathrm{req}}/2<\varepsilon\leq\varepsilon_{\mathrm{req}}$.
The new accuracy has binary encoding length $O(1+\log(1/\varepsilon))$,
and a $(1-\varepsilon)$ guarantee implies the requested guarantee.
For the algorithm descriptions and their analyses, $\varepsilon$ denotes
this preprocessed accuracy; the factor of at most two does not affect the
PTAS or FPTAS running-time classifications.

\medskip
\noindent\textbf{Rational approximation of the transformed objectives.}
The algorithms below replace values of $\phi_p$ by rational coefficients.
The following two lemmas give a common construction and quantify the
resulting error, including when multiplicities are encoded in binary. The standard (but technical) proof of the first lemma is given in Appendix~\ref{sec:tech-lemmas}.

\begin{lemma}[Directed rational approximation]
\label{lem:directed-rational-approximation}
Fix a finite rational $p$. Given a rational $\eta\in(0,1)$ and a rational
$v$ in the domain of $\phi_p$, one can compute a rational number
$\widetilde\phi_{p,\eta}(v)$ satisfying
\begin{align}
 \frac{v^p}{1+\eta}
 &\leq\widetilde\phi_{p,\eta}(v)\leq v^p,
 &&p>0,\label{eq:prelim-rational-positive}\\
 \log v-\log(1+\eta)
 &\leq\widetilde\phi_{0,\eta}(v)\leq\log v,
 &&p=0,\label{eq:prelim-rational-logarithmic}\\
 v^p
 &\leq-\widetilde\phi_{p,\eta}(v)\leq(1+\eta)v^p,
 &&p<0.\label{eq:prelim-rational-negative}
\end{align}
For $p>0$ and $v=0$, the output is zero. The running time and output
encoding length are polynomial in the binary encoding lengths of $v$
and $\eta$, with constants depending on $p$.
\end{lemma}

\begin{lemma}[Weighted objectives and comparisons]
\label{lem:weighted-rational-comparison}
Fix a finite rational $p$ and a rational $\eta\in(0,1)$. Let
$A=((\lambda_i,v_i))_{i=1}^s$ be a finite list of nonnegative rational
weights $\lambda_i$ and rational utilities $v_i$ in the domain of
$\phi_p$. Extend the objective notation to such lists by putting
\[
 W(A):=\sum_{i=1}^s\lambda_i,
 \qquad
 Q_p(A):=\sum_{i=1}^s\lambda_i v_i^p\quad(p\neq0),
 \qquad
 \Lambda(A):=\sum_{i=1}^s\lambda_i\log v_i\quad(p=0).
\]
Using the coefficients from
Lemma~\ref{lem:directed-rational-approximation}, define $\widetilde\Phi_{p,\eta}(A)
 :=\sum_{i=1}^s\lambda_i\widetilde\phi_{p,\eta}(v_i).$ 
This rational sum can be computed in time polynomial in the total
encoding length of the list and of $\eta$, and it satisfies
\begin{align}
 \frac{Q_p(A)}{1+\eta}
 &\leq\widetilde\Phi_{p,\eta}(A)\leq Q_p(A),
 &&p>0,\label{eq:prelim-weighted-positive}\\
 \Lambda(A)-W(A)\log(1+\eta)
 &\leq\widetilde\Phi_{0,\eta}(A)\leq\Lambda(A),
 &&p=0,\label{eq:prelim-weighted-logarithmic}\\
 Q_p(A)
 &\leq-\widetilde\Phi_{p,\eta}(A)\leq(1+\eta)Q_p(A),
 &&p<0.\label{eq:prelim-weighted-negative}
\end{align}
If $B$ is another such list and
$\widetilde\Phi_{p,\eta}(B)\geq\widetilde\Phi_{p,\eta}(A)$, then
\begin{equation}
\label{eq:prelim-objective-comparison}
 \begin{cases}
  Q_p(B)\geq Q_p(A)/(1+\eta),&p>0,\\
  \Lambda(B)\geq\Lambda(A)-W(A)\log(1+\eta),&p=0,\\
  Q_p(B)\leq(1+\eta)Q_p(A),&p<0.
 \end{cases}
\end{equation}
\end{lemma}

\begin{proof}
Multiplying each bound in
Lemma~\ref{lem:directed-rational-approximation} by its nonnegative
weight and summing gives
\eqref{eq:prelim-weighted-positive}--\eqref{eq:prelim-weighted-negative}.
For $p>0$, the assumed comparison consequently gives
\[
 Q_p(B)\geq\widetilde\Phi_{p,\eta}(B)
 \geq\widetilde\Phi_{p,\eta}(A)\geq\frac{Q_p(A)}{1+\eta}.
\]
For $p=0$, it gives $ \Lambda(B)\geq\widetilde\Phi_{0,\eta}(B)
 \geq\widetilde\Phi_{0,\eta}(A)
 \geq\Lambda(A)-W(A)\log(1+\eta).$ 
For $p<0$, negating the assumed comparison gives $ Q_p(B)\leq-\widetilde\Phi_{p,\eta}(B)
 \leq-\widetilde\Phi_{p,\eta}(A)\leq(1+\eta)Q_p(A).$

These are the three claims in~\eqref{eq:prelim-objective-comparison}.
Each coefficient has polynomial encoding length by the preceding
lemma. Multiplying by a binary-encoded rational weight and summing
over the explicitly listed entries preserves polynomial encoding
length and takes polynomial time. In particular, integer weights need
not be expanded into individual copies, and the numerical size of
$W(A)$ does not introduce an additional running-time factor.
\end{proof}

\section{Hardness Results}
\label{sec:hardness}

We establish hardness with two item types for every fixed nonzero
rational order below one, and with three item types for Nash and
egalitarian welfare. The two-type reduction makes each agent choose
between two prescribed bundle patterns. The three-type reductions use
separate constructions for the Nash and egalitarian objectives.
We then consider the complementary restriction of one agent type,
where the number of item types is unrestricted.

\begin{theorem}[Hardness with few item types]
\label{np-hard-pmsw}
For every fixed rational $p<1$ with $p\neq0$, deciding whether
$\OPT_p\geq\tau$ is $\mathsf{NP}$-hard with two item types.
For $p=0$ and $p=-\infty$, the same problem is $\mathsf{NP}$-hard
with three item types. These assertions hold for explicitly listed
agents, strictly positive integer item values, positive integer item
multiplicities, and a positive integer threshold, all encoded in binary.
Computing an exact optimal allocation is $\mathsf{NP}$-hard under
the same respective restrictions.
\end{theorem}

\begin{proof}
The two-type assertion follows from
Theorem~\ref{thm:two-types-finite-nonzero-hardness} and
Corollary~\ref{cor:two-types-finite-integer-hardness}.
The three-type assertions follow from
Theorems~\ref{thm:three-types-nash-hardness} and
\ref{thm:three-types-egalitarian-hardness}, respectively.
\end{proof}

We use \textsc{Subset Sum} for the two-type construction. For the
three-type constructions, we record the following restricted form
of \textsc{Partition}.

\begin{lemma}[Hardness of~\textsc{Balanced-Equal-Cardinality-Partition}]
\label{lem:becp}
It is $\mathsf{NP}$-hard to decide, for positive integers
$a_1,\ldots,a_n$ with both $n$ and $A:=\sum_i a_i$ even, whether
some set $I\subseteq[n]$ satisfies
$|I|=n/2$ and $\sum_{i\in I}a_i=A/2$.
\end{lemma}

\begin{proof}
Given a \textsc{Partition} instance $b_1,\ldots,b_t$ of positive
integers with sum $B$, form the $2t$ numbers
\[
 2b_1+1,\ldots,2b_t+1,
 \underbrace{1,\ldots,1}_{t\text{ copies}}.
\]
Their total is $2B+2t$. A set of exactly $t$ of these numbers has sum
$B+t$ if and only if the original numbers whose shifted copies it
contains have sum $B/2$. Any selected original subset can be padded
to cardinality $t$ using the added ones. The construction is polynomial
and its number of entries and total sum are even.
\end{proof}

\subsection{Two item types and nonzero finite orders}
\label{subsec:two-types-finite-hardness}

We first prove two inequalities whose equality cases are the two
bundle patterns used in the reduction. Their quantitative gaps allow
us to replace the ideal utility coefficients by rationals while
preserving the reduction. The proofs of the following technical lemmas are provided in Appendix~\ref{sec:tech-lemmas}.

\begin{lemma}\label{lem:two-type-positive-support}
Fix $p\in(0,1)$. Let $a,C,Z$ be positive integers such that
$a<C\leq Z$, and put
\[
 \gamma:=\frac{1-p}{p},\qquad
 H:=pC+(1-p)a,\qquad
 \alpha:=a^{\gamma},\qquad \beta:=H^{1/p}.
\]
For all integers $x,y\in\{0,\ldots,Z\}$,
\begin{equation}\label{eq:two-type-positive-support}
 F(x,y):=\frac{(\alpha x+\beta y)^p}{p}-x-Cy
 \leq \gamma a.
\end{equation}
Equality holds if and only if $(x,y)=(a,0)$ or $(x,y)=(0,1)$.
For every other such pair,
\begin{equation}\label{eq:two-type-positive-gap}
 F(x,y)\leq\gamma a-\delta_+,
 \qquad
 \delta_+:=\frac{\min\{p,1-p\}}{2Z^2}.
\end{equation}
\end{lemma}

\begin{lemma}\label{lem:two-type-negative-support}
Fix $r>0$. Let $a,C,Z$ be positive integers such that
$a<C\leq Z$ and $C<(1+1/r)a$. Put
\[
 \gamma:=\frac{r+1}{r},\qquad
 H:=(r+1)a-rC,\qquad
 \alpha:=a^{-\gamma},\qquad \beta:=H^{-1/r}.
\]
For all integers $x,y\in\{0,\ldots,Z\}$ with $x+y>0$,
\begin{equation}\label{eq:two-type-negative-support}
 G(x,y):=\frac{(\alpha x+\beta y)^{-r}}r+x+Cy
 \geq\gamma a.
\end{equation}
Equality holds if and only if $(x,y)=(a,0)$ or $(x,y)=(0,1)$.
For every other such pair,
\begin{equation}\label{eq:two-type-negative-gap}
 G(x,y)\geq\gamma a+\delta_-,
 \qquad
 \delta_-:=\frac{\min\{1,r\}}{2Z^{\lceil r\rceil+2}}.
\end{equation}
\end{lemma}

\begin{lemma}
\label{lem:two-type-rational-between-powers}
Fix a positive integer $s$. Given positive rational numbers
$R_0<R_1$, one can compute a positive rational number $z$ such that $ R_0<z^s<R_1$. 
in time polynomial in the total binary encoding length of $R_0$
and $R_1$. The encoding length of $z$ is also polynomial.
\end{lemma}

\begin{theorem}[Hardness with two item types]
\label{thm:two-types-finite-nonzero-hardness}
For every fixed rational $p<1$ with $p\neq0$, computing an
allocation that maximizes $\PSW_p$ is
$\mathsf{NP}$-hard even with two item types, explicitly listed agents,
and strictly positive rational item values. The problem of
deciding whether the optimal $p$-mean welfare is at least a
given positive rational threshold is $\mathsf{NP}$-hard under the same
restrictions. The item multiplicities are encoded in binary.
\end{theorem}

\begin{proof}
We reduce from \textsc{Subset Sum}. Let the instance consist
of positive integers $b_1,\ldots,b_t$ and an integer target
$T$, where $0\leq T\leq A:=\sum_{i=1}^t b_i$.
Choose a positive integer $D$ depending only on $p$, with
$D\geq2$; if $p=-r<0$, also require $D>r$.
Put
\begin{equation}\label{eq:two-type-source-construction}
 \begin{split}
 Q&:=D(A+1),\\
 a_i&:=Q+b_i\quad(1\leq i\leq t),\\
 a_i&:=Q\quad(t<i\leq2t),\\
 C&:=Q+A+1,\qquad n:=2t,\\
 m_1&:=tQ+T,\qquad m_2:=t,\qquad
 W:=\sum_{i=1}^n a_i=2tQ+A.
 \end{split}
\end{equation}
There are $n$ agents and two item types, with multiplicities
$m_1$ and $m_2$. Call an allocation \emph{canonical} if agent
$i$ receives either $(a_i,0)$ or $(0,1)$.
A full canonical allocation gives the second bundle to exactly
$t$ agents and the first bundle to the other $t$ agents. Thus
it exists if and only if exactly $t$ of the numbers $a_i$ sum
to $m_1$. By construction, this is equivalent to choosing some
of the original numbers $b_i$ with sum $T$: any such choice
can be padded to cardinality $t$ using the additional $t$
copies of $Q$.

Fix the integer $ Z:=\max\{2,n,m_1,m_2,C,a_1,\ldots,a_n\}.$ 
It has polynomial binary encoding length, and every coordinate
of every feasible bundle is at most $Z$. For every $i$,
$a_i<C$ and $C-a_i\geq1$. In the negative case, $\frac C{a_i}\leq\frac CQ=1+\frac1D<1+\frac1r.$ 
We first construct ideal utility coefficients, which may be
irrational, and then replace them by rational coefficients.

\medskip\noindent
\textbf{The ideal instance for $0<p<1$.}
Put $\gamma:=(1-p)/p$ and, for every agent $i$, define
\[
 H_i:=pC+(1-p)a_i,\qquad
 \alpha_i:=a_i^\gamma,\qquad
 \beta_i:=H_i^{1/p},\qquad
 u_i(x_i):=\alpha_i x_{i,1}+\beta_i x_{i,2}.
\]
By Lemma~\ref{lem:two-type-positive-support}, $\frac{u_i(x_i)^p}{p}-x_{i,1}-Cx_{i,2}\leq\gamma a_i,$ 
with equality precisely for the two canonical bundles.
Consequently, for an allocation $x$,
\begin{equation}\label{eq:two-type-positive-total}
 \Phi(x):=\sum_{i=1}^n\frac{u_i(x_i)^p}{p}
 \leq U:=\gamma W+m_1+Cm_2.
\end{equation}
Canonical allocations attain $U$, whereas every noncanonical
allocation has value at most $U-\delta_+$, where
$\delta_+=\min\{p,1-p\}/(2Z^2)$.
Maximizing $\Phi$ is equivalent to maximizing the $p$-mean.

\medskip\noindent
\textbf{The ideal instance for $p=-r<0$.}
Put $\gamma:=(r+1)/r$ and define
\[
 H_i:=(r+1)a_i-rC,\qquad
 \alpha_i:=a_i^{-\gamma},\qquad
 \beta_i:=H_i^{-1/r},\qquad
 u_i(x_i):=\alpha_i x_{i,1}+\beta_i x_{i,2}.
\]
The preceding bounds imply $0<H_i<a_i$.
For every nonempty bundle, Lemma~\ref{lem:two-type-negative-support}
gives $ \frac{u_i(x_i)^{-r}}r+x_{i,1}+Cx_{i,2}\geq\gamma a_i.$ 
Therefore, for an allocation $x$ with no empty
bundle,
\begin{equation}\label{eq:two-type-negative-total}
 \Psi(x):=\sum_{i=1}^n\frac{u_i(x_i)^{-r}}r
 \geq L:=\gamma W-m_1-Cm_2.
\end{equation}
Canonical allocations attain $L$, whereas every noncanonical
allocation with no empty bundle has cost at least
$L+\delta_-$, where
$\delta_-=\min\{1,r\}/(2Z^{\lceil r\rceil+2})$.
Furthermore,
\begin{equation}\label{eq:two-type-negative-benchmark-positive}
 L=t\left(\frac{2D}r-1\right)(A+1)
   +\frac{r+1}rA-T>0,
\end{equation}
because $D>r$ and $T\leq A$.
Minimizing $\Psi$ is equivalent to maximizing the $p$-mean
on allocations with positive utility for every agent.
An allocation with an empty bundle has $p$-mean zero.
Since $m_1\geq n$, a positive allocation exists, obtained by
giving one item of the first type to every agent and
distributing the remaining items arbitrarily. Hence no
optimal allocation has an empty bundle.

\medskip\noindent
\textbf{Rational utility coefficients.}
We now give the required polynomial-time construction of a
rational instance. Let $\delta$ denote $\delta_+$ or
$\delta_-$, according to the sign of $p$.
We first bound the magnitude of a single transformed
objective term on every feasible nonempty bundle.
For $0<p<1$, set $E:=\max\{\lceil\gamma\rceil,\lceil1/p\rceil\}$, and $F:=\frac{2Z^{E+1}}p.$ 
Both ideal coefficients are at most $Z^E$, because
$a_i\leq Z$ and $H_i\leq C\leq Z$. Thus
$u_i(x_i)\leq2Z^{E+1}$, and
$u_i(x_i)^p/p\leq F$. For $p=-r<0$, set $E:=\lceil\gamma\rceil$, $J:=\lceil Er\rceil$, $F:=\frac{Z^J}r.$ 
Both ideal coefficients are at least $Z^{-E}$: this is
immediate for $\alpha_i$, and for $\beta_i$ it follows from
$0<H_i<a_i\leq Z$ and $1/r<\gamma\leq E$.
Every nonempty integer bundle therefore has utility at least
$Z^{-E}$, giving $u_i(x_i)^{-r}/r\leq F$.

Define the rational constant
\[
 c_p:=
 \begin{cases}
  1,&0<p<1,\\
  \max\{1,r2^{\lceil r\rceil+1}\},&p=-r<0,
 \end{cases}
 \qquad
 \rho:=\min\left\{\frac14,\frac{\delta}{8nc_pF}\right\}.
\]
For each $i$, choose positive rational approximations
$\widehat\alpha_i,\widehat\beta_i$ satisfying
\[
 (1-\rho)\alpha_i\leq\widehat\alpha_i\leq(1+\rho)\alpha_i,
 \qquad
 (1-\rho)\beta_i\leq\widehat\beta_i\leq(1+\rho)\beta_i.
\]
All ideal coefficients are fixed rational powers of positive rationals
with polynomial encoding length. Apply
Lemma~\ref{lem:directed-rational-approximation} with accuracy $\rho$
and the corresponding exponent, taking the magnitude of its output
when that exponent is negative. The resulting approximation to any
such power $w$ lies in $[w/(1+\rho),w]$ or in $[w,(1+\rho)w]$.
Both intervals lie in $[(1-\rho)w,(1+\rho)w]$, since
$1/(1+\rho)\geq1-\rho$. This gives the required approximations in
polynomial time and with polynomial encoding length: the rational
$\rho$ has polynomial encoding length, because the exponents occurring
in $\delta$ and $F$ depend only on the fixed parameter $p$.

Write
$\widehat u_i(x_i):=\widehat\alpha_i x_{i,1}+\widehat\beta_i x_{i,2}$.
For every nonempty bundle $x_i$ there is an $e\in[-\rho,\rho]$
with $\widehat u_i(x_i)=(1+e)u_i(x_i)$. If $0<p<1$, then $\bigl|(1+e)^p-1\bigr|\leq\rho,$ 
because $(1-\rho)^p\geq1-\rho$ and
$(1+\rho)^p\leq1+\rho$; if $p=-r<0$, the derivative of
$z^{-r}$ on $[1/2,3/2]$ has absolute value at most
$r2^{r+1}\leq c_p$, and hence $\bigl|(1+e)^{-r}-1\bigr|\leq c_p\rho.$ 
It follows that, for $0<p<1$, rounding changes each summand $u_i^p/p$ in
$\Phi$ by at most $\rho F$. For $p=-r<0$, it changes each
summand $u_i^{-r}/r$ in $\Psi$ by at most $c_p\rho F$.
Thus, in both cases, each summand changes by at most $c_p\rho F$. Summing these bounds over all $n$ agents
shows that the absolute change in $\Phi$ or $\Psi$ is at most
\(
 n c_p\rho F\leq\frac{\delta}{8},
\)
by the choice of $\rho$.
For positive $p$, this also covers an empty bundle, whose
term remains zero. For negative $p$, empty bundles continue
to have zero welfare, and all nonempty bundles continue to
have positive utility.

All constructed item values are strictly positive. Every agent still values the second item type strictly more than the first after approximation; hence the two item types remain distinct. Indeed,
in either sign case the ideal ratio $\beta_i/\alpha_i$
is greater than $a_i\geq Q\geq4$, whereas relative
approximation with $\rho\leq1/4$ multiplies this ratio
by a factor of at least $(1-\rho)/(1+\rho)\geq3/5$.

\medskip\noindent
\textbf{Hardness of computing an optimal allocation.}
In a yes instance (of \textsc{Subset Sum}), there is a canonical allocation. For
$0<p<1$, its rounded transformed objective is at least
$U-\delta/8$, while every noncanonical allocation has
rounded objective at most $U-7\delta/8$.
For $p=-r<0$, its rounded transformed cost is at most
$L+\delta/8$, while every noncanonical positive allocation
has rounded cost at least $L+7\delta/8$.
Thus every optimal allocation of the rational instance is
canonical in a yes instance. In a no instance no full
canonical allocation exists. An exact allocation algorithm
would consequently decide the source instance by checking
whether its returned allocation is canonical. This check
uses only integer bundle coordinates and does not require
comparison of algebraic welfare values. There are only
$n=2t$ explicitly listed agents, so expanding a compressed
output for this check also takes polynomial time.

\medskip\noindent
\textbf{Hardness with a rational welfare threshold.}
For $0<p<1$, choose a positive rational $\tau$ satisfying
\begin{equation}\label{eq:two-type-positive-threshold}
 U-\frac{7\delta}8
 <\frac{n\tau^p}p
 <U-\frac\delta8.
\end{equation}
The interval is positive since $U\geq m_1\geq2$ and
$\delta\leq1/16$. For every allocation $x$ of the rounded
instance,
\[
 \PSW_p(x)\geq\tau
 \quad\Longleftrightarrow\quad
 \sum_{i=1}^n\frac{\widehat u_i(x_i)^p}p
 \geq\frac{n\tau^p}p.
\]
The bounds just established show that a yes instance has
an allocation strictly above the threshold, whereas every
allocation in a no instance is strictly below it.

For $p=-r<0$, choose a positive rational $\tau$ satisfying
\begin{equation}\label{eq:two-type-negative-threshold}
 L+\frac\delta8
 <\frac{n\tau^{-r}}r
 <L+\frac{7\delta}8.
\end{equation}
For allocations with no empty bundle,
\[
 \PSW_{-r}(\mathcal A)\geq\tau
 \quad\Longleftrightarrow\quad
 \sum_{i=1}^n\frac{\widehat u_i(\mathcal A_i)^{-r}}r
 \leq\frac{n\tau^{-r}}r.
\]
Again the yes and no instances are strictly separated.
Allocations containing an empty bundle have $p$-mean welfare zero which is strictly less than the positive threshold $\tau$.

For completeness, both thresholds can be constructed with
polynomially many rational bits. In the case $0<p<1$, write
$p=s/d$ in lowest terms. If
$a:=U-7\delta/8$ and $b:=U-\delta/8$, then
\eqref{eq:two-type-positive-threshold} is equivalent to $ \left(\frac{pa}n\right)^d
 <\tau^s<
 \left(\frac{pb}n\right)^d.$ 
For the negative case write $r=s/d$, and put
$a:=L+\delta/8$ and $b:=L+7\delta/8$.
Condition~\eqref{eq:two-type-negative-threshold} is
equivalent to $\left(\frac n{rb}\right)^d
 <\tau^s<
 \left(\frac n{ra}\right)^d.$ 
In both cases the endpoints are positive rationals of
polynomial encoding length. Lemma~\ref{lem:two-type-rational-between-powers}
therefore computes the required rational $\tau$ in
polynomial time. This proves the decision hardness as well
as the allocation hardness.
\end{proof}

\begin{corollary}\label{cor:two-types-finite-integer-hardness}
The hardness assertions of
Theorem~\ref{thm:two-types-finite-nonzero-hardness} also hold
with strictly positive integer item values. In the decision
problem, the welfare threshold may also be required to be
a positive integer.
\end{corollary}

\begin{proof}
Take the rational instance constructed in the theorem.
Let $M$ be the product of the denominators of all its utility
coefficients and, for the decision problem, of its threshold
$\tau$. Multiplying every utility coefficient by $M$ produces
strictly positive integers and preserves the two distinct
item types. The $p$-mean is positively homogeneous:
multiplying all utilities by $M$ multiplies every allocation's
welfare by $M$. Thus the set of optimal allocations is unchanged,
and the decision comparison is preserved by replacing $\tau$
with the positive integer $M\tau$. The binary encoding length
of $M$ is at most the sum of the encoding lengths of the
original denominators and is therefore polynomial.
\end{proof}

\subsection{Three item types and Nash welfare}
\label{subsec:three-types-nash-hardness}

The preceding construction uses the nonzero power of each utility.
For Nash welfare we use a different construction: two ordinary item
families determine the agents' choices, and a third item type enforces
the required sum within one family.

\begin{theorem}[Nash-welfare hardness with three item types]
\label{thm:three-types-nash-hardness}
For Nash welfare, deciding whether the optimum is at least a specified
positive rational threshold is $\mathsf{NP}$-hard with three item types and
explicitly listed agents. All item multiplicities, utilities, and the threshold may be
required to be positive integers of polynomial binary encoding length.
Computing an exact optimal allocation is $\mathsf{NP}$-hard under the same
restriction.
\end{theorem}

\begin{proof}
We reduce from \textsc{Balanced-Equal-Cardinality-Partition}, which
is $\mathsf{NP}$-hard by Lemma~\ref{lem:becp}. Let
$a_1,\ldots,a_n$ be positive integers, where $n$ and
$A:=\sum_i a_i$ are even. We construct one agent for each integer $a_i$, so the
resulting welfare instance also has exactly $n$ agents. Set
\[
 s:=n/2,\qquad B:=A/2,\qquad S:=A+1,\qquad R:=A+2,
\]
and
\[
 \delta:=\frac{1}{4n^3R^2S^2},\qquad
 L:=\frac{2R}{\delta}=8n^3R^3S^2,\qquad
 \gamma:=\frac{1}{n(L+A)^2},\qquad
 \varepsilon:=\frac{\gamma}{8n}.
\]
In particular, $\gcd(R,S)=1$, $L>RA$, and every parameter has
polynomial binary encoding length.

\medskip
\noindent\textbf{An ideal construction and its gap.}
We first analyze real coefficients that will subsequently be replaced
by positive rational numbers with a certified error bound.
There is one agent for each $a_i$, and there are three item types:
\[
\begin{array}{c|c|c}
 \text{Item type}&\text{Multiplicity}&\text{Value to agent }i\\ \hline
 \text{First ordinary}&Ss-B&(L-a_i)/S\\[2pt]
 \text{First rewarded}&B&(L-a_i)/S+1\\[2pt]
 \text{Second ordinary}&Rs&V_i/R
\end{array}
\qquad V_i:=L\exp(-a_i/L).
\]
All multiplicities are positive integers. The three value vectors are
distinct: the rewarded value exceeds the ordinary first-family value,
and $V_i/R<L/R<(L-a_i)/S$, where the last inequality follows from
$L>Ra_i$ and $R-S=1$. For an allocation, let
$f_i$ be agent $i$'s number of first-family items, let $g_i$ be its
number of second-family items, and let $r_i$ be its number of rewarded
items. Define $z_i:=f_i/S+g_i/R.$ 
Full allocations satisfy $\sum_i z_i=n$ and $\sum_i r_i=B$.
It suffices to consider full allocations because every item has
strictly positive value. The ideal utility for agent $i$ is $ u_i=(L-a_i)f_i/S+V_i g_i/R+r_i.$ 
Since $0<V_i<L$ and $r_i\leq f_i\leq Rz_i$,
$u_i\leq(L+R)z_i$.

Suppose first that $z_i\neq1$ for at least one agent $i$. Each $z_i$ is a nonnegative
integer multiple of $1/(RS)$, and their sum is $n$. Therefore $ \sum_i(z_i-1)^2\geq\frac{2}{R^2S^2}.$ 
Let $f(t)=\sqrt{t}$. For $0<t\leq n$, since $n\geq1$, $f''(t)=-\frac{1}{4t^{3/2}}
 \leq-\frac{1}{4n^2}.$ 
Taylor's theorem about $t=1$, together with
$f(1)=1$ and $f'(1)=1/2$, yields
\[
 \sqrt{t}\leq
 1+\frac{t-1}{2}-\frac{(t-1)^2}{8n^2}
 \qquad (0<t\leq n).
\]
The same inequality holds at $t=0$.
Applying this inequality to each $z_i\in[0,n]$
and averaging, we obtain
\begin{align*}
 \frac1n\sum_i\sqrt{z_i}
 \leq 1+\frac{1}{2n}\sum_i(z_i-1)
       -\frac{1}{8n^3}\sum_i(z_i-1)^2
 =1-\frac{1}{8n^3}\sum_i(z_i-1)^2
 \leq1-\frac{1}{4n^3R^2S^2}
 =1-\delta,
\end{align*}
where the equality uses $\sum_i z_i=n$.
For a nonnegative vector $v$, write $\NSW(v):=(\prod_i v_i)^{1/n}$.
By the arithmetic--geometric mean
inequality applied to $\sqrt{z_1},\ldots,\sqrt{z_n}$,
\begin{align*}
 \NSW(u)
 &\leq (L+R)\NSW(z)\leq (L+R)\left(\frac1n\sum_i\sqrt{z_i}\right)^2\leq (L+R)(1-\delta)^2\leq (L+R)(1-\delta)\\
 &=L-R-R\delta<L-R,
\end{align*}
where we used $u_i\leq(L+R)z_i$, $0<\delta<1$ and $L\delta=2R$.
The inequality remains valid if a coordinate is zero. 
Write $ F_*:=n\log L-B/L$ 
for the ideal target log-product. If all utilities are positive, then
\[
  F_*-\sum_i\log u_i>F_*-n\log(L-R)
 =n\log\frac{L}{L-R}-\frac BL
 \geq\frac{nR-B}{L}
 \geq\frac1L
 \geq\gamma,
\]
where $-\log(1-x)\geq x$ was applied to $x=R/L$.
Thus every nonunit-load allocation has log-product strictly below
$F_*-\gamma$. An allocation with a zero utility has zero Nash welfare
and also cannot attain any positive threshold.

Suppose now that $z_i=1$ for every agent. Coprimality of $R$ and $S$
implies $(f_i,g_i)\in\{(S,0),(0,R)\}.$ 
Let $I:=\{i\in[n]:(f_i,g_i)=(S,0)\}$ be the set of agents receiving $S$ first-family items and put $W:=\sum_{i\in I}a_i$.
The total number of first-family items forces $|I|=s$. Every
rewarded item belongs to a first-family agent, so $ \sum_{i\in I}u_i=sL+B-W.$
 For $i\notin I$, $\log u_i=\log L-a_i/L$. Jensen's inequality
therefore yields
\[
 \sum_i\log u_i
 \leq s\log\left(L+\frac{B-W}{s}\right)
       +(n-s)\log L-\frac{A-W}{L}.
\]
The interval between $L$ and $L+(B-W)/s$ is contained in
$[L-A,L+A]$. On this interval the second derivative of $\log t$
is at most $-1/(L+A)^2$. Hence Taylor’s theorem therefore gives $\sum_i\log u_i
 \leq n\log L-\frac{A-B}{L}
       -\frac{(B-W)^2}{2s(L+A)^2}
 =F_*-\frac{(B-W)^2}{n(L+A)^2}.$ 
In particular, if $W\ne B$, integrality gives a gap ($F_*-\sum_i\log u_i$) of at least
$\gamma$. 
If $W=B$ but some first-family utility differs from $L$, put $d_i:=u_i-L$ for $i\in I$.
The deviations $d_i$ are integers satisfying
\(
 \sum_{i\in I}d_i=0.
\)
Since they are not all zero, at least one is positive and at least
one is negative. Consequently,
\(
 {\sum}_{i\in I}d_i^2\geq2.
\)
Each first-family utility lies in $[L-A,L+A]$, because $r_i\leq B$.
Taylor's theorem, together with
$(\log t)''=-1/t^2\leq-1/(L+A)^2$ on this interval, gives $\log u_i
 \leq\log L+\frac{d_i}{L}
       -\frac{d_i^2}{2(L+A)^2}.$ 
Summing over $i\in I$ and using the preceding observations yields $\sum_{i\in I}\log u_i
 \leq |I|\log L-\frac{1}{(L+A)^2}.$ 
Combining this with the previously established bound for the
remaining agents, we obtain
\[
 \sum_i\log u_i
 \leq F_*-\frac{1}{(L+A)^2}
 \leq F_*-\gamma,
\]
giving again a log-product gap of $\gamma$.

If the source instance has a balanced partition $I$, give each
$i\in I$ exactly $a_i$ rewarded items and $S-a_i$ ordinary
first-family items, and give each remaining agent $R$ second-family
items. All multiplicities are met exactly. The first-family utilities
are $L$, and the total log-product is $F_*$. Conversely, every
allocation that is not of this form for some balanced partition has
ideal log-product at most $F_*-\gamma$, or zero Nash welfare. Indeed, consider an allocation with positive Nash welfare which has $\sum_i\log u_i> F_*-\gamma$. By the above arguments, we must have $z_i=1$ for all $i$. Letting $I$ be the set of agents receiving $S$ first-family items and putting $W:=\sum_{i\in I}a_i$, we must have $W=B$ and hence, $u_i=L$ for every $i\in I$.
Since $u_i=L-a_i+r_i$, $(i\in I),$ we must have $r_i=a_i$. Thus each $i\in I$ receives exactly
$a_i$ rewarded items and $S-a_i$ ordinary first-family items,
while each agent outside $I$ receives exactly $R$ second-family
items. Moreover, $|I|=s$ and $\sum_{i\in I}a_i=B$, so $I$ is a
balanced partition. 

\medskip
\noindent\textbf{Replacing the ideal coefficients by rationals.}
For an odd integer $N\geq3$, define the rational polynomial $E_N(t):=\sum_{h=0}^{N}\frac{(-t)^h}{h!}.$
 For $0\leq t\leq1$, the alternating-series bounds imply
\[
 E_N(t)\leq e^{-t},\qquad
 e^{-t}-E_N(t)\leq\frac{1}{(N+1)!}\leq2^{-N}.
\]
Moreover, the odd partial sums are nondecreasing in $N$, and
$E_3(t)=1-t+t^2/2-t^3/6\geq E_3(1)=1/3$ on $[0,1]$, so \(E_N(t)\geq E_3(t)\geq 1/3\) for odd $N\ge 3$.
Since $0<E_N(t)\leq e^{-t}$ for odd $N$, the inequality
$\log(1+x)\leq x$ for $x\geq0$ gives
\begin{align*}
 0\leq\log(e^{-t})-\log E_N(t)
 &=\log\!\left(1+\frac{e^{-t}-E_N(t)}{E_N(t)}\right)\leq\frac{e^{-t}-E_N(t)}{E_N(t)}
 \leq3\cdot2^{-N}.
\end{align*}
Choose an odd $N\geq3$ satisfying $3\cdot2^{-N}\leq\varepsilon$.
This requires only $N=O(\log n+\log(L+A))$.

In the actual reduction, replace $V_i$ by the positive rational $\widetilde V_i:=L \cdot E_N(a_i/L).$ 
Use the same first-family values and all the same multiplicities.
The second-family item value is now $\widetilde V_i/R$.
The preceding bounds give $e^{-\varepsilon}V_i\leq\widetilde V_i\leq V_i.$ 
Consequently, for every allocation its actual utility
$\widetilde u_i$ satisfies $e^{-\varepsilon}u_i\leq\widetilde u_i\leq u_i.$

Indeed, only one nonnegative summand of the additive utility has been
reduced. For allocations with positive utilities this implies $ \sum_i\log u_i-n\varepsilon
 \leq\sum_i\log\widetilde u_i
 \leq\sum_i\log u_i.$
 
Define a positive rational Nash-welfare threshold by $t_*:=\frac{B}{nL}+\frac{\gamma}{2n}$,  and $\widetilde\tau:=L \cdot E_N(t_*).$ 
Here both $a_i/L$ and $t_*$ lie in $[0,1]$. The same logarithmic
error bound applied at \(t_*\) gives $ F_*-\frac{5\gamma}{8}
 \leq n\log\widetilde\tau
 \leq F_*-\frac\gamma2.$ 
If the source instance is a yes-instance, the allocation
corresponding to a balanced partition has log-product in the
rational instance at least $ \sum_i\log\widetilde u_i
 \geq F_*-n\varepsilon
 =F_*-\frac{\gamma}{8},$ 
which is strictly above
$n\log\widetilde\tau$. If the source instance is a no-instance,
every positive-welfare allocation has ideal log-product at most
$F_*-\gamma$. Its actual log-product is no larger, and is therefore
strictly below $n\log\widetilde\tau$. Zero-welfare allocations also
fall below the positive threshold. This proves the desired
equivalence using rational coefficients and a rational threshold.

\medskip
\noindent\textbf{Encoding length and exact optimization.}
All exponents in the Taylor polynomials are polynomially bounded in
the input bit length. Exact rational evaluation at a rational $t$
uses numerators and denominators with
$O(N\,\operatorname{bits}(t)+N\log N)$ bits; for example, a common
denominator is $\operatorname{den}(t)^N N!$. Thus the reduction
constructs all coefficients and the threshold in polynomial time
and with polynomial encoding length.

For a yes-instance, the rational perturbation leaves a strict gap
between partition allocations and every other allocation. Hence an
exact optimal allocation reveals a balanced partition by checking
$z_i=1$, the first-family set $I$, and $\sum_{i\in I}a_i=B$.
These are rational and integer checks and require no exact logarithm
or exponential evaluation. Multiplying all utilities and the
threshold by a common product of their denominators makes every
utility and the threshold a positive integer of polynomial bit
length, without changing the answer. This proves all assertions.
\end{proof}

\subsection{Three item types and egalitarian welfare}
\label{subsec:three-types-egalitarian-hardness}

For the egalitarian objective, a common price on the three item types
forces every feasible bundle to have a prescribed price. A small
rational perturbation leaves exactly two threshold-feasible bundles
at that price for each agent.

\begin{theorem}[Egalitarian-welfare hardness with three item types]
\label{thm:three-types-egalitarian-hardness}
Deciding whether an allocation of egalitarian welfare at least a specified
positive rational threshold exists is $\mathsf{NP}$-hard with three item types and
explicitly listed agents. The item multiplicities, utilities, and threshold
can all be required to be positive integers of polynomial binary encoding length.
Computing an exact optimal allocation is $\mathsf{NP}$-hard under the same
restriction. The utility matrix may be required to have rank at most
two, and all agents may be required to have the same strict ranking
of the three item types.
\end{theorem}

\begin{proof}
We reduce from \textsc{Balanced-Equal-Cardinality-Partition}, which
is $\mathsf{NP}$-hard by Lemma~\ref{lem:becp}. Let
$a_1,\ldots,a_n$ be positive integers, where $n$ and
$A:=\sum_i a_i$ are even. 
Set
\[
 s:=n/2,\qquad B:=A/2,\qquad M:=3A+1,\qquad
 \varepsilon:=\frac1{M^2},\qquad
 h_i:=M(M+1-a_i).
\]
There is one agent for each source integer and there are three item
types. Their multiplicities and the value of one item to agent $i$ are
\[
\begin{array}{c|c|c}
 \text{Item type}&\text{Multiplicity}&\text{Value to agent }i\\ \hline
 1&s(M+1)-B&M/h_i\\[2pt]
 2&sM-3B&(M+1+\varepsilon a_i)/h_i\\[2pt]
 3&B&\bigl(2M+3-\varepsilon(M-3a_i)\bigr)/h_i.
\end{array}
\]
The welfare threshold is $1$. All multiplicities are positive
integers: $s\geq1$, $B=A/2>0$, $M+1>B$, and $M>3B$.
All values are positive rationals: $h_i>0$ and
$0<M-3a_i<M$, so the third numerator is greater than
$2M+3-1/M>0$. Moreover, for every agent the three item values
are strictly increasing in the displayed type order, since $M<M+1+\varepsilon a_i
 <2M+3-\varepsilon(M-3a_i).$ 
Thus the construction has exactly three distinct item types.
If $v_1,v_2,v_3$ denote its utility-column vectors, then $ v_3=3v_2-(1+\varepsilon)v_1,$ 
so the utility matrix has rank at most two.

\medskip
\noindent\textbf{A common resource price forces equality.}
For a bundle $(x,y,z)\in\mathbb Z_{\geq0}^3$, define its price $ C(x,y,z):=Mx+(M+1)y+(2M+3)z.$ 
This price is independent of the agent. Let $(x_i,y_i,z_i)$ be the bundle assigned to agent $i$. Agent $i$ obtains utility at least $1$ precisely when
\begin{equation}
\label{eq:three-esw-utility-test}
 C(x_i,y_i,z_i)+\varepsilon\bigl(a_i y_i-(M-3a_i)z_i\bigr)\geq h_i.
\end{equation}
Every bundle satisfying this inequality has price at least $h_i$.
Otherwise, since $C(x_i,y_i,z_i)$ and $h_i$ are integers,
\(
 C(x_i,y_i,z_i)\leq h_i-1.
\)
Moreover, $(M+1)y_i\leq C(x_i)<h_i\leq M(M+1),$ 
so $y_i<M$. Using $a_i\leq A$ and $M=3A+1$, we obtain $\varepsilon a_i y_i\leq\frac{A}{M}<1.$  
Since $M-3a_i>0$ and $z_i\geq0$, it follows that $C(x_i,y_i,z_i)
  +\varepsilon\bigl(a_i y_i-(M-3a_i)z_i\bigr)
 \leq h_i-1+\varepsilon a_i y_i
 <h_i,$ 
contradicting~\eqref{eq:three-esw-utility-test}. Hence $C(x_i,y_i,z_i)\geq h_i$.

The total price of all available items is
\begin{align*}
 M\bigl(s(M+1)-B\bigr)
 +(M+1)(sM-3B)+(2M+3)B
 =nM(M+1)-MA
 ={\sum}_{i=1}^n h_i.
\end{align*}
Feasibility therefore implies $\sum_{i=1}^n h_i
 \leq\sum_{i=1}^n C(x_i,y_i,z_i)
 \leq\sum_{i=1}^n h_i.$ 
Thus $\sum_{i=1}^n\bigl(C(x_i,y_i,z_i)-h_i\bigr)=0.$ 
Each summand is nonnegative, so every summand is zero.
Consequently, $C(x_i,y_i,z_i)=h_i\qquad\text{for every }i.$ 
Furthermore, every item must be assigned: leaving any item
unassigned would make the upper bound on the total assigned
price strict, since all three item prices are positive.

\medskip
\noindent\textbf{Every agent has exactly two feasible bundles at its
forced price.}
Fix agent $i$ and an integer bundle with $C(x_i,y_i,z_i)=h_i$. Reducing
this equality modulo $M$ shows that $y_i+3z_i=kM$ 
for some integer $k\geq0$. Substitution gives
\begin{equation}\label{eq:xy}
 y_i=kM-3z_i,
 \qquad
 x_i=(M+1)(1-k)-a_i+z_i.
\end{equation}
If $k\geq2$, nonnegativity of $x_i$ and $y_i$ would imply $(M+1)(k-1)+a_i\leq z_i\leq\frac{kM}{3}.$ 
This is impossible, because the left endpoint exceeds the right
endpoint by $\frac{(2k-3)M}{3}+(k-1)+a_i>0.$ 
Thus $k$ is either zero or one. When $k=0$, nonnegativity forces
$y_i=z_i=0$, giving the bundle $P_i:=(M+1-a_i,0,0).$
When $k=1$,~\eqref{eq:xy} and nonnegativity of $x_i$ and $y_i$ give $(x_i,y_i,z_i)=(z_i-a_i,M-3z_i,z_i)$, where $a_i\leq z_i\leq\lfloor M/3\rfloor.$ 
On such a bundle, the utility of agent $i$ is
\[
 u_i=1+\frac{\varepsilon}{h_i}
       \bigl(a_i y_i-(M-3a_i)z_i\bigr)=1+\frac{\varepsilon}{h_i}
       \bigl(M(a_i-z_i)\bigr).
\]
where we use $C(x_i,y_i,z_i)=h_i$.
As $z_i\geq a_i$, this utility is at most $1$, and it is at least $1$ if and only if $z_i=a_i$. Thus, among the bundles in the case $k=1$, exactly one meets the utility threshold, namely, $ Q_i:=(0,M-3a_i,a_i).$ 
Both $P_i$ and $Q_i$ are nonnegative integer bundles and give agent
$i$ utility exactly $1$. Thus these are precisely its two feasible bundles at price $h_i$.

\medskip
\noindent\textbf{Recovering the balanced partition.}
In a threshold-feasible allocation, let $I$ be the agents receiving
$Q_i$; all remaining agents receive $P_i$. The constraint of type~3 items
gives $ \sum_{i\in I}a_i=B.$ 
The constraint of type~2 items then gives $\sum_{i\in I}(M-3a_i)=M|I|-3B=sM-3B,$ 
so $|I|=s$. Hence $I$ is a balanced partition.

Conversely, if $I$ is a balanced partition, give $Q_i$ to each agent
$i\in I$ and $P_i$ to every other agent. The type~2 and type~3
constraints are met by the two balanced-partition equalities. The type~1
usage is
\[
 {\sum}_{i\notin I}(M+1-a_i)
 =(n-s)(M+1)-(A-B)
 =s(M+1)-B,
\]
which also matches its stock. Every agent obtains utility exactly
$1$. This proves both directions of the reduction.

\medskip
\noindent\textbf{Encoding length and exact optimization.}
Every stock, numerator, and denominator has polynomial binary
encoding length. There are only $n$ explicitly listed agents, and the
construction uses rational arithmetic on integers with
$O(\log(n+1)+\log(A+1))$ bits. Multiplying all utility coefficients
and the threshold by the product of their denominators converts
them to positive integers. The encoding length of that common
multiplier is at most the sum of the individual denominator lengths,
so this conversion remains polynomial.

An exact optimal-allocation algorithm would decide the source
problem by testing whether every returned utility is at least the
threshold. These are exact rational comparisons of polynomial
encoding length. Hence exact allocation optimization is $\mathsf{NP}$-hard as
well.
\end{proof}

The reductions with a fixed number of item types establish ordinary
$\mathsf{NP}$-hardness; they do not establish strong
$\mathsf{NP}$-hardness. Their numerical gaps need not be
inverse-polynomial in the binary input length. Thus these results are
consistent with the FPTAS in Section~\ref{sec:fixed-item-types}.

\subsection{One agent type and an unrestricted number of item types}
\label{subsec:one-agent-type-hardness}

\begin{theorem}[Hardness with one agent type]
\label{thm:one-agent-type-hardness}
For every fixed
$p\in(\mathbb Q\cap(-\infty,1))\cup\{-\infty\}$, the decision problem of
determining whether $\OPT_p\geq\tau$ is strongly $\mathsf{NP}$-hard, even
when there is only one agent type. Consequently, exact maximization of
$p$-mean social welfare is strongly $\mathsf{NP}$-hard under the same
restriction. The reduction uses positive integer utilities; the number of
item types is not fixed.
\end{theorem}

\begin{proof}
We reduce from \textsc{3-Partition}. An instance consists of positive
integers $a_1,\ldots,a_{3r}$ and $B$ such that $\sum_{j=1}^{3r}a_j=rB$ and  $\frac B4<a_j<\frac B2$ for every $j\in[3r].$ The question is whether the numbers can be partitioned into $r$ triples,
each having sum $B$. Under these restrictions, \textsc{3-Partition} is
strongly $\mathsf{NP}$-complete.

Construct an instance with $r$ agents of a single agent type. For every
$j\in[3r]$, introduce one item of value $a_j$ to every agent. Equivalently,
one may initially regard each item as a type of multiplicity one and then
merge types having the same value. Thus there are at most $3r$ item types,
all utilities are positive integers, and the decision threshold is
$\tau:=B$.

Let $x$ be an allocation and put $v_i:=u_1(x_i)$ for every $i\in[r]$.
Because $x$ assigns every item exactly once, $\sum_{i=1}^r v_i=\sum_{j=1}^{3r}a_j=rB.$ 
Thus the arithmetic mean of $v_1,\ldots,v_r$ is $B$.

First let $p<1$ be finite. If $p\leq0$ and some $v_i=0$, then the
zero-utility convention gives $\PSW_p(x)=0<B$. Otherwise, the
generalized-mean inequality gives $\PSW_p(x)\leq B$. Since $p<1$, equality
holds if and only if $v_1=\cdots=v_r=B$. For $p=-\infty$,
$\PSW_{-\infty}(x)=\min_i v_i\leq B$; moreover, equality implies
$v_i\geq B$ for every $i$, and the displayed sum then forces $v_i=B$ for
every $i$. Hence, for every value of $p$ in the statement, $\PSW_p(x)\leq B$ with equality if and only if $
 v_i=B$ for every $i\in[r]$.

If the \textsc{3-Partition} instance is a yes-instance, assign the three
items corresponding to each triple to one agent. Every agent then has
utility $B$, so the resulting allocation meets the threshold.

Conversely, suppose that $\PSW_p(x)\geq B$. The preceding equivalence gives
$v_i=B$ for every $i\in[r]$. As $\frac B4<a_j<\frac B2$, each bundle contains exactly three items: a
bundle containing at most two items has value strictly less than $B$, whereas
a bundle containing at least four items has value strictly greater than $B$.
Because every item is allocated exactly once, the $r$ bundles yield a
partition into triples of sum $B$. Thus the original \textsc{3-Partition}
instance is a yes-instance.

The reduction is polynomial and copies the numerical values of the
\textsc{3-Partition} instance without increasing them. It therefore remains
a polynomial reduction when the source integers are polynomially bounded,
which proves strong $\mathsf{NP}$-hardness.

The claim for computing an exact optimal allocation does not require
an exact comparison of irrational welfare values. Given such an
allocation, check whether every positive-multiplicity bundle has integer
utility $B$. By the equivalence proved above, this happens if and only
if the source instance is a yes-instance, and the check takes polynomial
time for both explicit and compressed allocations.
\end{proof}

Strong hardness does not by itself rule out an FPTAS for a real-valued
objective. For the preceding reduction, an inverse-polynomial gap below
the optimum of a yes-instance gives the required approximation barrier.

\begin{corollary}[No FPTAS with one agent type]
\label{cor:no-fptas-one-agent-type}
For every fixed
$p\in(\mathbb Q\cap(-\infty,1))\cup\{-\infty\}$, maximizing $p$-mean
social welfare admits no FPTAS unless $\mathsf{P}=\mathsf{NP}$, even with
one agent type and positive integer utilities. Here the number of item
types is unrestricted.
\end{corollary}

\begin{proof}
Use the instances constructed in the proof of
Theorem~\ref{thm:one-agent-type-hardness}, restricting the source
\textsc{3-Partition} instances to the strongly $\mathsf{NP}$-hard family
in which $B$ is polynomially bounded in the source input length. There
are $r$ identical agents and $3r$ items of total value $rB$. For every
allocation $x$, the utilities $v_i:=u_1(x_i)$ are nonnegative integers
satisfying $\sum_{i=1}^r v_i=rB$. The source instance is a yes-instance
if and only if an allocation with $v_i=B$ for every $i$ exists.

First fix a finite rational $p<1$ and define the positive rational
constants and the rational gap parameter
\begin{equation}
 \alpha:=\frac{1+\max\{p,0\}}2,
 \qquad
 c:=\alpha(1-\alpha),
 \qquad
 \delta:=\frac{c}{r^3B^2}.
 \label{eq:no-fptas-gap-parameter}
\end{equation}
Thus $p<\alpha<1$ and $\alpha\geq1/2$. Normalize the utilities by
setting $z_i:=v_i/B$. Then $0\leq z_i\leq r$ and
$\sum_{i=1}^r z_i=r$. If the utilities are not all equal to $B$, the
integer deviations $v_i-B$ sum to zero and include both a positive
integer and a negative integer. Consequently,
\begin{equation}
 \sum_{i=1}^r(z_i-1)^2
 =\frac1{B^2}\sum_{i=1}^r(v_i-B)^2
 \geq\frac2{B^2}.
 \label{eq:no-fptas-integer-deviation}
\end{equation}

The function $f(t):=t^\alpha$ satisfies, for $0<t\leq r$, $f''(t)=-c\,t^{\alpha-2}\leq-\frac{c}{r^2},$ 
where we use $r\geq1$. By continuity, the inequality also holds at $t=0$. Taylor's theorem around $1$ gives
\begin{equation}
 t^\alpha
 \leq 1+\alpha(t-1)-\frac{c}{2r^2}(t-1)^2
 \qquad(0\leq t\leq r).
 \label{eq:no-fptas-concavity}
\end{equation}
Averaging this inequality at $z_1,\ldots,z_r$, the linear terms cancel.
For every allocation whose utilities are not all $B$,
\eqref{eq:no-fptas-integer-deviation} therefore implies
\begin{equation}
 \frac1r\sum_{i=1}^r z_i^\alpha
 \leq 1-\frac{c}{2r^3}\sum_{i=1}^r(z_i-1)^2
 \leq 1-\delta.
 \label{eq:no-fptas-power-gap}
\end{equation}
By monotonicity of generalized means in their order and by homogeneity,
\begin{equation}
 \PSW_p(x)
 \leq \PSW_\alpha(x)
 =B\left(\frac1r\sum_{i=1}^r z_i^\alpha\right)^{1/\alpha}
 \leq B(1-\delta)^{1/\alpha}
 \leq B(1-\delta).
 \label{eq:no-fptas-welfare-gap}
\end{equation}
The first inequality also holds under the zero-utility convention when
$p\leq0$ and some $v_i=0$. The last inequality uses
$0<\delta<1$ and $1/\alpha>1$.

Suppose an FPTAS exists for this fixed $p$. Run it on the constructed
instance with the rational accuracy parameter $\varepsilon:=\frac\delta2
 =\frac{c}{2r^3B^2}.$ 
On a yes-instance, $\OPT_p=B$, so the returned allocation $x$ satisfies $\PSW_p(x)\geq(1-\delta/2)B>(1-\delta)B.$ 
By \eqref{eq:no-fptas-welfare-gap}, its utilities must all equal $B$.
On a no-instance, no such allocation exists. Thus one can decide the
source instance by checking whether every utility in the returned
allocation equals the integer $B$. This check also applies to a
compressed allocation by checking every bundle with positive
multiplicity, and requires no evaluation or comparison of irrational
welfare values. Since $p$ is fixed, $c$ is a fixed positive rational,
and $\frac1\varepsilon=\frac{2r^3B^2}{c}$ 
is polynomially bounded on the chosen source family. The construction,
the accuracy encoding, the FPTAS running time, and the final integer
checks are therefore polynomial in the source input length. This would
imply $\mathsf{P}=\mathsf{NP}$.

For $p=-\infty$, use the same reduction and set
$\varepsilon:=1/(2B)$. Every allocation whose utilities are not all
$B$ has some integer utility at most $B-1$, and hence has egalitarian
welfare at most $B-1$. On a yes-instance an FPTAS would return welfare
at least $(1-1/(2B))B=B-1/2$, forcing every utility to equal $B$.
Checking the integer utilities again decides \textsc{3-Partition} in
polynomial time, since $1/\varepsilon=2B$ is polynomially bounded.
\end{proof}

The restriction in Corollary~\ref{cor:no-fptas-one-agent-type} fixes the
number $K$ of agent types, not the number $n$ of agents: the reduction
has $K=1$ but $n=r$ grows with the source instance. It also allows the
number $q$ of item types to grow. Thus the corollary does not preclude
an FPTAS when $q$ is fixed, nor does it establish the same obstruction
for a fixed number of agents.

\section{An FPTAS for $p$-Mean Social Welfare with a Fixed Number of Item Types}
\label{sec:fixed-item-types}

In this section, the number $q$ of item types is fixed, whereas the number
$K$ of agent types is unrestricted. Both the item multiplicities $m_j$
and the agent multiplicities $n_k$ are encoded in binary. In particular,
the total number $n$ of agents can be exponential in the input length.
We use the notation and welfare definitions introduced in Section~\ref{sec:notation}.

\begin{theorem}
\label{thm:fixed-item-types-ptas}
For every fixed number $q$ of item types and every fixed
$p\in\mathbb Q\cup\{-\infty\}$, maximizing $p$-mean social welfare under
additive nonnegative rational utilities admits an FPTAS in the fully
compact model. More precisely, for every rational
$\varepsilon\in(0,1)$, the algorithm returns a compressed allocation $x$
such that $\PSW_p(x)\geq(1-\varepsilon)\OPT_p$ in time polynomial in the
compact input length, the supplied encoding length of the accuracy
parameter, and $1/\varepsilon$. The degree of the polynomial depends
only on $q$ and $p$, and is independent of $K$, $n$, and $\varepsilon$.
\end{theorem}

\subsection{Algorithmic overview}
\label{subsec:fixed-item-overview}

The algorithm works with groups of identical agents throughout. We first
show that an optimum can be chosen so that agents of each type receive
at most $2^q$ distinct bundles. This statement concerns the existence of
a suitable optimum; it does not require us to compute its bundles or
their multiplicities. We next partition each type into a polynomial
number of blocks, each containing a power-of-two number of agents.
The partition depends only on $n_k$ and can represent every distribution
of $n_k$ agents among at most $2^q$ bundles. Consequently, after rounding
bundles to a geometric grid, a near-optimal allocation can be represented
by assigning one configuration to each block.

The main difficulty is that even one block may contain exponentially
many agents. We address it by distinguishing the number of items received
by an individual agent from the aggregate consumption of a block.
For each item coordinate, we guess an anchor block with maximum aggregate
consumption among blocks whose individual coordinate is large. A null
anchor records the case in which no individual coordinate is large.
The anchor configurations are fixed, and all other configurations are
required to respect the guessed bounds on large coordinates. Small
individual coordinates need not satisfy an aggregate anchor bound.

We round aggregate resource consumptions down to suitable powers of two
and process the remaining blocks in nonincreasing order of size. The
dynamic program records exact residual capacities only while they could
restrict the remaining blocks. Once a residual capacity is at least an
upper bound on the consumption of every remaining suffix assignment,
the coordinate is replaced by the symbol $\mathsf{FREE}$. This certificate
remains valid at later layers. Divisibility of the power-of-two block
sizes bounds the number of residual capacities below the suffix bound,
including in coordinates with a null anchor.

Resource rounding may introduce a small excess. We remove this excess
from the relevant anchor block, distributing the deletions as evenly as
possible among its agents. An anchor is used only for a large individual
coordinate, so every agent retains a multiplicative fraction of each
repaired coordinate and hence of her utility. In particular, the repair
preserves positive utility when $p\leq0$. The repair, completion, and
output are performed on compressed groups, without expanding any block.
Algorithm~\ref{alg:fixed-item-compact-fptas} summarizes the construction.

\begin{algorithm}[t]
\caption{Compact FPTAS for a fixed number of item types}
\label{alg:fixed-item-compact-fptas}
\begin{algorithmic}[1]
\Require Agent types $(n_k,(u_{k,j})_{j\in[q]})_{k\in[K]}$, item
multiplicities $(m_j)_{j\in[q]}$, fixed $p$, and accuracy $\varepsilon$
\State Delete empty types; handle the instance with no items; choose
$\eta$ as in \eqref{eq:fixed-item-eta}
\If{$p\leq0$ and the positive-utility flow test fails}
 \State \Return a compressed allocation assigning all goods to one agent
\EndIf
\State Construct the universal dyadic blocks, the configuration set
$\mathcal C$, and the rational objective coefficients
\State Set $V^{\mathrm{best}}\gets-\infty$
\For{each guess of a null anchor or an anchored block and configuration
for every coordinate}
 \State Check consistency, original capacities, and all fixed-block
coordinate bounds; discard a guess that fails a check
 \State Construct residual admissible sets; discard a guess with an empty set
 \State Choose the powers of two $\Delta_j$, order the residual blocks,
and compute the suffix bounds $M_{t,j}$
 \State Run Algorithm~\ref{alg:fixed-item-compact-dp} for this guess
 \If{a profile is returned with label $V>V^{\mathrm{best}}$}
  \State Store this profile and its guess; set $V^{\mathrm{best}}\gets V$
 \EndIf
\EndFor
\State Repair each excess by balanced deletions from its anchor block
\State Give all leftover items to one agent, and merge identical bundles
within each type
\State \Return the resulting compressed allocation
\end{algorithmic}
\end{algorithm}

\subsection{Sparse allocations and universal blocks}
\label{subsec:fixed-item-blocks}

The two structural lemmas below allow us to replace the agents by a
polynomial number of blocks while retaining a suitable optimum.

\begin{lemma}
\label{lem:fixed-item-sparse-support}
If $q\geq1$, there is an optimal allocation in which each agent type
receives at most $s:=2^q$ distinct bundle vectors.
\end{lemma}

\begin{proof}
First suppose that $p<1$ is finite and that, when $p\leq0$, an allocation
with positive welfare exists. Among all optimal allocations choose one
minimizing $\sum_{i=1}^n\|x_i\|_2^2$. Such a choice exists because the
set of allocations is finite. Suppose that two distinct bundles $x,y$
assigned to agents of the same type have the same parity in every
coordinate. Replace these two bundles by $(x+y)/2$ and $(x+y)/2$.
Both new bundles are integral, and their total consumption is unchanged.
Because the agents have the same additive utility function, their new
utilities both equal the average of their previous utilities.

For $p<1$, the function $\phi_p$ is concave on its domain. Thus the sum
of the two transformed utilities does not decrease. For $p\leq0$, the
two original utilities are positive, so the average remains in the
domain of $\phi_p$. The new allocation therefore remains optimal.
On the other hand, $ \|x\|_2^2+\|y\|_2^2
 -2\left\|\frac{x+y}{2}\right\|_2^2
 =\frac12\|x-y\|_2^2>0,$ 
contradicting the choice of the secondary objective. There is consequently
at most one distinct bundle in each of the $2^q$ parity classes.

For $p=-\infty$ and positive optimum, use the same secondary objective.
Averaging two utilities cannot decrease their minimum, so the same
argument applies. If $p\leq0$ and the optimum is zero, every allocation
is optimal. Assigning all goods to one agent and giving empty bundles
to all other agents uses at most two distinct bundles within any type.

Finally, suppose that $p\geq1$. Starting with an optimum, combine all
bundles assigned to a given type into the bundle of one agent of that
type, and give its other agents empty bundles. If the original utilities
within this type are $v_1,\ldots,v_{n_k}$, its contribution to $Q_p$
changes from $\sum_i v_i^p$ to $(\sum_i v_i)^p$, which is no smaller.
Doing this for every type preserves optimality and leaves at most two
distinct bundles per type. Since $q\geq1$, we have $2\leq2^q$.
\end{proof}

A histogram on at most $s$ categories is a vector
$(h_1,\ldots,h_t)\in\mathbb Z_{\geq0}^t$, where $t\leq s$, whose $j$th coordinate records the number of units in category $j$.
Its total multiplicity is $\sum_{j=1}^t h_j$.
We may append zero coordinates to regard it as a vector of length $s$. Here, the categories will represent distinct bundle vectors assigned
to agents of the same type, and each coordinate will record the number of agents receiving the corresponding bundle.

\begin{lemma}
\label{lem:fixed-item-universal-partition}
For every positive integer $N$ and integer $s\geq2$, one can compute a
partition of $N$ into $O(s(1+\log N))$ powers of two with the following
property: every histogram on at most $s$ categories with total
multiplicity $N$ can be represented by assigning one category to every
part of the partition. The partition can be computed in time polynomial
in $s$ and the binary encoding length of $N$.
\end{lemma}

\begin{proof}
Set $N_0:=N$. At level $\ell\geq0$, if $N_\ell<s$, set
$a_\ell:=N_\ell$, create $a_\ell$ parts of size $2^\ell$, and terminate.
Otherwise, choose the unique $a_\ell\in\{s-1,s\}$ satisfying
$a_\ell\equiv N_\ell\pmod2$, create $a_\ell$ parts of size $2^\ell$,
and set $N_{\ell+1}:=\frac{N_\ell-a_\ell}{2}.$ 
Since $a_\ell\leq N_\ell$ and
$a_\ell\equiv N_\ell\pmod2$, the remaining total
$N_{\ell+1}=(N_\ell-a_\ell)/2$ is a nonnegative integer. Rearranging the recurrence and multiplying by $2^\ell$ gives $ a_\ell2^\ell
 =2^\ell N_\ell-2^{\ell+1}N_{\ell+1}.$ 
Let $L$ be the terminal level, where $a_L=N_L$. Summing over
the nonterminal levels and adding the terminal contribution yields $ \sum_{\ell=0}^{L}a_\ell2^\ell
 =N_0-2^LN_L+2^La_L
 =N_0=N.$ 
Thus the constructed parts have total size exactly $N$. At every nonterminal level
$N_{\ell+1}\leq N_\ell/2$, so there are at most
$1+\lfloor\log_2N\rfloor$ levels containing a positive residual.
Each level contributes at most $s$ parts. The recurrence and the output
use binary integers of $O(1+\log N)$ bits.

We now show that the constructed partition can represent every histogram with at most $s$ categories and total $N$. Pad the given histogram with zero entries so that
it has exactly $s$ coordinates. At a nonterminal level, let $o$ be the
number of odd-valued entries of the current histogram, whose total is
$N_\ell$. Then $o\equiv N_\ell\equiv a_\ell\pmod2$ and
$o\leq a_\ell$. Indeed, $o\leq s$, and when $a_\ell=s-1$ the value
$o=s$ is excluded by parity. Remove one unit from each odd entry.
The remaining entries are even. Next, remove another $a_\ell-o$ units in pairs, each pair
consisting of two units from the same category. Equivalently,
perform $(a_\ell-o)/2$ operations, each subtracting $2$ from
an entry that is at least $2$. All entries remain nonnegative
and even.
 There are enough pairs because $\frac{N_\ell-o}{2}\geq\frac{a_\ell-o}{2}.$ 
Assign the categories of these $a_\ell$ removed units to the parts
created at this level. Halve the remaining histogram. It is again an
integral histogram, now of total $N_{\ell+1}$, to which the same argument
applies. At a terminal level assign one part to each remaining unit.
Restoring the factors $2^\ell$ shows that the parts assigned to each
category sum to its original multiplicity.
\end{proof}
A numerical example illustrating the proof of Lemma~\ref{lem:fixed-item-universal-partition} is given in appendix~\ref{sec:example-fixed-item-universal-partition}

Apply Lemma~\ref{lem:fixed-item-universal-partition} to each
nonempty agent type $k$, with $N=n_k$ and $s=2^q$.
Each resulting part of size $b$ represents a \emph{block} of
$b$ agents of type $k$. Let $H$ be the total number of resulting blocks.
Block $h\in[H]$ has type $k(h)$ and size $b_h$, where $b_h$ is a power
of two. These block identifiers are fixed throughout the algorithm.
Writing $n_{\max}:=\max_k n_k$, we have
\begin{equation}
\label{eq:fixed-item-block-count}
 H=O_q\left(\sum_{k=1}^K\log(1+n_k)\right)
  =O_q\bigl(K\log(1+n_{\max})\bigr).
\end{equation}
In particular, $H$ is polynomial in the compact input length. Moreover, $\sum_{h:k(h)=k}b_h=n_k$  for every $k$, and $\sum_{h=1}^H b_h=n.$ 
All agents in a block will receive
the same bundle: assigning a bundle $x$ to a block $h$ means giving $x$ to each of its $b_h$ agents, consuming $b_hx$ items
in total. The block sizes are fixed by the lemma; the bundle
assigned to each block will be chosen by the algorithm. 
Individual agents within a block are used only for
the analysis and the final compressed repair.

\subsection{Configurations and objective coefficients}
\label{subsec:fixed-item-configurations}

We now give the remaining construction and prove
Theorem~\ref{thm:fixed-item-types-ptas}. Throughout, $p$ and $q$ are fixed.
Delete agent types and item types of multiplicity zero, and relabel the
retained types. If there are no items, return the empty bundle with
multiplicity $n_k$ for each retained agent type. Henceforth, $q\geq1$,
$m_j\geq1$, and $n_k\geq1$. Choose the rational internal accuracy
\begin{equation}
\label{eq:fixed-item-eta}
 \eta:=
 \begin{cases}
  \displaystyle\frac{\varepsilon}{3+1/|p|},
      &p\in\mathbb Q\setminus\{0\},\\[2mm]
  \varepsilon/4,&p=0,\\
  \varepsilon/3,&p=-\infty.
 \end{cases}
\end{equation}
Thus $0<\eta<1/3$ and $\eta=\Theta_p(\varepsilon)$. By the accuracy
convention of Section~\ref{sec:notation}, $\eta$ has encoding length
$O_p(1+\log(1/\varepsilon))$.

\medskip
\noindent
\textbf{Positive-utility feasibility for $p\leq0$.}
For $p\leq0$, construct a flow network with a source, an item-type vertex
for each $j\in[q]$, an agent-type vertex for each $k\in[K]$, and a sink.
The source-to-item arc has capacity $m_j$. The arc from item type $j$
to agent type $k$ is present if and only if $u_{k,j}>0$, and has capacity
$n_k$. The arc from agent type $k$ to the sink has capacity $n_k$.
All capacities are integral and encoded in binary.

An integral flow of value $n$ supplies each type $k$ with $n_k$ positively
valued items, which can be distributed one per agent of that type.
Conversely, from an allocation with all utilities positive, choose one
positively valued item from each agent's bundle. Counting these chosen
items by item type and agent type gives an integral flow of value $n$.
Thus positive utility is feasible exactly when the maximum flow is $n$.
This test uses $K+q+2$ vertices and at most $Kq+K+q$ arcs and is
polynomial in the compact input length.

If the test fails, the optimum is zero under the convention for $p\leq0$.
Return all items in the bundle of one agent, with empty bundles for the
remaining agents, using binary-encoded group multiplicities. We may therefore
assume below that positive welfare is feasible whenever $p\leq0$.

\medskip
\noindent
\textbf{Geometric configurations.}
For every item type $j$, let
\[
 \Gamma_j:=\{0\}\cup
 \left(
  \left\{\left\lceil(1+\eta)^\ell\right\rceil:
    \ell=0,\ldots,\left\lceil\log_{1+\eta}m_j\right\rceil\right\}
  \cap\{1,\ldots,m_j\}
 \right).
\]
Duplicate values are removed. The intersection removes values exceeding
the available multiplicity. A configuration is a vector
$C=(c_1(C),\ldots,c_q(C))\in\mathcal C:=\prod_{j=1}^q\Gamma_j$.
Put $S:=|\mathcal C|$ and $\ell_m:=1+\log(1+m_{\max})$. Since
$\log(1+\eta)\geq\eta/2$, $ |\Gamma_j|\leq2+\left\lceil\log_{1+\eta}m_j\right\rceil$, and 
 $S=O_q\bigl((\ell_m/\eta)^q\bigr).$
The algorithm generates $\Gamma_j$ exactly by successive rational
multiplication by $1+\eta$, taking ceilings, and stopping when the ceiling
exceeds $m_j$. No evaluation of a logarithm is needed.

For an integer $a\in\{0,\ldots,m_j\}$, let $r_j(a)$ be the largest
grid member not exceeding $a$. Then $\frac{a}{1+\eta}\leq r_j(a)\leq a.$ 
 
Indeed, the assertion is immediate for $a=0$. For $a\geq1$, take
$\ell=\lfloor\log_{1+\eta}a\rfloor$. Integrality of $a$ implies
$\lceil(1+\eta)^\ell\rceil\leq a$, so this number belongs to $\Gamma_j$.
Moreover, $ a<(1+\eta)^{\ell+1}
 \leq(1+\eta)\left\lceil(1+\eta)^\ell\right\rceil
 \leq(1+\eta)r_j(a).$ 

Choose an optimum as in Lemma~\ref{lem:fixed-item-sparse-support} and
round each bundle down coordinatewise to the grids. Feasibility is
preserved, and additivity gives $u_i\bigl((r_j(x_{i,j}))_{j\in[q]}\bigr)
 \geq\frac{u_i(x_i)}{1+\eta}$ for every agent $i$. 
Each type still receives at most $s$ distinct configurations. For each agent type, the configuration multiplicities (i.e., number of agents receiving each
distinct configuration) form a
histogram with at most $s=2^q$ categories. By
Lemma~\ref{lem:fixed-item-universal-partition}, we can assign
one configuration to each previously constructed block so that,
for every configuration, the sizes of the blocks assigned to it
sum to its required multiplicity. Denote the resulting block profile by
$C^*=(C_h^*)_{h\in[H]}$. By monotonicity and homogeneity of the $p$-mean,
\begin{equation}
\label{eq:fixed-item-comparator}
 \sum_{h=1}^H b_hc_j(C_h^*)\leq m_j
 \quad(j\in[q]),
 \qquad
 \PSW_p(C^*)\geq\frac{\OPT_p}{1+\eta}.
\end{equation}
Here and below, welfare of a block profile refers to its conceptual
expansion to individual agents; this expansion is never computed.
The profile $C^*$ is a comparator for the analysis. The algorithm
neither computes it nor requires it to be optimal over all block
configuration profiles. For $p\leq0$, all its utilities are positive.

\medskip
\noindent
\textbf{Rational objective coefficients.}
For finite $p$, apply
Lemma~\ref{lem:directed-rational-approximation} with accuracy $\eta$ to
every type--configuration pair on the domain of $\phi_p$. Write
$\widetilde\phi_{k,p}(C):=\widetilde\phi_{p,\eta}(u_k(C))$. Thus
\begin{align}
 \frac{u_k(C)^p}{1+\eta}
 &\leq\widetilde\phi_{k,p}(C)\leq u_k(C)^p,
 &&p>0,\label{eq:fixed-item-coeff-positive}\\
  \log u_k(C)-\log(1+\eta)
 &\leq\widetilde\phi_{k,0}(C)\leq\log u_k(C),
&&p=0,\label{eq:fixed-item-coeff-logarithmic}
\\
 u_k(C)^p
 &\leq-\widetilde\phi_{k,p}(C)\leq(1+\eta)u_k(C)^p,
 &&-\infty<p<0.\label{eq:fixed-item-coeff-negative}
\end{align}
At zero utility and $p>0$, set the coefficient to zero. For $p\leq0$,
exclude zero-utility configurations whenever they are considered for a
block of type $k$. The finite-$p$ score of block $h$ with configuration
$C$ is $b_h\widetilde\phi_{k(h),p}(C)$. For $p=-\infty$, the block's
contribution to the minimum is the exact rational utility $u_{k(h)}(C)$,
independently of $b_h$. There are at most $KS$ coefficients to compute.

\subsection{Anchor guesses and the compressed dynamic program}
\label{subsec:fixed-item-dynamic-program}

\medskip
\noindent
\textbf{Anchors for large individual coordinates.}
Set $T:=\lceil1/\eta\rceil$. A coordinate $j$ of configuration $C$ is
called large if its individual count satisfies $c_j(C)\geq T$.
For each $j\in[q]$, guess either a null anchor $g_j=\bot$, or a block
$g_j\in[H]$ and its configuration $G_{g_j}\in\mathcal C$ satisfying
$c_j(G_{g_j})\geq T$. If the same block is guessed for more than one
coordinate, every occurrence must specify the same configuration.
Let $U\subseteq[H]$ be the set of distinct anchored blocks and put
\begin{equation}
\label{eq:fixed-item-anchor-load}
 L_j:=
 \begin{cases}
  b_{g_j}c_j(G_{g_j}),&g_j\neq\bot,\\
  0,&g_j=\bot.
 \end{cases}
\end{equation}
There are at most $(1+HS)^q$ guesses, and $|U|\leq q$.

Retain only guesses for which the fixed blocks satisfy $\sum_{h\in U}b_hc_j(G_h)\leq m_j$
for every $j\in[q],$ 
and, when $p\leq0$, have $u_{k(h)}(G_h)>0$ for every $h\in U$.
In addition, check all fixed blocks in all coordinates:
\begin{equation}
\label{eq:fixed-item-fixed-coordinate-check}
 c_j(G_h)\geq T\quad\Longrightarrow\quad b_hc_j(G_h)\leq L_j
 \qquad(h\in U,\ j\in[q]).
\end{equation}
This check applies even if block $h$ was guessed as an anchor for a
different coordinate. In particular, a null anchor for $j$ requires
every fixed block to have individual coordinate less than $T$ in $j$.

For each residual block $h\notin U$, define its admissible configurations
by
\begin{equation}
\label{eq:fixed-item-admissible}
 \mathcal A_h:=
 \left\{C\in\mathcal C:
  c_j(C)<T\ \text{or}\ b_hc_j(C)\leq L_j
  \text{ for every }j\in[q]\right\}.
\end{equation}
When $p\leq0$, also require $u_{k(h)}(C)>0$. Discard a guess if any
residual admissible set is empty. For a retained guess, the residual
capacities are
\begin{equation}
\label{eq:fixed-item-residual-capacities}
 a_j:=m_j-\sum_{h\in U}b_hc_j(G_h)\geq0.
\end{equation}

There is a correct guess relative to $C^*$. For each $j$, if any block
has $c_j(C_h^*)\geq T$, choose a block maximizing
$b_hc_j(C_h^*)$ among such blocks, and guess its configuration in $C^*$.
Otherwise choose $g_j=\bot$. The same block always has the same
configuration, so repeated anchors are consistent. Feasibility of $C^*$
implies that the fixed blocks fit the capacities. The maximizing choices
give \eqref{eq:fixed-item-fixed-coordinate-check} in every coordinate
and retain every residual configuration $C_h^*$ in $\mathcal A_h$.
Positivity also holds when required. Hence this guess passes all checks.

\medskip
\noindent
\textbf{Rounded aggregate resources.}
For a retained guess and every coordinate $j$, choose the integer power
of two
\begin{equation}
\label{eq:fixed-item-resource-step}
 \Delta_j:=
 \begin{cases}
  1,&L_j=0\ \text{or}\ \eta L_j/H<1,\\
  2^{\lfloor\log_2(\eta L_j/H)\rfloor},&\eta L_j/H\geq1.
 \end{cases}
\end{equation}
The power can be found by exact comparisons of binary integers and
rationals. Define the rounded consumption of a residual block by
\begin{equation}
\label{eq:fixed-item-rounded-consumption}
 d_j(h,C):=\Delta_j
 \left\lfloor\frac{b_hc_j(C)}{\Delta_j}\right\rfloor.
\end{equation}
If $\Delta_j=1$, this consumption is exact. If $\Delta_j>1$, then
$\Delta_j\leq\eta L_j/H$. In every case with $L_j>0$,
\begin{equation}
\label{eq:fixed-item-anchor-step-ratio}
 \frac{L_j}{\Delta_j}<\frac{2H}{\eta}.
\end{equation}
Indeed, when $\eta L_j/H<1$, we have $\Delta_j=1$ and
$L_j<H/\eta$. Otherwise $\Delta_j>(\eta L_j/H)/2$.
The fixed blocks retain their original, unrounded consumptions.

\medskip
\noindent
\textbf{Suffix bounds.}
Order the residual blocks as $h_1,\ldots,h_r$, where $r:=H-|U|$, so
that their sizes $\beta_t:=b_{h_t}$ satisfy
$\beta_1\geq\cdots\geq\beta_r$. Their original identifiers are
unchanged. Define the individual and suffix maxima by
\[
\begin{aligned}
 A_{t,j}:=\max_{C\in\mathcal A_{h_t}}d_j(h_t,C),
 &\qquad(t\in[r]),\\
 M_{t,j}:=\sum_{\ell=t+1}^r A_{\ell,j}
 &\qquad(0\leq t\leq r).
\end{aligned}
\]
The maxima exist because the retained admissible sets are nonempty.
The quantity $M_{t,j}$ bounds the rounded consumption in coordinate $j$
of every assignment to the blocks remaining after layer $t$. It is the
sum of separate maxima; no assertion of simultaneous attainment of the
maxima is needed. From \eqref{eq:fixed-item-admissible},
\begin{equation}
\label{eq:fixed-item-suffix-bounds}
 A_{t,j}\leq T\beta_t+L_j,
 \qquad
 M_{t,j}\leq HT\beta_{t+1}+HL_j\quad(0\leq t<r).
\end{equation}
For the first inequality, a small individual coordinate contributes
less than $T\beta_t$, whereas a large coordinate contributes at most
$L_j$; rounding only reduces the contribution. The second inequality
uses the ordering of the suffix block sizes and the fact that there
are at most $H$ blocks. At the terminal layer $M_{r,j}=0$ by the
empty-sum convention, and the second bound is needed only when $t<r$.
If $r=0$, all suffix bounds are zero.

\medskip
\noindent
\textbf{State compression by suffix certificates.}
At layer $t$, each resource-state coordinate is either a nonnegative
integer $z_j<M_{t,j}$ or the symbol $\mathsf{FREE}$. An integer records
the exact remaining rounded capacity. The symbol records that the
remaining capacity is at least $M_{t,j}$ and therefore suffices for
every admissible assignment to the entire remaining suffix.

A free coordinate can remain free at every later layer. To verify this,
suppose that the actual rounded residual $R_j$ before processing block
$h_t$ satisfies $R_j\geq M_{t-1,j}$. For every $C\in\mathcal A_{h_t}$,
\begin{equation}
\label{eq:fixed-item-free-certificate}
 R_j-d_j(h_t,C)
 \geq M_{t-1,j}-A_{t,j}=M_{t,j}.
\end{equation}
Thus the next residual is nonnegative and has the required suffix
certificate. The exact value of a free coordinate is unnecessary for
future feasibility decisions.

For an input state $z$ at layer $t-1$ and $C\in\mathcal A_{h_t}$,
define $\operatorname{next}_t(z,C)$ coordinatewise as follows. A free
coordinate remains free. From an integer coordinate, calculate
$y_j:=z_j-d_j(h_t,C)$. If any $y_j<0$, the transition is undefined.
Otherwise, record $y_j$ when $y_j<M_{t,j}$ and record $\mathsf{FREE}$
when $y_j\geq M_{t,j}$. Initialize $z^0_j$ to $a_j$ if
$a_j<M_{0,j}$, and to $\mathsf{FREE}$ otherwise.

\medskip
\noindent
\textbf{Objective labels and transitions.}
The initial label includes the fixed blocks:
\[
 \kappa_U:=
 \begin{cases}
  \displaystyle\sum_{h\in U}b_h\widetilde\phi_{k(h),p}(G_h),
      &p\in\mathbb Q,\\[2mm]
  \displaystyle\min_{h\in U}u_{k(h)}(G_h),&p=-\infty.
 \end{cases}
\]
For $U=\varnothing$, use $\kappa_U=0$ for finite $p$ and
$\kappa_U=+\infty$ for $p=-\infty$. Given a label $v$, define
\[
 \operatorname{upd}_p(v,h,C):=
 \begin{cases}
  v+b_h\widetilde\phi_{k(h),p}(C),&p\in\mathbb Q,\\
  \min\{v,u_{k(h)}(C)\},&p=-\infty.
 \end{cases}
\]
For every reachable state $z$ at layer $t$, the DP label $D_t(z)$ is
the maximum objective value of a partial block assignment with that
compressed resource state. Unreachable states are absent from the
stored table and have label $-\infty$ by convention. The algorithm
generates states from reachable predecessors and never enumerates all
integers up to an original capacity.

\begin{algorithm}[t]
\caption{Compressed configuration dynamic program for a retained guess}
\label{alg:fixed-item-compact-dp}
\begin{algorithmic}[1]
\Require Fixed configurations $(G_h)_{h\in U}$, ordered residual blocks
$h_1,\ldots,h_r$, admissible sets, rounded consumptions $d_j$, residual
capacities $a_j$, and suffix bounds $M_{t,j}$
\State Initialize $z^0_j\gets a_j$ if $a_j<M_{0,j}$ and
$z^0_j\gets\mathsf{FREE}$ otherwise
\State Set $D_0(z^0)\gets\kappa_U$, with every other state absent
\For{$t=1,\ldots,r$}
 \State Initialize an empty table $D_t$
 \For{each stored state $z$ of $D_{t-1}$}
  \For{each $C\in\mathcal A_{h_t}$}
   \State Compute $z'\gets\operatorname{next}_t(z,C)$
   \If{the transition is defined}
    \State Set $v'\gets\operatorname{upd}_p(D_{t-1}(z),h_t,C)$
    \If{$z'$ is absent from $D_t$, or $v'>D_t(z')$}
     \State Set $D_t(z')\gets v'$ and store predecessor $(z,C)$
    \EndIf
   \EndIf
  \EndFor
 \EndFor
\EndFor
\If{$D_r$ is empty}
 \State \Return no reachable terminal state
\EndIf
\State Choose a state $z^{\mathrm{term}}$ with maximum label in $D_r$
\State Recover the residual configurations by backtracking and combine
them with $(G_h)_{h\in U}$ on the original block identifiers
\State \Return $V_{\mathfrak g}:=D_r(z^{\mathrm{term}})$ and the
pre-repair block profile $\bar C^{\mathfrak g}$
\end{algorithmic}
\end{algorithm}

\begin{lemma}
\label{lem:fixed-item-dp-correctness}
For every retained guess, Algorithm~\ref{alg:fixed-item-compact-dp}
maximizes its objective over all residual admissible assignments
satisfying $\sum_{t=1}^r d_j(h_t,C_{h_t})\leq a_j$ for every $j$.
The assertion also holds when $r=0$.
\end{lemma}

\begin{proof}
For an exact state coordinate, the stored value is the actual remaining
rounded capacity. For a free coordinate, its suffix certificate follows
at initialization or at the transition where it first becomes free,
and remains valid by \eqref{eq:fixed-item-free-certificate}. Thus every
generated path is feasible for the rounded capacities. Conversely,
consider any rounded-feasible residual assignment. Since consumptions
are nonnegative, each of its prefixes is feasible. Exact coordinates
therefore never reject its next configuration, while free coordinates
accept every admissible suffix choice. It defines a path in the state
transition graph.

It remains to justify retaining only one objective label at each state.
Any two histories with the same compressed state have identical exact
residuals in its integer coordinates. In its free coordinates, both
histories permit every admissible remaining suffix. They therefore have
the same feasible future choices. Adding a fixed block score and taking
the minimum with a fixed utility are both nondecreasing functions of the
previous label. A history with a smaller label can never improve the
result of any common continuation. Induction over the layers proves the
DP invariant and the optimality of the largest terminal label.

At layer $r$, all $M_{r,j}=0$, so the only possible state is the
all-free state. When $r=0$, initialization is already at this layer,
and the empty loop returns exactly the fixed profile, which fits the
original capacities by the retained-guess checks.
\end{proof}

\begin{lemma}
\label{lem:fixed-item-state-count}
Every DP layer has at most $O_q((1+H^2/\eta)^q)$ reachable states.
\end{lemma}

\begin{proof}
Layer zero has one state, and the terminal layer has at most one state.
Consider a nonterminal layer $1\leq t<r$ and a coordinate $j$. Set
$w_{t,j}:=\max\{\beta_t,\Delta_j\}$. Every rounded consumption of an
already processed residual block is divisible by $w_{t,j}$. Indeed,
block sizes and $\Delta_j$ are powers of two. If
$\beta_t\geq\Delta_j$, each earlier block size is divisible by
$\beta_t$ and hence by $\Delta_j$, so its aggregate consumption is
rounded exactly and is divisible by $\beta_t=w_{t,j}$. If
$\beta_t<\Delta_j$, divisibility by $w_{t,j}=\Delta_j$ follows
directly from the definition of $d_j$.

Consequently, every integer state coordinate at layer $t$ satisfies $z_j\equiv a_j\pmod{w_{t,j}}$, and $0\leq z_j<M_{t,j}.$  
There are at most $1+M_{t,j}/w_{t,j}$ such integers. Including the free
symbol, using \eqref{eq:fixed-item-suffix-bounds}, and observing that
$\beta_{t+1}\leq\beta_t\leq w_{t,j}$ and $\Delta_j\leq w_{t,j}$,
the number of coordinate states is at most
\begin{align*}
 2+\frac{M_{t,j}}{w_{t,j}}
 \leq2+HT+\frac{HL_j}{\Delta_j}
 \leq2+HT+\frac{2H^2}{\eta}
 =O\left(1+\frac{H^2}{\eta}\right).
\end{align*}
When $L_j>0$, the second inequality uses
\eqref{eq:fixed-item-anchor-step-ratio}; when $L_j=0$, its corresponding
term is zero. Also $T=\lceil1/\eta\rceil=O(1/\eta)$.
Multiplying the coordinate bounds over the fixed number $q$ of item
types proves the claim. This counts only potential reachable residuals
in the indicated congruence classes, so storing and generating the
reachable states suffices.
\end{proof}

\medskip
\noindent
\textbf{Comparison with the near-optimal block profile.}
Under the correct guess, every residual $C_h^*$ is admissible and the
fixed configurations agree with $C^*$. Therefore,
\[
 \sum_{h\notin U}d_j(h,C_h^*)
 \leq\sum_{h\notin U}b_hc_j(C_h^*)
 \leq m_j-\sum_{h\in U}b_hc_j(C_h^*)=a_j.
\]
Lemma~\ref{lem:fixed-item-dp-correctness} shows that the comparator
defines a feasible DP path, including the empty-loop case $r=0$.
In particular, at least one guess returns a profile.

Run the DP for every retained guess and select a returned profile with
maximum label. Labels for different guesses are comparable because they
include all fixed and residual blocks. They are rational, except for
the temporary empty-set initialization $+\infty$ in the minimum case;
every returned label is finite because $H\geq1$. Let $\bar C$ denote
the selected block profile. For finite $p$,
\begin{equation}
\label{eq:fixed-item-dp-dominance}
 \sum_{h=1}^H b_h\widetilde\phi_{k(h),p}(\bar C_h)
 \geq\sum_{h=1}^H b_h\widetilde\phi_{k(h),p}(C_h^*).
\end{equation}
For $p=-\infty$,
\begin{equation}
\label{eq:fixed-item-dp-minimum}
 \min_{h\in[H]}u_{k(h)}(\bar C_h)
 \geq\min_{h\in[H]}u_{k(h)}(C_h^*).
\end{equation}
Every profile assigns exactly $n_k$ agents to type $k$, but its original
item consumption may exceed the capacities because of resource rounding.

\subsection{Repair, approximation guarantee, and running time}
\label{subsec:fixed-item-analysis}

\medskip
\noindent
\textbf{Repairing original capacities.}
All quantities in this paragraph refer to the selected guess. For every
coordinate, let
\[
 \nu_j:=\max\left\{0,\ \sum_{h=1}^H b_hc_j(\bar C_h)-m_j\right\}.
\]
If $\Delta_j=1$, rounded residual consumptions are exact; together with
the fixed blocks they fit $m_j$, so $\nu_j=0$. If $\Delta_j>1$, the
difference between original and rounded consumption of one residual
block is less than $\Delta_j$. Whenever $\nu_j>0$, it follows that
\begin{equation}
\label{eq:fixed-item-excess}
 0<\nu_j<r\Delta_j\leq H\Delta_j\leq\eta L_j.
\end{equation}
Thus a violated coordinate has a nonnull anchor $g_j$, with
$L_j=b_{g_j}c_j(G_{g_j})$ and $c_j(G_{g_j})\geq T$.
If $r=0$, no coordinate is violated.

Remove exactly $\nu_j$ items of type $j$ from the $b_{g_j}$ agents of
its anchor, distributing the removals as evenly as possible. Each such
agent loses at most
\begin{equation}
\label{eq:fixed-item-individual-deletion}
 \left\lceil\frac{\nu_j}{b_{g_j}}\right\rceil
 \leq\eta c_j(G_{g_j})+1
 \leq2\eta c_j(G_{g_j}),
\end{equation}
where the last inequality uses
$c_j(G_{g_j})\geq T\geq1/\eta$. Because $2\eta<1$, the removals
are possible and leave every affected coordinate nonnegative. If a
block anchors several coordinates, perform this operation separately
in each of them. For each agent, an unchanged coordinate retains its
full count, and every changed coordinate retains at least its
$(1-2\eta)$ fraction. Additivity and nonnegative utilities therefore give
\begin{equation}
\label{eq:fixed-item-repair-utility}
 u_i(x_i^{\mathrm{rep}})\geq(1-2\eta)u_i(\bar C_i)
 \qquad\text{for every individual agent }i.
\end{equation}
The notation $\bar C_i$ denotes that agent's configuration in the
conceptual expansion of $\bar C$. Exactly the excess is removed in
each violated coordinate, so $x^{\mathrm{rep}}$ is a feasible partial
allocation. Positive utilities remain positive. Monotonicity and
homogeneity imply
\begin{equation}
\label{eq:fixed-item-repair-welfare}
 \PSW_p(x^{\mathrm{rep}})\geq(1-2\eta)\PSW_p(\bar C).
\end{equation}

The balanced repair has a compressed implementation. In an anchor block
of size $b$, write each excess assigned to this block as
$\nu_j=b\alpha_j+\rho_j$, where $0\leq\rho_j<b$. Remove $\alpha_j$
items of coordinate $j$ from every agent and one more from the first
$\rho_j$ agents in a common ordering of that block. Sort the nonzero
remainder endpoints $\rho_j$. Agents between consecutive endpoints
receive identical repaired bundles. A block with $d$ repaired
coordinates splits into at most $d+1$ groups. Across all anchors, at
most $q$ coordinates are repaired, so the full repaired profile has at
most $H+q$ groups. This computation uses only binary integers; no
individual ordering is explicitly expanded.

Finally, collect all unassigned items into the bundle of one agent,
splitting that agent from her current group if necessary. This adds at
most one group and cannot decrease welfare. Merge groups of the same
type with the same bundle by adding their multiplicities. The resulting
allocation $x$ has at most $H+q+1$ groups, with exact item totals and
exact type multiplicities.

\medskip
\noindent
\textbf{Approximation guarantee.}
We verify all four cases of the welfare parameter. In weighted block
notation, the quantities defined in Section~\ref{sec:notation} satisfy
\[
 Q_p(C)=\sum_{h=1}^H b_hu_{k(h)}(C_h)^p
       \quad(p\in\mathbb Q\setminus\{0\}),
 \qquad
 \Lambda(C)=\sum_{h=1}^H b_h\log u_{k(h)}(C_h)
\]
whenever the displayed terms are defined.

For finite $p$, the score dominance in
\eqref{eq:fixed-item-dp-dominance} allows us to apply
Lemma~\ref{lem:weighted-rational-comparison} to the block profiles
$C^*$ and $\bar C$, with weights $b_h$ and total weight $n$.

Suppose first that $p>0$. The lemma gives $Q_p(\bar C)\geq\frac{Q_p(C^*)}{1+\eta}.$ 
Taking the positive $1/p$ power and using
\eqref{eq:fixed-item-comparator},
\eqref{eq:fixed-item-repair-welfare}, and completion gives
\begin{equation}
\label{eq:fixed-item-positive-guarantee}
 \PSW_p(x)
 \geq\frac{1-2\eta}{(1+\eta)^{1+1/p}}\OPT_p.
\end{equation}
This conclusion also holds when $\OPT_p=0$, without division by the
optimum or any requirement that the utilities be positive.

Suppose next that $-\infty<p<0$. The positive-utility test has passed,
and all configurations considered for a block have positive utility.
Lemma~\ref{lem:weighted-rational-comparison} gives $ Q_p(\bar C)\leq(1+\eta)Q_p(C^*).$ 
Since $1/p<0$, taking the $1/p$ power reverses the inequality. Therefore $ \PSW_p(\bar C)\geq(1+\eta)^{-1/|p|}\PSW_p(C^*),$ 
and the repair and geometric rounding give
\begin{equation}
\label{eq:fixed-item-negative-guarantee}
 \PSW_p(x)
 \geq\frac{1-2\eta}{(1+\eta)^{1+1/|p|}}\OPT_p.
\end{equation}

For $p=0$, the logarithmic conclusion of the same lemma gives $ \Lambda(\bar C)\geq\Lambda(C^*)-n\log(1+\eta).$ 
Divide by $n$ and exponentiate to obtain
$\PSW_0(\bar C)\geq\PSW_0(C^*)/(1+\eta)$. The binary block
multiplicities cause no additional error because their sum is exactly
$n$. Hence
\begin{equation}
\label{eq:fixed-item-zero-guarantee}
 \PSW_0(x)\geq\frac{1-2\eta}{(1+\eta)^2}\OPT_0.
\end{equation}

Finally, when $p=-\infty$, the DP uses exact utilities, and
\eqref{eq:fixed-item-dp-minimum} directly gives
$\PSW_{-\infty}(\bar C)\geq\PSW_{-\infty}(C^*)$. Thus
\begin{equation}
\label{eq:fixed-item-minimum-guarantee}
 \PSW_{-\infty}(x)\geq\frac{1-2\eta}{1+\eta}\OPT_{-\infty}.
\end{equation}

For every fixed $a>0$, convexity of $(1+\eta)^{-a}$ gives
$(1+\eta)^{-a}\geq1-a\eta$. Multiplying by $1-2\eta>0$ yields
\[
 (1-2\eta)(1+\eta)^{-a}
 \geq1-(a+2)\eta+2a\eta^2
 \geq1-(a+2)\eta.
\]
Applying this with $a=1+1/|p|$, $2$, and $1$, respectively, gives
\[
 \PSW_p(x)\geq
 \begin{cases}
  \bigl(1-(3+1/|p|)\eta\bigr)\OPT_p,
        &p\in\mathbb Q\setminus\{0\},\\
  (1-4\eta)\OPT_0,&p=0,\\
  (1-3\eta)\OPT_{-\infty},&p=-\infty.
 \end{cases}
\]
The choice \eqref{eq:fixed-item-eta} proves the required
$(1-\varepsilon)$ approximation in every case.

\medskip
\noindent
\textbf{Arithmetic and bit complexity.}
There are at most $(1+HS)^q$ anchor guesses. By
Lemma~\ref{lem:fixed-item-state-count}, a retained guess has at most
$H$ DP layers, $O_q((1+H^2/\eta)^q)$ states per layer, and at most
$S$ candidate configurations per state. Thus the number of DP
arithmetic operations, apart from polynomial table-management factors,
is bounded by
\begin{align}
\label{eq:fixed-item-operation-bound}
 O_q\left((1+HS)^qHS(1+H^2/\eta)^q\right)
 &=O_q\left(H^{3q+1}\ell_m^{q(q+1)}
              \eta^{-(q^2+2q)}\right).
\end{align}
Constructing admissible sets, checking fixed blocks, computing suffix
maxima, and recovering predecessors fit the same polynomial bounds.
Sorting residual blocks and implementing tables by exact comparisons
also take polynomial time. The bound on $H$ in
\eqref{eq:fixed-item-block-count} involves only the encoding lengths
of the agent multiplicities, rather than their numerical sum.

Let $L$ be the compact input encoding length, excluding the separately
supplied accuracy parameter. Every block size, item capacity, anchor
load, rounded consumption, and suffix maximum has polynomial encoding
length in $L+\log(1/\varepsilon)$. For example, a product $b_hc_j(C)$
has at most the sum of the encoding lengths of $n$ and $m_j$, and a
suffix sum adds at most $O(\log H)$ further bits. The same bound holds
for residual capacities and their differences. Exact grid generation
takes $O(\ell_m/\eta)$ rational multiplications; its intermediate
numerators and denominators have lengths $ O_p\left(\frac{\ell_m}{\eta}
                 (1+\log(1/\varepsilon))\right).$ 
Thus generation of the grid also has polynomial bit complexity, even
though no approximate logarithms are used.

Each rational configuration utility has encoding length polynomial in
$L$. Lemma~\ref{lem:directed-rational-approximation} therefore gives each of the
$KS$ objective coefficients encoding length polynomial in
$L+\log(1/\varepsilon)$. A finite-$p$ DP label is the sum of at most
$H$ terms, each a binary block size times one coefficient. Forming a
common denominator from these at most $H$ coefficients already bounds
the numerator and denominator lengths by a polynomial in $L$ and
$1/\varepsilon$. Comparisons between labels are consequently exact
rational comparisons of polynomial encoding length. For $p=-\infty$,
every finite label is one rational configuration utility. The symbols
used for unreachable states and the empty minimum are handled directly.

Repair and completion use binary division, remainders, and sorting of
at most $q$ endpoints per anchor. Each resulting bundle coordinate is
at most $m_j$, each output multiplicity is at most $n_k$, and there are
at most $H+q+1$ groups before identical bundles are merged. Hence both
the construction time and the output length are polynomial in the
compact input length and $1/\varepsilon$. Reading the originally
supplied accuracy encoding adds only its input length. Since
$\eta=\Theta_p(\varepsilon)$ and $q,p$ are fixed, the complete
algorithm is an FPTAS, proving
Theorem~\ref{thm:fixed-item-types-ptas}.

For comparison, when both numbers of types are fixed, exact egalitarian
welfare maximization follows directly from Theorem~12 of Goemans and
Rothvo{\ss}~\cite{GoemansR20}. We record this consequence for completeness;
the reduction to integer-cone feasibility, its bit complexity, and the
compressed reconstruction are given in
Appendix~\ref{sec:exact-egalitarian-appendix}.

\begin{proposition}[Consequence of Goemans--Rothvo{\ss}]
\label{prop:known-exact-egalitarian}
For fixed numbers $K$ of agent types and $q$ of item types, maximizing
egalitarian social welfare under additive nonnegative rational utilities
admits an exact polynomial-time algorithm when both agent and item
multiplicities are encoded in binary. The running time is polynomial in
the compact input length, and the algorithm returns a compressed
allocation.
\end{proposition}

\subsection{Exact optimal allocations with one item type}
\label{subsec:one-item-type-exact}

The case of one item type admits an exact allocation algorithm even when
the number of agent types is unrestricted.  The main ingredient is the
classical reduction of separable concave resource allocation to selection
in sorted implicit arrays~\cite{GalilM79}.  We give the details to
account for binary agent multiplicities and exact arithmetic.

\begin{theorem}[Optimal allocations with one item type]
\label{thm:one-item-type-exact}
Suppose that $q=1$, while the number $K$ of agent types is unrestricted.
For every fixed $p\in\mathbb Q\cup\{-\infty\}$, an allocation maximizing
$\PSW_p$ can be computed in time polynomial in the compact input length,
with both item and agent multiplicities encoded in binary.  The algorithm
returns a compressed allocation using at most two bundle sizes per agent
type.
\end{theorem}

\begin{proof}
Write $a_k:=u_{k,1}$ for the per-item value of type $k$, let $m$ be the
item multiplicity, and put $n=\sum_{k=1}^K n_k$.  Empty agent types can
be omitted.  If $m=0$, the assertion is immediate.  For $p\leq0$, if
some $a_k=0$ or $m<n$, every allocation has welfare zero.  Giving all
items to one agent is then optimal and has the required compressed
representation.  Hence, whenever $p\leq0$, we may assume that every
$a_k>0$ and $m\geq n$.

\medskip\noindent
\textbf{Finite $p<1$, $p\neq0$.}
For $0<p<1$, if every $a_k$ is zero, give all items to one agent;
this is optimal and uses at most two bundle sizes per type.
Otherwise, zero-valued types are assigned no items, and we restrict
attention to the remaining types.  Define
\[
 b:=\begin{cases}0,&0<p<1,\\1,&p<0,\end{cases}
 \qquad
 f_k(x):=\begin{cases}(a_kx)^p,&0<p<1,\\
                         -(a_kx)^p,&p<0.
          \end{cases}
\]
The domain of $f_k$ is $[0,\infty)$ in the first case and $[1,\infty)$
in the second.  Each $f_k$ is strictly increasing and strictly concave.
Maximizing $\PSW_p$ is equivalent to maximizing the sum of these
transformed utilities.  In the negative-$p$ case this equivalence uses
the fact that raising a positive number to the power $1/p$ reverses
order.  Moreover, an optimal allocation has positive welfare and hence
gives at least one item to every agent.

Give $b$ items to every agent and let $R:=m-bn$.  If $R=0$, this already
gives the desired allocation.  Otherwise, for each remaining type $k$
consider the strictly decreasing sequence
\begin{equation}
\label{eq:one-item-marginals}
 d_k(t):=f_k(b+t+1)-f_k(b+t),
 \qquad t=0,\ldots,R-1.
\end{equation}
Each entry occurs $n_k$ times, once for every agent of that type.
An agent receiving $z$ additional items contributes the prefix
$d_k(0)+\cdots+d_k(z-1)$ beyond its baseline utility.  Consequently,
the sum of the $R$ largest entries in the resulting multiset is an
upper bound on the additional transformed welfare of any allocation.
We next find these entries and show that they admit a feasible
allocation.  The multiset is represented implicitly by the $K$
sequences and their binary multiplicities; it is never expanded.

For a candidate value $v$, binary search in each decreasing sequence
computes
\[
 \begin{aligned}
 A_k(v)&:=\bigl|\{t\in\{0,\ldots,R-1\}:d_k(t)>v\}\bigr|,\\
 B_k(v)&:=\bigl|\{t\in\{0,\ldots,R-1\}:d_k(t)\geq v\}\bigr|.
 \end{aligned}
\]
Put $A(v):=\sum_k n_kA_k(v)$ and $B(v):=\sum_k n_kB_k(v)$.
The $R$-th largest entry has value $v$ precisely when
\begin{equation}
\label{eq:one-item-rank}
 A(v)<R\leq B(v).
\end{equation}
To find such a value, search each sequence $d_j$ in turn.  At a
midpoint $t$, compute the two ranks for $v=d_j(t)$.  If $A(v)\geq R$,
the desired value is strictly larger than $v$, so retain only indices
smaller than $t$.  If $B(v)<R$, retain only indices larger than $t$.
Otherwise, return $v$.  If this search exhausts the indices of one
sequence, proceed to the next.  Since the desired value belongs to
at least one sequence, one of these searches succeeds.  There are
$O(K\log(R+1))$ candidate evaluations, each requiring
$O(K\log(R+1))$ marginal comparisons.  Thus the total number of
comparisons is $O(K^2\log^2(R+1))$.

Let $\tau$ be the value returned.  Initially give every type-$k$ agent
$b+A_k(\tau)$ items.  This selects all marginal entries strictly larger
than $\tau$.  There remain $ E:=R-\sum_k n_kA_k(\tau)$ 
items.  Strict decrease of each sequence implies that
$B_k(\tau)-A_k(\tau)\in\{0,1\}$.  By
\eqref{eq:one-item-rank},
\[
 0<E\leq
 \sum_{k:\,B_k(\tau)=A_k(\tau)+1}n_k.
\]
For the types appearing in this sum, in any order, give one more
item to $\min\{E,n_k\}$ of their agents, and decrease $E$ by this
amount.  Stop when $E=0$.  All selected entries form prefixes for
their respective agents, so this is a feasible allocation attaining
the sum of the $R$ largest marginal entries.  It is therefore optimal.
Each type uses at most two consecutive bundle sizes, and their
multiplicities are computed with binary arithmetic.

\medskip\noindent
\textbf{Exact comparisons and bit complexity.}
Write the fixed parameter as $p=A/B$ in lowest terms, where
$A\in\mathbb Z\setminus\{0\}$ and $B\geq1$.  Comparing two entries
of~\eqref{eq:one-item-marginals} requires the sign of a sum of at most
four terms of the form $\pm r^{A/B}$.  Each $r$ is a rational number
of polynomial bit length, namely a per-item value multiplied by an
integer at most $m+1$.  Zero terms are treated separately; all
remaining $r$ are positive.  The number of terms, and the degrees of
their defining polynomials, depend only on $p$.

For completeness, exact signs can be determined as follows.  Each
$r^{A/B}$ is the positive root of $X^B-r^A$, after clearing
denominators.  A constant number of resultant computations produces
a nonzero integer polynomial $P$ whose roots include the difference
$\delta$ being compared.  The degree of $P$ is at most $B^4$ and its
coefficient bit lengths are polynomial in the input length.  Divide
out the largest power of $X$ that divides $P$, obtaining $Q$.
If $Q$ is constant, then $\delta=0$.  Otherwise, let $H\geq1$ bound
the absolute values of the coefficients of $Q$.  Since $Q$ has a
nonzero integer constant coefficient, applying Cauchy's root bound
to its reciprocal polynomial gives $ \delta\neq0\quad\Longrightarrow$
  $|\delta|\geq\frac{1}{1+H}.$  
Certified rational bisection of the positive $B$-th roots approximates
$\delta$ with absolute error less than $1/(4(1+H))$ in polynomial
time.  An approximate absolute value at most $1/(2(1+H))$ certifies
$\delta=0$; otherwise the approximation has the correct sign.
Thus each marginal comparison has polynomial bit complexity.
All ranks, multiplicities, and indices also have polynomial bit
length.  The selection and reconstruction therefore run in polynomial
time in the compact input length.  In particular, this argument never
compares a sum of a number of radicals that grows with $K$ or $n$.

\medskip\noindent
\textbf{Nash welfare, $p=0$.}
Under the positivity assumptions made above, maximizing Nash welfare
is equivalent to maximizing
\[
 \prod_{i=1}^n a_{k(i)}x_i
 =\left(\prod_{k=1}^K a_k^{n_k}\right)\prod_{i=1}^n x_i,
\]
where $k(i)$ is the type of agent $i$ and $x_i$ is its item count.
The first factor is independent of the allocation.  If
$x_i\geq x_j+2$, transferring one item from $i$ to $j$ increases
their product, since $ (x_i-1)(x_j+1)-x_ix_j=x_i-x_j-1>0.$ 
Hence a maximizing allocation has all item counts differing by at
most one.  Write $m=hn+r$, where $0\leq r<n$.  Give $h+1$ items to
$r$ agents and $h$ items to the others.  The agents receiving an extra
item can be assigned across types in one pass over the binary
multiplicities.  This gives an optimal compressed allocation with
at most two bundle sizes per type.

\medskip\noindent
\textbf{Egalitarian welfare, $p=-\infty$.}
Clear the denominators of the $a_k$ using a common positive integer
$D$, and write $c_k:=Da_k\in\mathbb Z_{>0}$.  The product of the input
denominators is a valid choice of $D$ and has polynomial bit length.
For an integer threshold $T\geq0$, an allocation with scaled minimum
utility at least $T$ exists if and only if
\begin{equation}
\label{eq:one-item-egalitarian}
 \sum_{k=1}^K n_k\left\lceil\frac{T}{c_k}\right\rceil\leq m.
\end{equation}
Indeed, every type-$k$ agent needs at least $\lceil T/c_k\rceil$
items, and these lower bounds can be met whenever their sum is at
most $m$.  All achievable scaled utilities are integers.  Binary
search for the largest feasible $T$ in
$[0,m\max_k c_k]$ therefore finds the exact optimum.  Each test of
\eqref{eq:one-item-egalitarian} uses polynomial-time binary arithmetic,
and the number of tests is polynomial in the input length.
Give each type-$k$ agent $\lceil T/c_k\rceil$ items and assign all
leftovers to one agent.  This is an optimal compressed allocation
with at most two bundle sizes per type.

\medskip\noindent
\textbf{Parameters $p\geq1$.}
Let $a_{\max}:=\max_k a_k$.  For any complete allocation,
\[
 \sum_{i=1}^n(a_{k(i)}x_i)^p
 \leq a_{\max}^p\sum_{i=1}^n x_i^p
 \leq a_{\max}^p\left(\sum_{i=1}^n x_i\right)^p
 =a_{\max}^pm^p.
\]
Giving all items to one agent of a type attaining $a_{\max}$ attains
this bound.  This completes the proof in every case.
\end{proof}

\begin{remark}[Allocation construction and threshold decision]
\label{rem:one-item-type-threshold}
Theorem~\ref{thm:one-item-type-exact} constructs an exact optimal
allocation.  It does not assert a polynomial-time algorithm for
comparing its welfare to an arbitrary rational threshold.  This
distinction matters even with one item type.

To see this, consider positive integers $b_1,\ldots,b_n,T$ and the
question whether $\sum_i\sqrt{b_i}\leq T$, an instance of
\textnormal{\textsc{Sqrt-Sum}}.  Set $p=-1/2$, take $m=n$ identical items, and
give agent $i$ per-item value $1/b_i$.  Every allocation with
positive welfare must give exactly one item to every agent.
Thus an optimal allocation is immediate, but its welfare is
\[
 \OPT_{-1/2}=\left(\frac{n}{\sum_{i=1}^n\sqrt{b_i}}\right)^2.
\]
The rational threshold $\tau:=n^2/T^2$ has polynomial encoding
length, and
\[
 \OPT_{-1/2}\geq\tau
 \quad\Longleftrightarrow\quad
 \sum_{i=1}^n\sqrt{b_i}\leq T.
\]
Hence the rational-threshold decision problem at $q=1$ and the
single fixed value $p=-1/2$ is at least as hard as
\textnormal{\textsc{Sqrt-Sum}}.  The allocation algorithm avoids this comparison:
each of its marginal comparisons involves at most four radicals,
whereas evaluating the final threshold can involve a sum with an
unbounded number of terms.
\end{remark}

\subsection{Exact egalitarian welfare with two item types}
\label{sec:two-item-types-egalitarian}

For egalitarian welfare, the case of two item types admits an exact
polynomial-time algorithm, even when the number of agent types is not
fixed. This is a consequence of the two-job-type algorithms of
McCormick, Smallwood, and Spieksma~\cite{McCormickSmallwoodSpieksma2001}.
We give a direct proof that also makes the dependence on the compact
input and the reconstruction of an allocation explicit.

\begin{theorem}
\label{thm:two-item-types-egalitarian}
For at most two item types and an arbitrary number $K$ of agent types,
exact egalitarian-welfare maximization with nonnegative rational
additive utilities is solvable in polynomial time. The running time is
polynomial in the compact input length, with both item and agent
multiplicities encoded in binary. The algorithm returns an optimal
compressed allocation using at most $K+2$ type--bundle pairs.
\end{theorem}

\begin{proof}
If there are no item types, every agent receives the empty bundle and
the optimum is zero. Otherwise, we may add an item type of multiplicity
zero if there is only one item type. Write $m_1,m_2$ for the two item
multiplicities, $n_k$ for the number of agents of type $k$, and $ u_k(x,y)=a_kx+b_ky$ 
for their common utility. Types of multiplicity zero can be deleted.
Thus $n_k\geq1$ for every remaining type. We first solve feasibility
for a common welfare threshold $T$, and then optimize the threshold.

\medskip
\noindent\textbf{The integer hull for one agent type.}
For each type $k$, let
\begin{equation}
\label{eq:two-types-covering-polyhedron}
 S_k(T)=\{(x,y)\in\mathbb Z_{\geq0}^{2}:
              a_kx+b_ky\geq T\},
 \qquad
 P_k(T)=\operatorname{conv}(S_k(T)).
\end{equation}
The threshold $T=0$ is always feasible: give all items to one agent
and leave the other bundles empty, using at most $K+1$ type--bundle
pairs. Suppose henceforth that $T>0$.
If $a_k=b_k=0$ for some $k$, the threshold is infeasible. If exactly one
coefficient is positive, then $P_k(T)$ is a translated nonnegative
orthant. For example, if $a_k>0$ and $b_k=0$, then $ P_k(T)=(\lceil T/a_k\rceil,0)+\mathbb R_{\geq0}^{2}.$

Suppose that both coefficients are positive, and put $ X_k=\lceil T/a_k\rceil$,
  $Y_k=\lceil T/b_k\rceil.$
Every point of $S_k(T)$ dominates a point of
$S_k(T)\cap([0,X_k]\times[0,Y_k])$: if its first coordinate exceeds
$X_k$, the point $(X_k,0)$ suffices, and the analogous statement holds
for the second coordinate. Consequently,
\begin{equation}
\label{eq:two-types-bounded-hull}
 P_k(T)=
 \operatorname{conv}\bigl(S_k(T)\cap
                  ([0,X_k]\times[0,Y_k])\bigr)
       +\mathbb R_{\geq0}^{2}.
\end{equation}
The bounded integer hull in~\eqref{eq:two-types-bounded-hull} is the
integer hull of a rational polygon described by five inequalities.
The standard polynomial-time integer-hull algorithm in dimension two~\cite{Hartmann89},
also used by McCormick et al.~\cite{McCormickSmallwoodSpieksma2001},
therefore computes its vertices in polynomial time. The number and
encoding lengths of the vertices are polynomial in the input length.
The same is true of the finite lower boundary of $P_k(T)$.

Each $P_k(T)$ is upward closed: if $v\in P_k(T)$ and
$r\in\mathbb R_{\geq0}^{2}$, then $v+r\in P_k(T)$. Indeed,
$S_k(T)$ is closed under addition of nonnegative integer vectors, and
$\operatorname{conv}(\mathbb Z_{\geq0}^{2})=
\mathbb R_{\geq0}^{2}$. Its finite lower boundary is an integral
polygonal chain whose successive vertices have strictly increasing
first coordinates and strictly decreasing second coordinates. A
translated orthant has only one such vertex. Notice also that
\begin{equation}
\label{eq:two-types-hull-integer-points}
 P_k(T)\cap\mathbb Z^{2}=S_k(T),
\end{equation}
because the original covering halfplane and the nonnegative orthant
are convex and contain $S_k(T)$.

\medskip
\noindent\textbf{The aggregate lower boundary.}
Consider the Minkowski sum
\begin{equation}
\label{eq:two-types-aggregate}
 P(T)=\sum_{k=1}^{K}n_kP_k(T).
\end{equation}
A threshold-feasible allocation implies
$(m_1,m_2)\in P(T)$. To see this, distribute any unused items
arbitrarily, which preserves feasibility by monotonicity. For each
type $k$, the sum of its $n_k$ feasible bundles belongs to
$n_kP_k(T)$, since $P_k(T)$ is convex.

To construct the finite lower boundary of $P(T)$, start with the sum
of the leftmost lower vertices of the scaled polyhedra $n_kP_k(T)$.
Merge their lower edges in slope order. When edges are parallel,
process them successively and retain the resulting intermediate
collinear points. Every segment in this representation then changes
the contribution of exactly one type, from $n_kA$ to $n_kB$, where
$A,B$ are consecutive lower vertices of $P_k(T)$; every other type
contributes a fixed scaled vertex. The boundary has at most the sum
of the numbers of local lower edges. We can record the type and the
two local endpoints responsible for each segment, together with the
current vertex of every other type. This construction is polynomial
in the compact input length and never expands the agent
multiplicities.

Membership of $z=(m_1,m_2)$ in $P(T)$ can now be checked in polynomial
time. We prove that membership is also sufficient for an integral
allocation, and give a compressed reconstruction.

\medskip
\noindent\textbf{Locating the target.}
Suppose that $z\in P(T)$. First check whether $z$ dominates any recorded lower
boundary point $V$. If it does not, it belongs to the closed axial triangle
\begin{equation}
\label{eq:two-types-axial-triangle}
 \operatorname{conv}\{V,W,(W_1,V_2)\}
\end{equation}
for two consecutive recorded points $V,W$, where
$V_1<W_1$ and $V_2>W_2$.
Indeed, if the first coordinate of $z$ lies between those of $V$ and
$W$, membership places $z$ on or above the segment $VW$.
If $z_2\geq V_2$, then $z$ dominates $V$. Otherwise it lies in the
triangle~\eqref{eq:two-types-axial-triangle}. If $z$ lies to the right
of the finite chain, membership implies that it dominates the
rightmost lower vertex. A chain consisting of one vertex gives only
the domination case.

\medskip
\noindent\textbf{Reconstruction in the domination case.}
Write $V=\sum_{k=1}^{K}n_kA_k,$ 
where $A_k$ is the recorded integral vertex of $P_k(T)$.
Give every type-$k$ agent bundle $A_k$. By
\eqref{eq:two-types-hull-integer-points}, every such bundle meets the
threshold. Give the nonnegative integer remainder $z-V$ to one
agent. This produces a feasible allocation with at most $K+1$
type--bundle pairs.

\medskip
\noindent\textbf{Reconstruction in the edge case.}
Suppose that $z$ belongs to the triangle associated with a segment
that changes type $k$. Write
\[
 V=C+sA,\qquad W=C+sB,\qquad
 C=\sum_{\ell\ne k}n_\ell A_\ell,\qquad s=n_k,
\]
where $A=(a,b)$ and $B=(c,d)$ are consecutive lower vertices of
$P_k(T)$, with $c>a$ and $b>d$. Give every type-$\ell$ agent, for
$\ell\ne k$, bundle $A_\ell$. The remaining resource vector
$Z=z-C$ is integral and belongs to $\operatorname{conv}\{sA,sB,(sc,sb)\}.$ 
In particular,
\begin{equation}
\label{eq:two-types-edge-inequality}
 sa\leq Z_1\leq sc,
 \qquad
 Z_2\geq sb+\frac{Z_1-sa}{c-a}(d-b).
\end{equation}
Set $ h=\left\lfloor\frac{Z_1-sa}{c-a}\right\rfloor.$ 
Since the domination case was excluded, $Z_1<sc$; otherwise
$z$ would dominate $W$. Hence $0\leq h<s$. Define
\begin{equation}
\label{eq:two-types-residual-bundle}
 R=Z-hB-(s-h-1)A.
\end{equation}
This vector is integral. Put $r=Z_1-sa-h(c-a)$, $0\leq r<c-a.$ 
Then $R_1=a+r$. Subtracting $hd+(s-h-1)b$ from
\eqref{eq:two-types-edge-inequality} yields $ R_2\geq b+\frac{r}{c-a}(d-b).$ 
The right-hand side is the second coordinate of the point on the
segment $AB$ with first coordinate $a+r$. It is at least $d\geq0$.
Thus $R$ is nonnegative and lies on or above $AB$. Since $P_k(T)$
contains this segment and is upward closed, $R\in P_k(T)$.
Equation~\eqref{eq:two-types-hull-integer-points} gives $R\in S_k(T)$.

Give $h$ agents of type $k$ bundle $B$, give $s-h-1$ agents bundle
$A$, and give the remaining agent bundle $R$. The numbers of agents
are nonnegative integers summing to $s$, and
\eqref{eq:two-types-residual-bundle} shows that their total resource
use is exactly $Z$. Every agent meets the threshold. After omitting
zero multiplicities and merging equal bundles, type $k$ uses at
most three type--bundle pairs and every other type uses one.
The total is therefore at most $K+2$.

This proves
\begin{equation}
\label{eq:two-types-feasibility-equivalence}
 T\text{ is feasible}
 \quad\Longleftrightarrow\quad
 (m_1,m_2)\in\sum_{k=1}^{K}n_kP_k(T),
\end{equation}
together with polynomial-time compressed reconstruction. All the
numbers used above have polynomial encoding length. In particular,
multiplication by $n_k$, comparison of slopes, and the quotient
calculation defining $h$ operate directly on binary integers and
rationals.

\medskip
\noindent\textbf{Exact optimization.}
Let $D$ be the product of the denominators of all input utility
coefficients. Multiplying the coefficients by $D$ produces
nonnegative integer coefficients $a'_k,b'_k$. The encoding length of
$D$ is at most the sum of the denominator encoding lengths, so this
scaling is polynomial in the input length. Every attained
egalitarian welfare in the scaled instance is an integer, and $ U=\max_{k}\{a'_km_1+b'_km_2\}$ 
is an upper bound with polynomial encoding length. If $U=0$, give all items to one agent and empty bundles to
the others, using at most $K+1$ compressed records. Otherwise, binary search over the integer
thresholds $0,\ldots,U$, using
\eqref{eq:two-types-feasibility-equivalence}, finds the largest
feasible threshold in $O(\log(U+1))$ calls. Reconstruct the allocation
at that threshold. Dividing its welfare by $D$ gives the exact
optimum for the original instance, and does not change the
allocation. The number of type--bundle pairs is at most $K+2$, as
claimed.
\end{proof}

\section{A PTAS for $p$-Mean Welfare with a Constant Number of Agent Types}
\label{sec:constant-agent-types}

\subsection{Statement}

In this section, the number $K$ of agent types is fixed, whereas the number
$q$ of good types is part of the input.  We use the model, welfare
conventions, and compressed-output representation from
Section~\ref{sec:notation}.

\begin{theorem}
\label{thm:constant-agent-types}
For every fixed number $K$ of agent types, every fixed rational $p<1$, and every $\varepsilon>0$, maximizing $p$-mean social welfare under additive nonnegative utilities admits a polynomial-time approximation scheme in the high-multiplicity model.  More precisely, for fixed $K,p,\varepsilon$, the algorithm runs in time polynomial in $q$ and in the binary encoding length of the instance and returns an allocation $x$ satisfying $\PSW_p(x)\geq(1-\varepsilon)\OPT_p$.
\end{theorem}

\begin{remark}
Theorem~\ref{thm:constant-agent-types} concerns finite $p<1$ and
therefore includes $p=0$, corresponding to Nash social welfare.
The case $p=0$ was previously established by Nguyen, Elbassioni, and
Rothe~\cite{NER24}; the theorem extends the result to every fixed finite
rational $p<1$. The limiting case $p=-\infty$ is not claimed here.
By Corollary~\ref{cor:no-fptas-one-agent-type}, this PTAS cannot be
strengthened to an FPTAS unless $\mathsf{P}=\mathsf{NP}$, already with
one agent type.
\end{remark}

\begin{proof}
If $m=0$, return the allocation giving every agent the empty bundle.
This allocation is optimal and has one compressed record per nonempty
agent type. Henceforth assume $m>0$, and therefore $q\geq1$.
Empty agent types may be deleted, so assume $n_k\geq1$ for every
$k\in[K]$. If $\varepsilon\geq1$, any allocation satisfies the claimed
guarantee. Hence, assume $\varepsilon\in(0,1)$ and choose a positive
rational $\delta$ satisfying
\begin{equation}
\label{eq:constant-agent-delta}
 0<\delta\leq
 \begin{cases}
  \displaystyle
  \min\left\{
   \frac{1}{10},
   \frac{\varepsilon}{12+2/|p|}
  \right\},
  &p\neq0,\\[4mm]
  \displaystyle
  \min\left\{
   \frac{1}{10},
   \frac{\varepsilon}{14}
  \right\},
  &p=0.
 \end{cases}
\end{equation}

\subsection{Algorithmic overview}
\label{subsec:constant-agent-types-overview}

We first describe the complete algorithm at a high level. The subsequent
subsections define its ingredients and prove their guarantees. The structural
lemma separates, within every agent type, singleton agents from residual
agents controlled by a common utility threshold. The algorithm enumerates a
geometric approximation of every threshold, represents residual agents by a
constant collection of configurations, and solves the resulting aggregate
allocation problem. Support guessing leaves only a constant number of integer
variables. A totally unimodular residual program then rounds the remaining
aggregate variables, after which batched procedures reconstruct a compressed
allocation. Algorithm~\ref{alg:constant-agent-types-ptas} summarizes these
steps.

\begin{algorithm}[t]
\caption{PTAS for a fixed number of agent types}
\label{alg:constant-agent-types-ptas}
\begin{algorithmic}[1]
\Require A high-multiplicity instance, a fixed rational $p<1$, and
$\varepsilon\in(0,1)$
\If{$m=0$}
 \State \Return the compressed allocation giving every agent the empty bundle
\EndIf
\State Choose $\delta$ according to \eqref{eq:constant-agent-delta}
\If{$p\leq0$}
 \State Run the positive-utility flow test; if it fails, return any allocation
\EndIf
\State Construct the threshold sets $\mathcal T_1,\ldots,\mathcal T_K$
\For{each enumerated threshold vector not discarded by the positivity condition}
 \State Normalize and round the values and construct the retained configurations
 \State Form the rationalized configuration integer program
 \For{each enumerated small-good support tuple $\mathbf J$}
  \State Solve the restricted mixed-integer program; skip it if it is infeasible
  \State Fix its integer variables and compute an optimal basic solution of the
  resulting linear program
  \State Round the aggregate assignment by solving the reduced assignment LP
  \State Reconstruct and store the corresponding compressed allocation
 \EndFor
\EndFor
\State Evaluate all reconstructed candidates using directed rational
approximations of their actual utilities
\State \Return the candidate with the best certified objective value
\end{algorithmic}
\end{algorithm}

\subsection{Transformed objectives and positivity}
\label{subsec:transformed-objectives-positivity}

We use the transformed function $\phi_p$ from Section~\ref{sec:notation}.
For every finite $p<1$, it is increasing and strictly concave on its stated
domain.

For $p\leq0$, apply the positive-utility flow test from
Section~\ref{subsec:fixed-item-configurations}. Its construction and proof
apply to arbitrary $K$ and $q$ and run in polynomial time in the compact
input length. If the maximum flow has value less than $n$, then
$\OPT_p=0$; return any compressed allocation and terminate. Henceforth,
when $p\leq0$, assume that a positive allocation exists and discard
zero-utility configurations.

\subsection{The residual structural lemma}

\begin{lemma}[Balancing exchange]
\label{lem:balancing-exchange}
Let $I\subseteq\mathbb R$ be an interval, and let
$\Phi:I\to\mathbb R$ be strictly concave. Suppose
$H,L,H-d,L+d\in I$, $d>0$, and $H-d>L$. Then $\Phi(H-d)+\Phi(L+d)>\Phi(H)+\Phi(L).$
\end{lemma}

\begin{proof}
Set $\lambda:=d/(H-L)$. The assumptions give $0<\lambda<1$, and $ H-d=(1-\lambda)H+\lambda L$, $L+d=\lambda H+(1-\lambda)L.$ 
Strict concavity at these two convex combinations yields $\Phi(H-d)>(1-\lambda)\Phi(H)+\lambda\Phi(L),$ $\Phi(L+d)>\lambda\Phi(H)+(1-\lambda)\Phi(L).$ 
Adding proves the assertion.
\end{proof}

\begin{lemma}[Structure of an optimal allocation]
\label{lem:residual-threshold}
Fix a finite $p<1$. If $p\leq0$, assume that an allocation giving
every agent positive utility exists. There is an optimal allocation
$
 x^*=(x_{i,j}^*:i\in[n],\ j\in[q])
$
with the following property. For every agent type $k$, its agent set
$A_k$ can be partitioned as
$
 A_k=H_k\mathbin{\dot\cup}R_k,
$
where $R_k\neq\varnothing$, and there is a threshold $\tau_k\geq0$
such that:
\begin{enumerate}
\item Every agent $i\in H_k$ receives exactly one good having positive
      value to type $k$:
      $
       \sum_{j:u_{k,j}>0}x_{i,j}^*=1.
      $
      If this good is of type $j$, then $u_{k,j}>\tau_k$.

\item Every positive-value good assigned to a residual agent has value
      at most $\tau_k$. That is, for every $i\in R_k$ and $j\in[q]$, if $ x_{i,j}^*>0$ and $u_{k,j}>0$
       then $u_{k,j}\leq\tau_k.$

\item If $\tau_k>0$, then every $i\in R_k$ satisfies $\frac{\tau_k}{2}
       \leq u_k(x^*_i)
       \leq2\tau_k.$

\item If $\tau_k=0$, then $u_k(x^*_i)=0$ for every $i\in R_k$.
\end{enumerate}

The threshold is the average utility of the residual agents: $\tau_k
 =
 \frac{1}{|R_k|}
 \sum_{i\in R_k}u_k(x^*_i).$
\end{lemma}

\begin{proof}
For $0<p<1$, choose an allocation $x^*$ maximizing the transformed
objective $\Phi_p$ over all allocations. For $p\leq0$, choose $x^*$ to
maximize $\Phi_p$ over the nonempty set of allocations giving every
agent positive utility. In either case, $x^*$ maximizes $\PSW_p$.

Fix an agent type $k$. We construct $H_k$, $R_k$, and $\tau_k$ while
leaving $x^*$ unchanged. Initially, let $H_k=\varnothing$ and
$R_k=A_k$. While $R_k$ is being constructed, write $ \tau(R_k)
 :=
 \frac{1}{|R_k|}
 \sum_{i\in R_k}u_k(x^*_i)$ 
for its current average utility.

Suppose that some agent $a\in R_k$ receives a good of type $j$
satisfying $ u_{k,j}>\tau(R_k).$ 
We claim that this is the only positive-value good received by $a$.
Indeed, since $u_k(x_i^*)$ has average $\tau(R_k)$ over $R_k$, there
is an agent $b\in R_k$ such that $ u_k(x^*_b)\leq\tau(R_k).$ 
Moreover, $b\neq a$, because $ u_k(x^*_a)\geq u_{k,j}>\tau(R_k).$

Suppose that $a$ receives another positive-value good, of type $j'$.
Let $x'$ be obtained from $x^*$ by transferring this good from $a$
to $b$; that is, $x'_{a,j'}=x_{a,j'}^*-1,$ and  $ x'_{b,j'}=x_{b,j'}^*+1,$ 
and all other entries remain unchanged. The allocation $x'$ is
feasible. Since the selected type-$j$ good remains with $a$, $u_k(x'_a)
 =
 u_k(x_a^*)-u_{k,j'}
 \geq u_{k,j}
 >
 \tau(R_k)
 \geq u_k(x_b^*),$ 
whereas $u_k(x'_b)=u_k(x_b^*)+u_{k,j'}$. When $p\leq0$, these
inequalities also show that both affected agents retain positive
utility, so $x'$ lies in the domain of $\phi_p$.

Lemma~\ref{lem:balancing-exchange} now gives $\phi_p\bigl(u_k(x'_a)\bigr)+\phi_p\bigl(u_k(x'_b)\bigr)
 >
\phi_p\bigl(u_k(x_a^*)\bigr)+\phi_p\bigl(u_k(x_b^*)\bigr).$ 
All other agents' bundles are unchanged, and hence
$\Phi_p(x')>\Phi_p(x^*)$, contradicting the optimality of $x^*$.
Therefore, $a$ receives exactly one positive-value good, namely the
selected type-$j$ good. Add $a$ to $H_k$ and remove it from $R_k$. 
The set $R_k$ contains at least two agents before this removal. Indeed, if
$R_k=\{a\}$, then $\tau(R_k)=u_k(x_a^*)$, and no good assigned to $a$
can have value greater than $u_k(x_a^*)$. Hence, if $r=|R_k|$ before
the removal and $\tau=\tau(R_k)$, then $r\geq2$. Since the removed agent
has utility $u_{k,j}>\tau$, the new residual average is $ \frac{r\tau-u_{k,j}}{r-1}<\tau.$ 

Continue removing agents in this way until no residual agent receives
a good whose value is greater than the current residual average. The
procedure terminates after at most $|A_k|-1$ removals. Let $H_k$ and
$R_k$ be the resulting sets, and put $ \tau_k
 :=
 \frac{1}{|R_k|}
 \sum_{i\in R_k}u_k(x_i^*).$ 
Each agent placed in $H_k$ receives exactly one positive-value good.
The value of that good exceeded the residual average at the time of
removal. Because the residual averages strictly decrease, its value
also exceeds the final average $\tau_k$. This proves property~1. Moreover, at termination, no residual agent receives a good whose value is
greater than $\tau_k$. This proves property~2.

Suppose that $\tau_k>0$. If some $a\in R_k$ had
$u_k(x^*_a)>2\tau_k$, then some $b\in R_k$ would satisfy
$u_k(x^*_b)\leq\tau_k$. Choose any positive-value good of type $j$
assigned to $a$. Property~2 gives $u_{k,j}\leq\tau_k$, and hence $ u_k(x^*_a)-u_{k,j}
 > \tau_k \geq u_k(x^*_b).$ Moving this good from $a$ to $b$ would strictly improve the transformed
objective by Lemma~\ref{lem:balancing-exchange}, contradicting
optimality. Therefore $u_k(x_i^*)\leq2\tau_k$ for every $i\in R_k$. Conversely, suppose that some $b\in R_k$ had
$u_k(x_b^*)<\tau_k/2$. Because $\tau_k$ is the residual average, there
would be an agent $a\in R_k$ with $u_k(x_a^*)>\tau_k$. Agent $a$ must
receive at least two positive-value goods; otherwise its unique
positive-value good would have value greater than $\tau_k$, contrary
to property~2. Choose a positive-value good of minimum value among those assigned to
$a$, and let $j$ be its type. Since $a$ receives at least two
positive-value goods, $ u_{k,j}\leq\frac{u_k(x^*_a)}{2}.$ 
It follows that $u_k(x^*_a)-u_{k,j}
 \geq\frac{u_k(x^*_a)}{2}
 >
 \frac{\tau_k}{2}
 >
 u_k(x^*_b).$ Moving this good from $a$ to $b$ would again contradict optimality by
Lemma~\ref{lem:balancing-exchange}. Thus
$u_k(x_i^*)\geq\tau_k/2$ for every $i\in R_k$, proving property~3.

Finally, if $\tau_k=0$, then $\sum_{i\in R_k}u_k(x^*_i)=0.$ 
Nonnegativity implies $u_k(x_i^*)=0$ for every $i\in R_k$, proving
property~4. When $p\leq0$, this case cannot occur because $x^*$ gives
every agent positive utility.

Applying this construction independently to every agent type proves
the lemma.
\end{proof}

\subsection{Guessing and rounding the residual thresholds}
Use the rational parameter $\delta$ fixed in
\eqref{eq:constant-agent-delta}. Fix also the optimal allocation $x^*$ and the sets
$H_k,R_k$ and thresholds $\tau_k$ given by
Lemma~\ref{lem:residual-threshold}. We refer to $x^*$ as the
\emph{canonical optimum}. This optimum and the preceding peeling
construction are used only in the analysis; the algorithm computes
neither of them.

For every agent type $k$, define
$
 U_k:=\sum_{j=1}^q m_j u_{k,j}.
$
Thus $U_k$ is an upper bound on the total utility that all available
goods could provide to type-$k$ agents. For every agent type $k$, define
$
 \mathcal J_k^+:=\{j\in[q]:u_{k,j}>0\}.
$
If $\mathcal J_k^+\neq\varnothing$, define
$
 u_k^{\min}:=\min_{j\in\mathcal J_k^+}u_{k,j}.
$
If $\mathcal J_k^+=\varnothing$, then $U_k=0$, every type-$k$ agent has utility
zero under every allocation, and we enumerate only the threshold
guess $0$.

Suppose first that $\tau_k>0$. By the definition of $\tau_k$, $|R_k|\tau_k
 =
 \sum_{i\in R_k}u_k(x_i^*).$ 
The sum on the right is positive and is composed of values of goods
assigned to agents in $R_k$. It is therefore at least
$u_k^{\min}$. Since $|R_k|\leq n_k$, we obtain $\tau_k\geq\frac{u_k^{\min}}{|R_k|}
 \geq\frac{u_k^{\min}}{n_k}.$ 
Moreover, $\tau_k
 \leq
 \sum_{i\in R_k}u_k(x_i^*)
 \leq
 U_k.$ 
Consequently, $\frac{u_k^{\min}}{n_k}\leq\tau_k\leq U_k.$ 

For a type $k$ with $U_k>0$, let $s_k$ be the smallest nonnegative
integer satisfying $ \frac{u_k^{\min}}{n_k}(1+\delta)^{s_k}\geq U_k,$ 
and define the set of threshold guesses
\begin{equation}\label{eq:Tk}
 \mathcal T_k
 :=
 \{0\}
 \cup
 \left\{
 \frac{u_k^{\min}}{n_k}(1+\delta)^s:
 s=0,1,\ldots,s_k
 \right\}.
\end{equation}
If $U_k=0$, define $\mathcal T_k:=\{0\}$.

If $\tau_k=0$, then $0\in\mathcal T_k$ is a correct guess. If
$\tau_k>0$, choose the smallest positive
$\widehat\tau_k\in\mathcal T_k$ satisfying
$\widehat\tau_k\geq\tau_k$. By minimality, $ \tau_k
 \leq
 \widehat\tau_k
 <
 (1+\delta)\tau_k.$ 
We call such a value $\widehat\tau_k$ a \emph{correct threshold
guess} for type $k$.

For fixed $\delta$, $ |\mathcal T_k|
 =
 O_\delta\left(
 1+\log\frac{n_kU_k}{u_k^{\min}}
 \right)$
 whenever $U_k>0$. The logarithm is bounded by a polynomial in the
binary encoding length of the instance. The guesses themselves also
have polynomial encoding length because $\delta$ is fixed and
rational. Since $K$ is fixed, all vectors $ (\widehat\tau_1,\ldots,\widehat\tau_K)
 \in
 \mathcal T_1\times\cdots\times\mathcal T_K$ 
can be enumerated in polynomial time, and at least one enumerated
vector is correct for every agent type.

\medskip
\noindent\textbf{Normalization and geometric rounding.}
Consider a type $k$ with $\widehat\tau_k>0$, and define $ \alpha_k:=\delta\widehat\tau_k.$ 
For every $j\in[q]$ with $u_{k,j}>0$, define the rounded normalized
value
\[
 \overline u_{k,j}
 :=
 (1+\delta)^{
 \left\lfloor
 \log_{1+\delta}\left(\frac{u_{k,j}}{\alpha_k}\right)
 \right\rfloor
 }.
\]
If $u_{k,j}=0$, set $\overline u_{k,j}:=0$. Thus, for every
$u_{k,j}>0$,
\begin{equation}\label{eq:geometric-rounding}
 \frac{u_{k,j}}{\alpha_k(1+\delta)}
 <
 \overline u_{k,j}
 \leq
 \frac{u_{k,j}}{\alpha_k}.
\end{equation}
The exponent can be computed in polynomial time by binary search,
using exact comparisons between rational powers.

Among the positive-value goods, call a good of type $j$

\begin{itemize}
\item \emph{small} for type $k$ if
      $0<\overline u_{k,j}<1$;

\item \emph{large} for type $k$ if
      $1\leq\overline u_{k,j}\leq1/\delta$;

\item \emph{huge} for type $k$ if
      $\overline u_{k,j}>1/\delta$.
\end{itemize}

Under a correct threshold guess, every huge good has original value
greater than $\tau_k$. Indeed, $ u_{k,j}
 \geq
 \alpha_k\overline u_{k,j}
 >
 \delta\widehat\tau_k\frac{1}{\delta}
 =
 \widehat\tau_k
 \geq
 \tau_k.$ 
It follows from Lemma~\ref{lem:residual-threshold} that any such good
assigned to a type-$k$ agent in $x^*$ is the unique positive-value
good assigned to an agent in $H_k$.

\medskip
\noindent\textbf{Regular configurations.}
Define
\[
 r_{\max}
 :=
 \left\lfloor
 \log_{1+\delta}\frac{1}{\delta}
 \right\rfloor,
 \qquad
 \mathcal R_\delta:=\{0,1,\ldots,r_{\max}\},
\]
and let $ B_\delta:=|\mathcal R_\delta|=r_{\max}+1.$ 
Every rounded large value is equal to $(1+\delta)^r$ for some
$r\in\mathcal R_\delta$.

For a type $k$ satisfying $\widehat\tau_k>0$, a regular configuration is a tuple
$
 C=
 \left(
 k(C),\,(a_r(C))_{r\in\mathcal R_\delta},\,\mu(C)
 \right),
$
where
\begin{itemize}
\item $k(C)\in[K]$ is an integer recording the agent type of the configuration;
\item $a_r(C)\geq0$, for $r\in\mathcal R_\delta$, is an integer representing the number of large goods of rounded normalized value $(1+\delta)^r$; and
\item $\mu(C)\geq0$ is an integer representing a rounded-down lower requirement on the total rounded normalized value of the small
      goods.
\end{itemize}

The estimated normalized utility of $C$ is $ w(C)
 :=
 \sum_{r\in\mathcal R_\delta}
 a_r(C)(1+\delta)^r+\mu(C).$ 
We retain only configurations satisfying
\begin{equation}\label{eq:retained}
 \frac{1}{5\delta}
 \leq
 w(C)
 \leq
 \frac{4}{\delta}.
\end{equation}
For every type $k$ with $\widehat\tau_k>0$, let $\mathcal C_k$ denote
the set of retained regular configurations $C$ satisfying $k(C)=k$.  For every type
$k$ with $\widehat\tau_k=0$, define
$
 \mathcal C_k:=\varnothing.
$
Finally, put $ \mathcal C
 :=
 \bigcup_{k:\widehat\tau_k>0}\mathcal C_k.$ 
Because the agent type is included as a component of every
configuration, the sets $\mathcal C_k$ are pairwise disjoint.

Because each large good contributes at least one to $w(C)$, every
retained configuration satisfies $ \sum_{r\in\mathcal R_\delta}a_r(C)
 \leq
 \frac{4}{\delta}$ and $\mu(C)\leq\frac{4}{\delta}.$ 
Hence, for each type $k$, the number of retained regular
configurations is at most $\left(
 \left\lfloor\frac{4}{\delta}\right\rfloor+1
 \right)^{B_\delta+1},$ 
which depends only on $\delta$. Across all agent types, the number of
regular configurations depends only on $K$ and $\delta$.

\medskip
\noindent\textbf{Representation of canonical residual agents.}
Assume that $\widehat\tau_k$ is a correct positive guess. Fix an agent
$i\in R_k$, and put $ V_i:=\frac{u_k(x_i^*)}{\alpha_k}.$ 
Lemma~\ref{lem:residual-threshold} and the correctness of the guess
give $\frac{1}{2\delta(1+\delta)}
 <
 V_i
 \leq
 \frac{2}{\delta}.$ 
Indeed, $u_k(x_i^*)\geq\frac{\tau_k}{2}$, $\alpha_k
 =
 \delta\widehat\tau_k
 <
 \delta(1+\delta)\tau_k,$ 
which gives the lower bound, while $u_k(x_i^*)\leq2\tau_k$,
 $\alpha_k\geq\delta\tau_k,$ gives the upper bound.

No huge good is assigned to $i$, because every huge good has original
value greater than $\tau_k$. Let $L_i$ and $S_i$ denote,
respectively, the total rounded normalized values of the large and
small goods assigned to $i$. Then $ \overline V_i:=L_i+S_i$ 
is the total rounded normalized utility of $i$. We get by~\eqref{eq:geometric-rounding} $\frac{V_i}{1+\delta}
 <
 \overline V_i
 \leq
 V_i.$

Define a regular configuration $C_i$ by retaining the number of large
goods in each class and setting $ \mu(C_i):=\lfloor S_i\rfloor.$ Its estimated normalized utility is 
$w(C_i)=L_i+\lfloor S_i\rfloor.$ 
Since $\lfloor S_i\rfloor>S_i-1$, we have $w(C_i)
 >
 \overline V_i-1
 >
 \frac{V_i}{1+\delta}-1.$
Consequently, $w(C_i)
 >
 \frac{1}{2\delta(1+\delta)^2}-1
 \geq
 \frac{1}{5\delta},$ where the last inequality uses $\delta\leq1/10$. Furthermore, $ w(C_i)
 \leq
 \overline V_i
 \leq
 V_i
 \leq
 \frac{2}{\delta}
 \leq
 \frac{4}{\delta}.$ 
Thus $C_i$ is retained by~\eqref{eq:retained}.

Because $w(C_i)$ is expressed in normalized units,
$\alpha_k w(C_i)$ is the estimated utility of $C_i$ on the original
utility scale. From
the preceding inequalities, $\alpha_kw(C_i)
 >
 \frac{u_k(x_i^*)}{1+\delta}-\alpha_k.$ Since $u_k(x_i^*)\geq\tau_k/2$ and
$\alpha_k<\delta(1+\delta)\tau_k$, we have $\alpha_k
 <
 2\delta(1+\delta)u_k(x_i^*).$ 
It follows that
\begin{align*}
 \alpha_kw(C_i)
 >
 \left(
 \frac{1}{1+\delta}-2\delta(1+\delta)
 \right)u_k(x_i^*)
 \geq
 (1-4\delta)u_k(x_i^*),
\end{align*}
where the final inequality holds for $\delta\leq1/10$.

\medskip
\noindent\textbf{Representation of canonical singleton agents.}
Suppose that $\widehat\tau_k$ is a correct positive guess. Consider an
agent $i\in H_k$, and suppose that its unique positive-value good is of
type $j$. If
$\overline u_{k,j}>1/\delta$, represent $i$ by the corresponding huge
singleton.

Suppose instead that $\overline u_{k,j}\leq1/\delta$. Since
$u_{k,j}>\tau_k$ and
$\alpha_k<\delta(1+\delta)\tau_k$, we have $ \overline u_{k,j}
 >
 \frac{u_{k,j}}{\alpha_k(1+\delta)}
 >
 \frac{1}{\delta(1+\delta)^2}.$ 
In particular, this good is large. Represent $i$ by the regular
configuration having one good in the corresponding large class,
no other large goods, and $\mu=0$. Its estimated normalized utility is
$w(C_i)=\overline u_{k,j}$. For $\delta\leq1/10$, $ \frac{1}{5\delta}
 <
 w(C_i)
 \leq
 \frac{1}{\delta}
 \leq
 \frac{4}{\delta},$ 
so this configuration is retained. Moreover, $\alpha_kw(C_i)
 >
 \frac{u_{k,j}}{1+\delta}
 \geq
 (1-4\delta)u_{k,j}.$

\medskip
\noindent\textbf{Zero-threshold types.}
Finally, suppose that the correct guess is
$\widehat\tau_k=\tau_k=0$. By
Lemma~\ref{lem:residual-threshold}, every agent in $R_k$ has utility
zero, while every agent in $H_k$ receives exactly one positive-value
good. This case can occur only when $0<p<1$, because for $p\leq0$ the
canonical optimum gives every agent positive utility. For a zero-threshold type, represent each agent in $R_k$ by a zero
configuration and represent each agent in $H_k$ by its positive
singleton with its exact, unnormalized value $u_{k,j}$. Goods having
zero value to type $k$ need not be represented and may temporarily
remain unused. We summarize these observations in the following lemma.

\begin{lemma}[Threshold guesses and bundle representation]
\label{lem:threshold-bundle-representation}
Fix $K$ and a rational $\delta\in(0,1/10]$ and let $\mathcal T_k$ be defined as in~\eqref{eq:Tk}. There exists a threshold vector $(\widehat\tau_1,\ldots,\widehat\tau_K)
 \in
 \mathcal T_1\times\cdots\times\mathcal T_K$ that satisfies, for every $k\in[K]$, $\widehat\tau_k=0$ 
if $\tau_k=0,$ 
and $\tau_k
 \leq
 \widehat\tau_k
 <
 (1+\delta)\tau_k$
if $\tau_k>0$. 
Moreover, for every type $k$ with
$\widehat\tau_k>0$, every bundle $x_i^*$ assigned to an agent
$i\in A_k$ in the canonical optimum can be represented in one of the
following ways:

\begin{enumerate}
\item If $i\in R_k$, then $x_i^*$ induces a retained regular
      configuration $C_i$ (satisfying~\eqref{eq:retained}).

\item If $i\in H_k$ and its unique positive-value good is huge, then
      $x_i^*$ is represented by the corresponding huge singleton.

\item If $i\in H_k$ and its unique positive-value good is not huge,
      then $x_i^*$ induces a retained regular configuration containing
      exactly that one large good.
\end{enumerate}

Let $\widetilde u_i$ denote the represented utility on the original
scale, namely 
\[
 \widetilde u_i
 :=
 \begin{cases}
 \alpha_k w(C_i),
 &\text{if $x_i^*$ is represented by a regular configuration},\\[1ex]
 \alpha_k\overline u_{k,j},
 &\text{if $x_i^*$ is represented by a huge type-$j$ singleton},
 \end{cases}
\]
Then
\[
 (1-4\delta)u_k(x_i^*)
 \leq
 \widetilde u_i
 \leq
 u_k(x_i^*).
\]

If $\widehat\tau_k=\tau_k=0$, then $0<p<1$, every agent in $R_k$
has utility zero, and every agent in $H_k$ receives exactly one
positive-value good. Such agents are represented, respectively, by
zero configurations and exact positive singletons.

Finally, the number of retained regular configurations over all agent
types is bounded by a function of $K$ and $\delta$ only.
\end{lemma}
\subsection{The configuration formulation}
\label{subsec:configuration-formulation}

Fix an enumerated threshold vector
$(\widehat\tau_1,\ldots,\widehat\tau_K)$. Define
\[
 \mathcal K_+
 :=
 \{k\in[K]:\widehat\tau_k>0\}
 \qquad\text{and}\qquad
 \mathcal K_0
 :=
 \{k\in[K]:\widehat\tau_k=0\}.
\]
We call the types in $\mathcal K_+$ \emph{positive-threshold types}
and those in $\mathcal K_0$ \emph{zero-threshold types}. For a correct
threshold vector, these sets agree, respectively, with the types
satisfying $\tau_k>0$ and $\tau_k=0$.

When $p\leq0$, the positivity assumption implies that every true
threshold $\tau_k$ is positive. We may therefore discard every
enumerated threshold vector for which $\mathcal K_0\neq\varnothing$ in this case.

For every $k\in\mathcal K_+$, define
\[
 \mathcal S_k
 :=
 \{j\in[q]:0<\overline u_{k,j}<1\},
\]
\[
 \mathcal L_{k,r}
 :=
 \{j\in[q]:\overline u_{k,j}=(1+\delta)^r\},
 \qquad r\in\mathcal R_\delta,
\]
and
\[
 \mathcal J_k^{\mathrm{huge}}
 :=
 \left\{
 j\in[q]:
 \overline u_{k,j}>\frac1\delta
 \right\}.
\]
Thus $\mathcal S_k$ is the set of small good types,
$\mathcal L_{k,r}$ is the set of large good types in class $r$, and
$\mathcal J_k^{\mathrm{huge}}$ is the set of huge good types, all
relative to agent type $k$.

Recall that $\mathcal C_k$ is the set of retained regular
configurations of type $k$, and
$
 \mathcal C=\bigcup_{k=1}^K\mathcal C_k.
$
For $k\in\mathcal K_0$, we have $\mathcal C_k=\varnothing$.

\medskip
\noindent\textbf{Variables and represented values.}
For every $k\in[K]$ and $j\in[q]$, the variable $x_{k,j}$ is the number
of type-$j$ goods reserved for type-$k$ agents. For every
$C\in\mathcal C$, the variable $y(C)$ is the number of agents of type
$k(C)$ represented by $C$. For every $k\in\mathcal K_0$, the variable
$z_k$ is the number of type-$k$ agents represented by the zero
configuration.

Define the set of singleton pairs $\mathcal I_{\mathrm{sing}}
 :=
 \bigl\{
 (k,j):
 k\in\mathcal K_+,\ 
 j\in\mathcal J_k^{\mathrm{huge}}
 \bigr\}
 \cup
 \bigl\{
 (k,j):
 k\in\mathcal K_0,\ 
 j\in\mathcal J_k^+
 \bigr\}.$ 
For every $(k,j)\in\mathcal I_{\mathrm{sing}}$, define its represented
utility by
\[
 v_{k,j}
 :=
 \begin{cases}
 \alpha_k\overline u_{k,j},
 & k\in\mathcal K_+,\\[1ex]
 u_{k,j},
 & k\in\mathcal K_0.
 \end{cases}
\]
Thus, huge singletons use their rounded original-scale values, whereas
positive singletons of zero-threshold types use their exact values.

For every regular configuration $C\in\mathcal C$, define its
represented utility on the original scale by $ v(C):=\alpha_{k(C)}w(C).$ 
Define the exact transformed objective
\[
 \Psi_p(x,y)
 :=
 \sum_{C\in\mathcal C}
 \phi_p\bigl(v(C)\bigr)y(C)
 +
 \sum_{(k,j)\in\mathcal I_{\mathrm{sing}}}
 \phi_p(v_{k,j})x_{k,j}.
\]

\medskip
\noindent\textbf{Configuration ILP.}
The exact configuration formulation is the following integer linear
program.
\begin{subequations}
\label{eq:configuration-ip}
\begin{empheq}[
 left={\textnormal{(Conf-IP)}\quad\empheqlbrace}
]{align}
 &\underset{x,y,z}{\operatorname{maximize}}
 &&\Psi_p(x,y)
 &&\label{eq:configuration-objective}\\[-1mm]
 &\text{subject to:}\notag\\[-1mm]
 &&&\sum_{k=1}^Kx_{k,j}\leq m_j
 &&j\in[q],
 &&\label{eq:configuration-item-capacity}\\
 &&&\sum_{C\in\mathcal C_k}y(C)
   +\sum_{j\in\mathcal J_k^{\mathrm{huge}}}x_{k,j}=n_k
 &&k\in\mathcal K_+,
 &&\label{eq:configuration-agent-count-positive}\\
 &&&\sum_{j\in\mathcal L_{k,r}}x_{k,j}
   \geq\sum_{C\in\mathcal C_k}a_r(C)y(C)
 &&k\in\mathcal K_+, r\in\mathcal R_\delta,
 &&\label{eq:configuration-large-capacity}\\
 &&&\sum_{j\in\mathcal S_k}\overline u_{k,j}x_{k,j}
   \geq\sum_{C\in\mathcal C_k}\mu(C)y(C)
 &&k\in\mathcal K_+,
 &&\label{eq:configuration-small-capacity}\\
 &&&z_k+\sum_{j\in\mathcal J_k^+}x_{k,j}=n_k
 &&k\in\mathcal K_0,
 &&\label{eq:configuration-agent-count-zero}\\
 &&&x_{k,j}=0
 &&k\in[K],\ j\in[q]\text{ with }u_{k,j}=0,
 &&\label{eq:configuration-zero-value}\\
 &&&x_{k,j}\in\mathbb Z_{\geq0}
 &&k\in[K],\ j\in[q],
 &&\label{eq:configuration-x-domain}\\
 &&&y(C)\in\mathbb Z_{\geq0}
 &&C\in\mathcal C,
 &&\label{eq:configuration-y-domain}\\
 &&&z_k\in\mathbb Z_{\geq0}
 &&k\in\mathcal K_0.
 &&\label{eq:configuration-z-domain}
\end{empheq}
\end{subequations}
We denote this exact integer formulation by
$\operatorname{ConfIP}(\widehat{\boldsymbol\tau})$. It is used only for
analysis; the algorithm solves the rational formulation introduced
below.

Constraint~\eqref{eq:configuration-item-capacity} limits the reserved
supply. Constraint~\eqref{eq:configuration-agent-count-positive}
represents every positive-threshold agent either by a regular
configuration or by a huge singleton. Constraints
\eqref{eq:configuration-large-capacity} and
\eqref{eq:configuration-small-capacity} reserve the resources required
by the regular configurations. Constraint
\eqref{eq:configuration-agent-count-zero} represents every
zero-threshold agent either by a zero configuration or by an exact
positive singleton. Constraint~\eqref{eq:configuration-zero-value}
discards reservations that provide no utility. Some goods may remain
unreserved and will be assigned after reconstruction.

Let $\mathcal F(\widehat{\boldsymbol\tau})$ be the set of nonnegative
real vectors $(x,y,z)$ satisfying
\eqref{eq:configuration-item-capacity},
\eqref{eq:configuration-agent-count-positive},
\eqref{eq:configuration-large-capacity},
\eqref{eq:configuration-small-capacity},
\eqref{eq:configuration-agent-count-zero}, and
\eqref{eq:configuration-zero-value}. Let
$\mathcal F_{\mathrm{IP}}(\widehat{\boldsymbol\tau})$ be the subset
also satisfying \eqref{eq:configuration-x-domain}--%
\eqref{eq:configuration-z-domain}.

For $0<p<1$, a zero configuration contributes
$\phi_p(0)=0$ and is therefore omitted from
\eqref{eq:configuration-objective}. For $p\leq0$, threshold vectors
having a zero coordinate have already been discarded. Hence every
represented utility appearing in \eqref{eq:configuration-objective}
is positive, and $\phi_p$ is well-defined.

The normalization factors $\alpha_k$ cannot in general be omitted.
For $p\neq0$, the coefficient of a regular configuration contains the
factor $\alpha_{k(C)}^p$, which may differ between agent types. For
$p=0$, $ \log v(C)
 =
 \log\alpha_{k(C)}+\log w(C).$ 
By constraint
\eqref{eq:configuration-agent-count-positive}, exactly $n_k$
represented agents belong to every positive-threshold type $k$.
Thus the terms involving $\log\alpha_k$ sum to the constant
$n_k\log\alpha_k$ for that type. We nevertheless retain these terms
so that $v(C)$ always denotes utility on the original scale.

\begin{lemma}[Representation of the canonical optimum]
\label{lem:configuration-representation}
For a correct threshold vector, the canonical optimum $x^*$ induces
a feasible solution of $\operatorname{ConfIP}
(\widehat{\boldsymbol\tau})$. For every agent having
positive utility in $x^*$, its represented utility $\widetilde u_i$
satisfies
\begin{equation}
 (1-4\delta)u_k(x_i^*)
 \leq
 \widetilde u_i
 \leq
 u_k(x_i^*),
 \qquad i\in A_k.
 \label{eq:configuration-represented-utility}
\end{equation}
\end{lemma}

\begin{proof}
Fix a correct threshold vector. For every $k,j$, set
\begin{equation}
 x_{k,j}:=
 \begin{cases}
  \sum_{i\in A_k}x_{i,j}^*,&u_{k,j}>0,\\
  0,&u_{k,j}=0.
 \end{cases}
 \label{eq:induced-x-definition}
\end{equation}
For $k\in\mathcal K_+$, let $I_k^{\mathrm{reg}}\subseteq A_k$ be the
type-$k$ agents represented by regular configurations in
Lemma~\ref{lem:threshold-bundle-representation}, and use the
configurations $C_i$ supplied there. Set $ y(C):=|\{i\in I_{k(C)}^{\mathrm{reg}}:C_i=C\}|.$ 
For $k\in\mathcal K_0$, set $ z_k:=|R_k|.$ 
%
All these variables are nonnegative integers. Ignoring zero-value
assignments can only decrease item use, so the item-capacity and
zero-value constraints hold. For a positive-threshold type, every agent
outside $I_k^{\mathrm{reg}}$ receives exactly one huge good, while no
regular agent receives a huge good. This proves the corresponding
agent-count equality. For a zero-threshold type, the positive-value
goods are exactly the singleton goods of $H_k$, and
$z_k=|R_k|=n_k-|H_k|$, proving its agent-count equality.

Every large or small good assigned to a positive-threshold type belongs
to a regular agent's bundle. Hence, for every $k\in\mathcal K_+$ and
$r\in\mathcal R_\delta$,
\[
 \sum_{j\in\mathcal L_{k,r}}x_{k,j}
 =\sum_{i\in I_k^{\mathrm{reg}}}a_r(C_i)
 =\sum_{C\in\mathcal C_k}a_r(C)y(C).
\]
Writing
$S_i:=\sum_{j\in\mathcal S_k}\overline u_{k,j}x_{i,j}^*$,
we also have
\[
 \sum_{C\in\mathcal C_k}\mu(C)y(C)
 =\sum_{i\in I_k^{\mathrm{reg}}}\lfloor S_i\rfloor
 \leq\sum_{j\in\mathcal S_k}\overline u_{k,j}x_{k,j}.
\]
These are the large- and small-capacity constraints, completing the
feasibility verification.

The represented-utility bound follows from
Lemma~\ref{lem:threshold-bundle-representation} for positive-threshold
types. For a positive singleton of a zero-threshold type,
$\widetilde u_i=u_k(x_i^*)$, so the same bound holds.
\end{proof}

\subsection{Solving and rounding the configuration formulation}
\label{sec:solving-configuration-formulation}

\medskip
\noindent\textbf{Directed rational objective coefficients.}
Apply Lemma~\ref{lem:directed-rational-approximation} with accuracy
$\delta$ to each positive rational represented value $v$. Write
$\widetilde c_p(v):=\widetilde\phi_{p,\delta}(v)$ for $0\leq p<1$ and
$\widetilde c_p(v):=-\widetilde\phi_{p,\delta}(v)$ for $p<0$. Thus
\begin{equation}
 \frac{v^p}{1+\delta}
 \leq
 \widetilde c_p(v)
 \leq
 v^p
 \qquad\text{if }0<p<1,
 \label{eq:approx1}
\end{equation}
and
\begin{equation}
 v^p
 \leq
 \widetilde c_p(v)
 \leq
 (1+\delta)v^p
 \qquad\text{if }p<0.
  \label{eq:approx2}
\end{equation}
For $p=0$, the same construction gives $ \log v-\delta
 \leq
 \widetilde c_0(v)
 \leq
 \log v.$ 
When $0<p<1$, we also put $\widetilde c_p(0)=0$.

The logarithmic bound follows from
$\log(1+\delta)\leq\delta$. The lemma computes all these coefficients
in polynomial time and with polynomial encoding length.

Recall that $v(C)=\alpha_{k(C)}w(C)$ is the represented value of a regular
configuration $C$, and that $v_{k,j}$ is the represented value of a
singleton pair $(k,j)\in\mathcal I_{\mathrm{sing}}$. Define the signed
rational coefficient
\[
 \widetilde\phi_p(v)
 :=
 \begin{cases}
  \widetilde c_p(v),&0\leq p<1,\\
  -\widetilde c_p(v),&p<0,
\end{cases}
 =\widetilde\phi_{p,\delta}(v),
\]
and the rational transformed objective
\[
 \widetilde\Psi_p(x,y)
 :=
 \sum_{C\in\mathcal C}
 \widetilde\phi_p\bigl(v(C)\bigr)y(C)
 +
 \sum_{(k,j)\in\mathcal I_{\mathrm{sing}}}
 \widetilde\phi_p(v_{k,j})x_{k,j}.
\]
Thus every approximated objective coefficient is a rational number of
polynomial encoding length. The sign in the definition of
$\widetilde\phi_p$ permits every formulation below to be written as a
maximization problem.

Fix an enumerated threshold vector $\widehat{\boldsymbol\tau}
 :=
 (\widehat\tau_1,\ldots,\widehat\tau_K).$ 
Define the rationalized configuration program
$\widetilde{\operatorname{ConfIP}}(\widehat{\boldsymbol\tau})$ by
maximizing $\widetilde\Psi_p(x,y)$ over
$(x,y,z)\in\mathcal F_{\mathrm{IP}}(\widehat{\boldsymbol\tau})$.
It has the same integral feasible solutions as
$\operatorname{ConfIP}(\widehat{\boldsymbol\tau})$ and differs only in
its objective coefficients.

\medskip
\noindent\textbf{Guessing the support of the small-good assignment.}
Let $D:=K(B_\delta+2)$, $F:=2D,$ and $d:=\left\lceil\frac{F}{\delta}\right\rceil.$ 
Here, $D\ge |\mathcal K_0|+(B_\delta+2)|\mathcal K_+|$ is an upper
bound on the total number of non-item constraints
\eqref{eq:configuration-agent-count-positive},
\eqref{eq:configuration-large-capacity},
\eqref{eq:configuration-small-capacity}, and
\eqref{eq:configuration-agent-count-zero}:
for each positive-threshold type, there is exactly one agent-count
constraint, $B_\delta$ large-class constraints, and one weighted
small-value constraint. For each zero-threshold type, there is exactly
one agent-count constraint and no large-class or weighted small-value
constraints.

For every $k\in\mathcal K_+$, enumerate every set
$J_k\subseteq\mathcal S_k$ satisfying $|J_k|\leq d$. If $|J_k|<d$,
interpret $J_k$ as the complete support of the type-$k$ small-good
assignment. The empty set is allowed. If $|J_k|=d$, interpret $J_k$ as
the $d$ highest-valued small good types in that support and define $\gamma_k:=\min_{j\in J_k}\overline u_{k,j}.$ 
In the second case, unguessed types of rounded value at most $\gamma_k$
may still be used. Define the set of forbidden small-good types by
\[
 \mathcal D_k(J_k)
 :=
 \begin{cases}
  \mathcal S_k\setminus J_k,
  &|J_k|<d,\\[1mm]
  \{j\in\mathcal S_k\setminus J_k:
    \overline u_{k,j}>\gamma_k\},
  &|J_k|=d.
 \end{cases}
\]

Fix a support tuple
$\mathbf J=(J_k)_{k\in\mathcal K_+}$ and put $\mathcal I_{\mathbf J}^x
 :=
 \{(k,j):k\in\mathcal K_+,\ j\in J_k\}.$ 
The associated restricted mixed-integer program is
\begin{subequations}
\label{eq:restricted-configuration-mip}
\begin{empheq}[
 left={\textnormal{($\widetilde{\mathrm{Conf\text{-}MIP}}$)}\quad
       \empheqlbrace}
]{align}
 &\underset{x,y,z}{\operatorname{maximize}}
 &&\widetilde\Psi_p(x,y)
 &&\label{eq:restricted-configuration-objective}\\[-1mm]
 &\text{subject to:}\notag\\[-1mm]
 &&&(x,y,z)\in\mathcal F(\widehat{\boldsymbol\tau})
 &&
 &&\label{eq:restricted-configuration-base}\\
 &&&x_{k,j}\geq1
 &&k\in\mathcal K_+,\ j\in J_k,
 &&\label{eq:restricted-configuration-protected}\\
 &&&x_{k,j}=0
 &&k\in\mathcal K_+,\ j\in\mathcal D_k(J_k),
 &&\label{eq:restricted-configuration-forbidden}\\
 &&&y(C)\in\mathbb Z_{\geq0}
 &&C\in\mathcal C,
 &&\label{eq:restricted-configuration-y-domain}\\
 &&&z_k\in\mathbb Z_{\geq0}
 &&k\in\mathcal K_0,
 &&\label{eq:restricted-configuration-z-domain}\\
 &&&x_{k,j}\in\mathbb Z_{\geq0}
 &&(k,j)\in\mathcal I_{\mathbf J}^x,
 &&\label{eq:restricted-configuration-integer-x}\\
 &&&x_{k,j}\in\mathbb R_{\geq0}
 &&(k,j)\in([K]\times[q])\setminus\mathcal I_{\mathbf J}^x.
 &&\label{eq:restricted-configuration-continuous-x}
\end{empheq}
\end{subequations}
We denote this program by
$\widetilde{\operatorname{ConfMIP}}
(\widehat{\boldsymbol\tau},\mathbf J)$. Variables fixed by
\eqref{eq:configuration-zero-value} or
\eqref{eq:restricted-configuration-forbidden} may be eliminated.

For each $k$, the number of possible sets $J_k$ is at most $\binom{q}{d}
 +  \sum_{h=0}^{d-1}\binom{q}{h}
 \leq  (q+1)^d.$ 
Consequently, the total number of restricted programs is at most $ (q+1)^{Kd}
 =
 q^{O_{K,\delta}(1)},$
 which is polynomial for fixed $K$ and $\delta$. Moreover, every restricted
program has at most $ |\mathcal C|+K+Kd
 =
 O_{K,\delta}(1)$
 integer variables. It can therefore be solved in polynomial time by
fixed-dimensional mixed-integer linear programming
\cite{len:j:integer-linear-programming}.

For a feasible restricted program, let $\boldsymbol\xi$ be the values of
all variables required to be integral in
\eqref{eq:restricted-configuration-y-domain}--%
\eqref{eq:restricted-configuration-integer-x} at an optimal solution.
After fixing these variables to $\boldsymbol\xi$, let
$\widetilde{\operatorname{ConfLP}}
(\widehat{\boldsymbol\tau},\mathbf J;\boldsymbol\xi)$ denote the
remaining linear program, with the constraints and rational objective
of~\eqref{eq:restricted-configuration-mip}. Choose an optimal basic
feasible solution of this program.

The following lemma establishes the required rounding property.

\begin{lemma}[Rounding the aggregate assignment]
\label{lem:assignment-rounding}
Fix an arbitrary enumerated threshold vector
$\widehat{\boldsymbol\tau}$.

If $\widetilde{\operatorname{ConfIP}}
(\widehat{\boldsymbol\tau})$ is feasible, then among the restricted
mixed-integer programs associated with the tuples
$\mathbf J=(J_k)_{k\in\mathcal K_+}$, at least one contains an optimal
solution of $\widetilde{\operatorname{ConfIP}}
(\widehat{\boldsymbol\tau})$.

For any feasible restricted mixed-integer program, the following
rounding assertion holds. Solve that program
optimally. Let $(\widehat x,\widehat y,\widehat z)$ be an optimal basic
feasible solution of the associated fixed-variable program
$\widetilde{\operatorname{ConfLP}}
(\widehat{\boldsymbol\tau},\mathbf J;\boldsymbol\xi)$.

Then one can compute in polynomial time a vector
$\overline x\in\mathbb Z_{\geq0}^{K\times q}$ such that
$(\overline x,\widehat y,\widehat z)$ has the following properties.

\begin{enumerate}
\item
The vector $(\overline x,\widehat y,\widehat z)$ satisfies
\eqref{eq:configuration-item-capacity},
\eqref{eq:configuration-agent-count-positive},
\eqref{eq:configuration-large-capacity},
\eqref{eq:configuration-agent-count-zero}, and
\eqref{eq:configuration-zero-value}.

\item
The singleton contribution to the rational objective satisfies
\[
 \sum_{(k,j)\in\mathcal I_{\mathrm{sing}}}
 \widetilde\phi_p(v_{k,j})\overline x_{k,j}
 \geq
 \sum_{(k,j)\in\mathcal I_{\mathrm{sing}}}
 \widetilde\phi_p(v_{k,j})\widehat x_{k,j}.
\]
Since the variables $\widehat y(C)$ and $\widehat z_k$ do not change
during the rounding, $\widetilde\Psi_p$ does not decrease.

\item
For every $k\in\mathcal K_+$, define $\widehat S_k
 :=
 \sum_{j\in\mathcal S_k}
 \overline u_{k,j}\widehat x_{k,j}.$ 
Then
\begin{equation}
\label{eq:small-rounding-guarantee}
 \sum_{j\in\mathcal S_k}
 \overline u_{k,j}\overline x_{k,j}
 \geq
 (1-\delta)\widehat S_k
 \geq
 (1-\delta)
 \sum_{C\in\mathcal C_k}\mu(C)\widehat y(C).
\end{equation}
\end{enumerate}
\end{lemma}

\begin{proof}
Suppose first that
$\widetilde{\operatorname{ConfIP}}(\widehat{\boldsymbol\tau})$ is feasible,
and let $(x^\circ,y^\circ,z^\circ)$ be an optimal solution of this program.
We show that one of the restricted mixed-integer programs retains this
solution.
For every
$k\in\mathcal K_+$, define the support of its small-good assignment by $T_k^\circ
 :=
 \{j\in\mathcal S_k:x_{k,j}^\circ>0\}.$ If $|T_k^\circ|<d$, choose $J_k=T_k^\circ$. The restrictions associated
with $J_k$ require $x_{k,j}$ to be a positive integer for every
$j\in J_k$ and set $x_{k,j}=0$ for every
$j\in\mathcal S_k\setminus J_k$. These restrictions are satisfied by
$x^\circ$.

Suppose instead that $|T_k^\circ|\geq d$. Choose $J_k$ to be the set of
the $d$ highest-valued members of $T_k^\circ$, breaking ties
arbitrarily, and define $\gamma_k
 :=
 \min_{j\in J_k}\overline u_{k,j}.$
 Every good type in $T_k^\circ\setminus J_k$ has rounded value at most
$\gamma_k$. Therefore, $x^\circ$ satisfies $x_{k,j}=0 $ for every $ j\in\mathcal S_k\setminus J_k$
such that $\overline u_{k,j}>\gamma_k$.  
Moreover, $x_{k,j}^\circ\geq1$ for every $j\in J_k$.
 It follows that the restricted program associated with the tuple
$\mathbf J=(J_k)_{k\in\mathcal K_+}$ contains
$(x^\circ,y^\circ,z^\circ)$ as a feasible solution. Thus at least one
of the enumerated restricted programs contains an optimal solution of
$\widetilde{\operatorname{ConfIP}}(\widehat{\boldsymbol\tau})$.

We now prove the rounding assertion for any feasible restricted program.
Let $(\widehat x,\widehat y,\widehat z)$ be the solution specified in
the statement of the lemma. Thus, $\widehat y$, $\widehat z$, and the
variables $\widehat x_{k,j}$ declared integral by the restrictions
associated with $\mathbf J$ have been fixed at their optimal integral
values. The remaining coordinates of $\widehat x$ form an optimal basic
feasible solution of the resulting linear program.

Let $\mathcal V$ be the set of indices $(k,j)$ corresponding to these
remaining continuous variables. For every $(k,j)\in\mathcal V$,
introduce a residual variable $f_{k,j}$ by $ x_{k,j}
 =
 \lfloor\widehat x_{k,j}\rfloor+f_{k,j},$ $0\leq f_{k,j}\leq1.$ 
Equivalently, we intersect the feasible region with $\lfloor\widehat x_{k,j}\rfloor
 \leq
 x_{k,j}
 \leq
 \lfloor\widehat x_{k,j}\rfloor+1$
for every $(k,j)\in\mathcal V.$ 
These additional bounds contain $\widehat x$. Consequently,
$\widehat x$ remains optimal and remains a basic feasible solution of
the restricted linear program.
 At $\widehat x$, the value of the residual variable is $\widehat f_{k,j}
 :=
 \widehat x_{k,j}
 -
 \lfloor\widehat x_{k,j}\rfloor,$  $0\leq\widehat f_{k,j}<1$.
 Define $ \mathcal I_{\mathrm{frac}}
 :=
 \{(k,j)\in\mathcal V:0<\widehat f_{k,j}<1\},$ $E:=|\mathcal I_{\mathrm{frac}}|.$ We next bound the number of fractional components of $\widehat x$.
\begin{claim}
\label{clm:fractional-assignment-bound}
The basic feasible solution $\widehat x$ has at most $F=2D$
fractional components. Equivalently, $E\leq2D=F.$
\end{claim}

\begin{proof}
Fix every residual variable satisfying $\widehat f_{k,j}=0$ at zero and
consider the resulting system on the variables indexed by
$\mathcal I_{\mathrm{frac}}$. The vector
$(\widehat f_{k,j})_{(k,j)\in\mathcal I_{\mathrm{frac}}}$ remains a basic feasible
solution of this system. None of these variables is fixed by a bound
constraint because $0<\widehat f_{k,j}<1$ for every $(k,j)\in\mathcal I_{\mathrm{frac}}.$

Let $J$ be the number of tight item-capacity constraints containing at
least one variable indexed by $\mathcal I_{\mathrm{frac}}$. Besides these constraints,
there are at most $D$ non-item constraints. Since the active constraints
must have rank at least $E$ on the variables indexed by $\mathcal I_{\mathrm{frac}}$,
we obtain $ E\leq J+D.$

Every tight item-capacity constraint counted by $J$ contains at least
two fractional variables. Indeed, after the fixed integral values and
the integral floors have been subtracted, its residual right-hand side
is an integer. If the constraint contained exactly one fractional
variable, tightness would force that variable to have an integral value,
contradicting the definition of $\mathcal I_{\mathrm{frac}}$.

Every variable $f_{k,j}$ occurs in exactly one item-capacity constraint,
namely the constraint associated with good type $j$. Counting incidences
between the $J$ tight item-capacity constraints and the $E$ fractional
variables therefore gives $2J\leq E.$ 
Combining the preceding inequalities yields $ E
 \leq
 J+D
 \leq
 \frac E2+D.$ 
Hence $ E\leq2D=F,$ which proves the claim.
\end{proof}

We now define the reduced system used to round the components indexed by
$\mathcal I_{\mathrm{frac}}$. For every $(k,j)\in[K]\times[q]$, define its fixed
integral part by
\[
 x_{k,j}^{\mathrm{int}}
 :=
 \begin{cases}
 \widehat x_{k,j},
   &\text{if $x_{k,j}$ was fixed before the residual system was
     constructed},\\
 \lfloor\widehat x_{k,j}\rfloor,
   &\text{if $(k,j)\in\mathcal V$}.
 \end{cases}
\]
The first case includes variables declared integral by the restrictions
associated with $\mathbf J$ and variables fixed to zero. Thus
$x_{k,j}^{\mathrm{int}}$ is an integer for every $(k,j)$. Moreover,
$ \widehat x_{k,j}
 =
 x_{k,j}^{\mathrm{int}}+\widehat f_{k,j}$ 
for every $(k,j)\in\mathcal I_{\mathrm{frac}}$, 
whereas
$ \widehat x_{k,j}
 =
 x_{k,j}^{\mathrm{int}}
$ for every $(k,j)\notin\mathcal I_{\mathrm{frac}}$.
 For every good type $j\in[q]$, define the residual item capacity by $M_j^{\mathrm{res}}
 :=
 m_j-\sum_{k=1}^Kx_{k,j}^{\mathrm{int}}.$ 
For every $k\in\mathcal K_+$, define the residual agent-count requirement
by
\[
 N_k^{\mathrm{res}}
 :=
 n_k
 -
 \sum_{C\in\mathcal C_k}\widehat y(C)
 -
 \sum_{j\in\mathcal J_k^{\mathrm{huge}}}
 x_{k,j}^{\mathrm{int}}.
\]
For every $k\in\mathcal K_0$, define
\[
 N_k^{\mathrm{res}}
 :=
 n_k
 -
 \widehat z_k
 -
 \sum_{j\in\mathcal J_k^+}x_{k,j}^{\mathrm{int}}.
\]
Finally, for every $k\in\mathcal K_+$ and
$r\in\mathcal R_\delta$, define the residual large-class requirement by
\[
 L_{k,r}^{\mathrm{res}}
 :=
 \sum_{C\in\mathcal C_k}a_r(C)\widehat y(C)
 -
 \sum_{j\in\mathcal L_{k,r}}x_{k,j}^{\mathrm{int}}.
\]
All these residual quantities are integers.

Remove the weighted small-value constraints and fix every residual
variable outside $\mathcal I_{\mathrm{frac}}$ at zero. The remaining variables are
$f=(f_{k,j})_{(k,j)\in\mathcal I_{\mathrm{frac}}}$. We solve the following linear
program.
\begin{subequations}
\label{eq:reduced-assignment-system}
\begin{empheq}[
 left={\textnormal{(Round-LP)}\quad\empheqlbrace}
]{align}
 &\underset{f}{\operatorname{maximize}}
 &&\sum_{(k,j)\in\mathcal I_{\mathrm{frac}}\cap\mathcal I_{\mathrm{sing}}}
   \widetilde\phi_p(v_{k,j})f_{k,j}
 &&\label{eq:rounding-lp-objective}\\[-1mm]
 &\text{subject to:}\notag\\[-1mm]
 &&&\sum_{k:(k,j)\in\mathcal I_{\mathrm{frac}}}f_{k,j}
 \leq M_j^{\mathrm{res}}
 &&j\in[q]
 &&\label{eq:rounding-item-capacity}\\
 &&&\sum_{\substack{j\in\mathcal J_k^{\mathrm{huge}}\\
                     (k,j)\in\mathcal I_{\mathrm{frac}}}}f_{k,j}
 =N_k^{\mathrm{res}}
 &&k\in\mathcal K_+
 &&\label{eq:rounding-positive-count}\\
 &&&\sum_{\substack{j\in\mathcal J_k^+\\
                     (k,j)\in\mathcal I_{\mathrm{frac}}}}f_{k,j}
 =N_k^{\mathrm{res}}
 &&k\in\mathcal K_0
 &&\label{eq:rounding-zero-count}\\
 &&&\sum_{\substack{j\in\mathcal L_{k,r}\\
                     (k,j)\in\mathcal I_{\mathrm{frac}}}}f_{k,j}
 \geq L_{k,r}^{\mathrm{res}}
 &&k\in\mathcal K_+,\ r\in\mathcal R_\delta
 &&\label{eq:rounding-large-capacity}\\
 &&&0\leq f_{k,j}\leq1
 &&(k,j)\in\mathcal I_{\mathrm{frac}}
 &&\label{eq:rounding-bounds}
\end{empheq}
\end{subequations}
Let $\mathcal P_{\mathrm{red}}$ denote the feasible region of
\textnormal{(Round-LP)}.

The vector
$
 \widehat f
 :=
 (\widehat f_{k,j})_{(k,j)\in\mathcal I_{\mathrm{frac}}}
$
belongs to $\mathcal P_{\mathrm{red}}$. Indeed, it is obtained from
$\widehat x$ by subtracting the fixed integral parts. Under the residual
definitions above, constraints
\eqref{eq:configuration-item-capacity},
\eqref{eq:configuration-agent-count-positive},
\eqref{eq:configuration-agent-count-zero}, and
\eqref{eq:configuration-large-capacity} become, respectively,
\eqref{eq:rounding-item-capacity}--%
\eqref{eq:rounding-large-capacity}. The only structural constraints
omitted from \textnormal{(Round-LP)} are the weighted small-value
constraints~\eqref{eq:configuration-small-capacity}.

We next verify that the polytope $\mathcal P_{\mathrm{red}}$ is integral. Every
variable $f_{k,j}$ occurs in its item-capacity constraint and in at most
one additional constraint. More precisely, a huge-singleton variable
for a positive-threshold type occurs in one agent-count constraint, a
positive-singleton variable for a zero-threshold type occurs in one
agent-count constraint, and a large-good variable occurs in one
large-class constraint. A small-good variable occurs in no additional
constraint because the weighted small-value constraints have been
removed.

After multiplying every agent-count and large-class row by $-1$, each
column contains at most one entry equal to $1$ and at most one entry
equal to $-1$. The resulting matrix is a submatrix of a node-arc
incidence matrix and is therefore totally unimodular. Since the
right-hand sides and variable bounds in
\eqref{eq:rounding-item-capacity}--\eqref{eq:rounding-bounds} are integral,
$\mathcal P_{\mathrm{red}}$ is an integral polytope.

Since $\widehat f\in\mathcal P_{\mathrm{red}}$ and
$\mathcal P_{\mathrm{red}}$ is a nonempty bounded integral polytope,
\textnormal{(Round-LP)} has an optimal solution $ g\in\mathcal P_{\mathrm{red}}\cap\{0,1\}^{\mathcal I_{\mathrm{frac}}}.$ 
By optimality, the value of
\eqref{eq:rounding-lp-objective} at $g$ is at least its value at
$\widehat f$.

Define
\[
 \overline x_{k,j}
 :=
 \begin{cases}
 x_{k,j}^{\mathrm{int}}+g_{k,j},
   &(k,j)\in\mathcal I_{\mathrm{frac}},\\
 x_{k,j}^{\mathrm{int}},
   &(k,j)\notin\mathcal I_{\mathrm{frac}}.
 \end{cases}
\]
By the residual definitions,
\eqref{eq:rounding-item-capacity}--%
\eqref{eq:rounding-large-capacity} imply that $\overline x$ satisfies
\eqref{eq:configuration-item-capacity},
\eqref{eq:configuration-agent-count-positive},
\eqref{eq:configuration-agent-count-zero}, and
\eqref{eq:configuration-large-capacity}. Variables constrained by
\eqref{eq:configuration-zero-value} were fixed at zero and remain zero.
We next prove item~2. For every
$(k,j)\notin\mathcal I_{\mathrm{frac}}$, the construction gives
$\overline x_{k,j}=\widehat x_{k,j}$. For every
$(k,j)\in\mathcal I_{\mathrm{frac}}$, we have $ \overline x_{k,j}-\widehat x_{k,j}
 =
 g_{k,j}-\widehat f_{k,j}.$ 
Consequently,
\begin{align*}
 &\sum_{(k,j)\in\mathcal I_{\mathrm{sing}}}
 \widetilde\phi_p(v_{k,j})
 \bigl(\overline x_{k,j}-\widehat x_{k,j}\bigr)=
 \sum_{(k,j)\in\mathcal I_{\mathrm{frac}}\cap\mathcal I_{\mathrm{sing}}}
 \widetilde\phi_p(v_{k,j})
 \bigl(g_{k,j}-\widehat f_{k,j}\bigr).
\end{align*}
Both $g$ and $\widehat f$ belong to $\mathcal P_{\mathrm{red}}$.
By the choice of $g$, the right-hand side is nonnegative. This proves
item~2. It remains to prove the small-good guarantee. Fix
$k\in\mathcal K_+$ and recall that $\widehat S_k
 =
 \sum_{j\in\mathcal S_k}
 \overline u_{k,j}\widehat x_{k,j}.$

Suppose first that $|J_k|=d$. Every variable $x_{k,j}$ with
$j\in J_k$ was declared integral and was therefore fixed before
$\widehat x$ was chosen. Moreover, every unguessed small-good variable
with rounded value greater than $ \gamma_k
 =
 \min_{j\in J_k}\overline u_{k,j}$ 
was fixed to zero. Consequently, every small-good coordinate of type
$k$ that belongs to $\mathcal I_{\mathrm{frac}}$ has rounded value at most $\gamma_k$.

By Claim~\ref{clm:fractional-assignment-bound}, at most $F$ coordinates
can change. Each such coordinate can decrease by less than one.
Therefore, $\sum_{j\in\mathcal S_k}
 \overline u_{k,j}\overline x_{k,j}
 \geq
 \widehat S_k-F\gamma_k.$ 
On the other hand, the $d$ protected variables satisfy
$\widehat x_{k,j}\geq1$ for every $j\in J_k$. Hence $\widehat S_k
 \geq
 \sum_{j\in J_k}
 \overline u_{k,j}\widehat x_{k,j}
 \geq
 d\gamma_k.$ 
It follows that $ F\gamma_k
 \leq
 \frac Fd\widehat S_k
 \leq
 \delta\widehat S_k,$ 
where the last inequality follows from
$d=\lceil F/\delta\rceil$. Therefore, $\sum_{j\in\mathcal S_k}
 \overline u_{k,j}\overline x_{k,j}
 \geq
 (1-\delta)\widehat S_k.$

Suppose now that $|J_k|<d$. In this case, $J_k$ is interpreted as the
complete support of the small-good assignment. Every variable
$x_{k,j}$ with $j\in J_k$ was declared integral, and every variable
$x_{k,j}$ with $j\in\mathcal S_k\setminus J_k$ was fixed to zero.
Therefore, no small-good coordinate of type $k$ belongs to
$\mathcal I_{\mathrm{frac}}$, and $\sum_{j\in\mathcal S_k}
 \overline u_{k,j}\overline x_{k,j}
 =
 \widehat S_k.$

Finally, constraint~\eqref{eq:configuration-small-capacity} implies $ \widehat S_k
 \geq
 \sum_{C\in\mathcal C_k}
 \mu(C)\widehat y(C).$
 Combining this inequality with the preceding two cases gives
\[
 \sum_{j\in\mathcal S_k}
 \overline u_{k,j}\overline x_{k,j}
 \geq
 (1-\delta)\widehat S_k
 \geq
 (1-\delta)
 \sum_{C\in\mathcal C_k}\mu(C)\widehat y(C).
\]
This proves item~3 and completes the proof.
\end{proof}
\medskip
\noindent\textbf{Loss due to the rational objective coefficients.}
We next quantify the loss caused by replacing the exact objective
coefficients by their directed rational approximations. Fix a threshold
vector $\widehat{\boldsymbol\tau}$, and let $Z=(x,y,z)$ be any feasible
solution of $\operatorname{ConfIP}(\widehat{\boldsymbol\tau})$. Let
$Z^\circ=(x^\circ,y^\circ,z^\circ)$ be an optimal solution of
$\widetilde{\operatorname{ConfIP}}(\widehat{\boldsymbol\tau})$.
By the first conclusion of Lemma~\ref{lem:assignment-rounding}, some
enumerated tuple $\mathbf J^\circ$ gives a restricted program containing
$Z^\circ$. Let
$\widehat Z=(\widehat x,\widehat y,\widehat z)$ be the optimal basic
solution selected from that restricted program, and let
$\overline Z=(\overline x,\widehat y,\widehat z)$ be the rounded integral
vector produced by the lemma. The vector $\overline Z$ satisfies the
constraints
\eqref{eq:configuration-item-capacity},
\eqref{eq:configuration-agent-count-positive},
\eqref{eq:configuration-large-capacity},
\eqref{eq:configuration-agent-count-zero}, and
\eqref{eq:configuration-zero-value}, as well as the weakened small-value
bound~\eqref{eq:small-rounding-guarantee}; it need not satisfy
\eqref{eq:configuration-small-capacity}.

For $p\neq0$, define the exact power objective of a configuration
solution by
\[
 Q_p^{\mathrm{conf}}(Z)
 :=
 \sum_{C\in\mathcal C}v(C)^py(C)
 +
 \sum_{(k,j)\in\mathcal I_{\mathrm{sing}}}
 v_{k,j}^p x_{k,j},
\]
where zero configurations, which can occur only when $0<p<1$, contribute
zero. For $p=0$, define
\[
 \Lambda^{\mathrm{conf}}(Z)
 :=
 \sum_{C\in\mathcal C}\log(v(C))y(C)
 +
 \sum_{(k,j)\in\mathcal I_{\mathrm{sing}}}
 \log(v_{k,j})x_{k,j}.
\]
These are weighted objectives of the form considered in
Lemma~\ref{lem:weighted-rational-comparison}, with weights $y(C)$ and
$x_{k,j}$. For $p\leq0$, all represented utilities are positive.

Write $\widetilde\Psi_p(Z')$ for
$\widetilde\Psi_p(x',y')$ when $Z'=(x',y',z')$. Since the selected
restricted program contains $Z^\circ$, its optimal solution
$\widehat Z$ has rational objective at least that of $Z^\circ$.
Lemma~\ref{lem:assignment-rounding} ensures that rounding
$\widehat Z$ to $\overline Z$ does not decrease this objective.
Together with the optimality of $Z^\circ$ in the rationalized
configuration program, this gives $\widetilde\Psi_p(\overline Z)
 \geq\widetilde\Psi_p(\widehat Z)
 \geq\widetilde\Psi_p(Z^\circ)
 \geq\widetilde\Psi_p(Z).$ 
Apply Lemma~\ref{lem:weighted-rational-comparison} with accuracy
$\delta$ to the weighted lists represented by $Z$ and $\overline Z$.
For $0<p<1$, it gives
\begin{equation}
\label{eq:positive-configuration-comparison}
 Q_p^{\mathrm{conf}}(\overline Z)
 \geq
 \frac{1}{1+\delta}Q_p^{\mathrm{conf}}(Z).
\end{equation}
For $p<0$, it gives
\begin{equation}
\label{eq:negative-configuration-comparison}
 Q_p^{\mathrm{conf}}(\overline Z)
 \leq
 (1+\delta)Q_p^{\mathrm{conf}}(Z).
\end{equation}
Finally, for $p=0$, constraints
\eqref{eq:configuration-agent-count-positive} and
\eqref{eq:configuration-agent-count-zero} imply that the total weight
of represented utilities is $n$, including for fractional configuration
vectors. The logarithmic conclusion of the lemma, together with
$\log(1+\delta)\leq\delta$, therefore gives
\begin{equation}
\label{eq:nash-configuration-comparison}
 \Lambda^{\mathrm{conf}}(\overline Z)
 \geq
 \Lambda^{\mathrm{conf}}(Z)-n\delta.
\end{equation}

\subsection{Reconstructing a compressed allocation}
\label{sec:compressed-reconstruction}

For every feasible restricted mixed-integer program associated with an
enumerated threshold vector, Lemma~\ref{lem:assignment-rounding}
produces an
integral vector $ \overline Z
 :=
 (\overline x,\widehat y,\widehat z).$ 
The vector $\overline x$ specifies how many goods of each type are
reserved for each agent type, but it does not specify their distribution
among the individual agents. We now construct such a distribution in
compressed form.

\begin{lemma}[Compressed reconstruction]
\label{lem:compressed-reconstruction}
Let $\overline Z=(\overline x,\widehat y,\widehat z)$ be an integral
vector satisfying the conclusions of
Lemma~\ref{lem:assignment-rounding}, including
\eqref{eq:small-rounding-guarantee}.
Then $\overline Z$ can be converted in polynomial time into a feasible
compressed allocation $x^{\mathrm{rec}}$ with the following properties.

\begin{enumerate}
\item
Every agent of type $k\in\mathcal K_+$ represented by a regular
configuration $C\in\mathcal C_k$ receives utility at least $(1-6\delta)v(C)
 =
 (1-6\delta)\alpha_k w(C).$
\item
Every agent represented by a singleton pair
$(k,j)\in\mathcal I_{\mathrm{sing}}$ receives one type-$j$ good and,
before any unused goods are assigned, satisfies
\begin{equation}
\label{eq:reconstructed-singleton-utility}
 u_k(x_i^{\mathrm{rec}})\geq v_{k,j}.
\end{equation}

\item
Every agent represented by a zero configuration initially receives no
good having positive value to its type and therefore initially has
utility zero. Assigning unused goods at the end may only increase this
utility.

\item
For every good type $j\in[q]$, at most $m_j$ copies are assigned, and
the number of distinct bundle patterns in $x^{\mathrm{rec}}$ is
$O_{K,\delta}(q)$.
\end{enumerate}
\end{lemma}

\begin{proof}
We first partition the agents of each type according to the integral
variables $\widehat y$, $\widehat z$, and $\overline x$.

Fix $k\in\mathcal K_+$. For every $C\in\mathcal C_k$, choose a set
$A_{k,C}\subseteq A_k$ of cardinality $|A_{k,C}|=\widehat y(C).$ For every $j\in\mathcal J_k^{\mathrm{huge}}$, choose a set
$A_{k,j}^{\mathrm{sing}}\subseteq A_k$ of cardinality $ |A_{k,j}^{\mathrm{sing}}|
 =
 \overline x_{k,j}.$ 
Constraint~\eqref{eq:configuration-agent-count-positive} implies that
these sets can be chosen to form a partition of $A_k$.
 Now fix $k\in\mathcal K_0$. Choose a set
$A_{k,0}\subseteq A_k$ of cardinality $|A_{k,0}|=\widehat z_k,$
and, for every $j\in\mathcal J_k^+$, choose a set
$A_{k,j}^{\mathrm{sing}}\subseteq A_k$ of cardinality $ |A_{k,j}^{\mathrm{sing}}|
 =
 \overline x_{k,j}.$
Constraint~\eqref{eq:configuration-agent-count-zero} implies that these
sets can be chosen to form a partition of $A_k$.
 For every $k\in[K]$ and $j\in[q]$, reserve $\overline x_{k,j}$ copies
of type $j$ for agents of type $k$. By
\eqref{eq:configuration-item-capacity}, these reservations can be chosen
disjointly.
All remaining copies are temporarily left unassigned.
 For every singleton pair
$(k,j)\in\mathcal I_{\mathrm{sing}}$, assign one reserved type-$j$ good
to each agent in $A_{k,j}^{\mathrm{sing}}$. It remains to distribute the
goods reserved for agents represented by regular configurations.

\medskip
\noindent\textbf{Assignment of large goods.}
Fix $k\in\mathcal K_+$ and $r\in\mathcal R_\delta$. Define $D_{k,r}
 :=
 \sum_{C\in\mathcal C_k}a_r(C)\widehat y(C).$
By constraint~\eqref{eq:configuration-large-capacity},
\begin{equation}
\label{eq:large-supply-reconstruction}
 D_{k,r}
 \leq
 \sum_{j\in\mathcal L_{k,r}}\overline x_{k,j}.
\end{equation}
Thus the goods reserved for type $k$ contain sufficiently many goods
from the large-value class $\mathcal L_{k,r}$.

Fix orders $\mathcal C_k=(C_1,\ldots,C_s)$ and $\mathcal L_{k,r}=(j_1,\ldots,j_t).$
 For every configuration $C\in\mathcal C_k$ and every
$h\in\{1,\ldots,a_r(C)\}$, define a slot label $(C,r,h)$. Every agent
in $A_{k,C}$ must receive one class-$r$ good in each such slot. Algorithm~\ref{alg:large-good-assignment}, given below, fills these slots in batches.

\begin{algorithm}[t]
\caption{Batched assignment of class-$r$ large goods to type-$k$ agents}
\label{alg:large-good-assignment}
\begin{algorithmic}[1]
\Require Ordered configurations
         $(C_1,\ldots,C_s)$; multiplicities $\widehat y(C_\nu)$;
         demands $a_r(C_\nu)$; ordered good types
         $(j_1,\ldots,j_t)$; reserved supplies
         $\overline x_{k,j_\ell}$
\Ensure A list $\mathcal B_{k,r}$ of batched assignments filling every
        class-$r$ slot
\State $\mathcal B_{k,r}\gets\emptyset$
\For{$\ell=1,\ldots,t$}
    \State $b_\ell\gets\overline x_{k,j_\ell}$
        \Comment{Remaining reserved supply of type $j_\ell$}
\EndFor
\State $\ell\gets1$
\For{$\nu=1,\ldots,s$}
    \State $N\gets\widehat y(C_\nu)$
    \For{$h=1,\ldots,a_r(C_\nu)$}
        \State $a\gets0$
            \Comment{Number of occurrences of slot $(C_\nu,r,h)$ filled}
        \While{$a<N$}
            \While{$b_\ell=0$}
                \State $\ell\gets\ell+1$
            \EndWhile
            \State $\lambda\gets\min\{N-a,b_\ell\}$
            \State Append
            $(C_\nu,r,h,a+1,\lambda,j_\ell)$
            to $\mathcal B_{k,r}$
            \Comment{Assign type $j_\ell$ to positions
                     $a+1,\ldots,a+\lambda$}
            \State $a\gets a+\lambda$
            \State $b_\ell\gets b_\ell-\lambda$
        \EndWhile
    \EndFor
\EndFor
\State \Return $\mathcal B_{k,r}$
\end{algorithmic}
\end{algorithm}

A tuple
$
 (C,r,h,a+1,\lambda,j)
$
in $\mathcal B_{k,r}$ means that the agents in positions
$a+1,\ldots,a+\lambda$ of the block $A_{k,C}$ receive a type-$j$ good
in their slot $(C,r,h)$. Thus a batch is stored using only a constant
number of binary-encoded integers and indices; the agents in the batch
are not listed individually.

Algorithm~\ref{alg:large-good-assignment} is well-defined. At every
iteration, the total number of unfilled class-$r$ slots is at most the
total remaining reserved supply. This holds initially by
\eqref{eq:large-supply-reconstruction}, and each assignment decreases
both quantities by $\lambda$. Hence,
whenever an unfilled slot remains, the inner loop finds an index
$\ell\leq t$ with $b_\ell>0$.

The algorithm fills every occurrence of every slot label. Therefore,
each agent represented by $C$ receives exactly $a_r(C)$ goods from
$\mathcal L_{k,r}$. Moreover, it uses at most
$\overline x_{k,j}$ copies of each good type $j$.

Each execution of the innermost loop either exhausts one good type or
completes one slot label. Consequently, for fixed $k$ and $r$, the
number of generated batches is at most $ |\mathcal L_{k,r}|
 +
 \sum_{C\in\mathcal C_k}a_r(C).$ 
Summing over all $k\in\mathcal K_+$ and
$r\in\mathcal R_\delta$ gives
\begin{align*}
 &\sum_{k\in\mathcal K_+}
 \sum_{r\in\mathcal R_\delta}
 \left(
 |\mathcal L_{k,r}|
 +
 \sum_{C\in\mathcal C_k}a_r(C)
 \right)\leq
 Kq+
 \sum_{C\in\mathcal C}
 \sum_{r\in\mathcal R_\delta}a_r(C)
 =
 O_{K,\delta}(q).
\end{align*}
The last equality follows because the large-value classes are disjoint
for each agent type, $|\mathcal C|=O_{K,\delta}(1)$, and every retained
configuration contains at most $4/\delta$ large goods. Hence the
large-good assignment is computed in polynomial time and has a
polynomial-size compressed representation.

\medskip
\noindent\textbf{Assignment of small goods.}
Fix $k\in\mathcal K_+$. For every regular configuration
$C\in\mathcal C_k$, define its small-good target by $ t(C):=(1-\delta)\mu(C).$ 
The weakened small-good constraint gives
\begin{equation}
\label{eq:small-good-supply-reconstruction}
 \sum_{j\in\mathcal S_k}
 \overline u_{k,j}\overline x_{k,j}
 \geq
 \sum_{C\in\mathcal C_k}
 t(C)\widehat y(C).
\end{equation}

Fix arbitrary orders of the configurations in $\mathcal C_k$ and of
the small good types in $\mathcal S_k$. Within each block $A_{k,C}$,
fix a conceptual ordering of its agents; assignment batches always
refer to consecutive agents in this order. For every
$j\in\mathcal S_k$, let $b_j$ denote the residual number of reserved
type-$j$ goods. Initially, $ b_j:=\overline x_{k,j}.$ 
These residual supplies are maintained throughout the processing of
all configuration blocks of type $k$.

For a positive target $t$, the \emph{boundary type} is the first
positive-supply type at which the cumulative residual rounded value
reaches $t$. The associated \emph{boundary bundle} contains all
residual copies of the preceding types and as many copies of the
boundary type as possible without exceeding $t$. The invariant proved
below shows that this type exists whenever the algorithm requires it.

The assignment is performed in compressed form by
Algorithm~\ref{alg:small-good-assignment}.
\begin{algorithm}[t]
\caption{Batched assignment of small goods to type-$k$ agents}
\label{alg:small-good-assignment}
\begin{algorithmic}[1]
\Require Configuration blocks $A_{k,C}$ and reserved supplies
$\overline x_{k,j}$ for $j\in\mathcal S_k$
\Ensure A compressed partial assignment of the reserved small goods to
all agents represented by regular configurations
\State Initialize $b_j\gets\overline x_{k,j}$ for every
$j\in\mathcal S_k$
\For{each $C\in\mathcal C_k$ in the fixed order}
    \State Set $A'\gets A_{k,C}$ and
    $t\gets(1-\delta)\mu(C)$
    \If{$t=0$}
        \State Assign an empty small-good bundle to every agent in
        $A'$ and set $A'\gets\emptyset$
    \EndIf
    \While{$A'\neq\emptyset$}
        \State List the small good types having positive residual
        supply as $j_1,\ldots,j_\ell$, in the fixed order
        \Comment{The invariant below implies $\ell\geq1$}
        \If{$\overline u_{k,j_1}>t$}
            \State Assign an empty small-good bundle to every agent in
            $A'$ and set $A'\gets\emptyset$
        \Else
            \State Let $j_r$ be the first type for which the total
            residual value of $j_1,\ldots,j_r$ is at least $t$
            \State Form a boundary bundle from all residual copies of
            $j_1,\ldots,j_{r-1}$ and as many copies of $j_r$ as fit
            without exceeding $t$
            \State Assign the boundary bundle to one agent in $A'$ and
            update the residual supplies and $A'$
            \State Using the remaining copies of $j_r$, assign the
            largest feasible pure type-$j_r$ bundle to a maximum-size
            batch of agents in $A'$
            \State Update the residual supply of $j_r$ and $A'$
        \EndIf
    \EndWhile
\EndFor
\State \Return the recorded assignment batches
\end{algorithmic}
\end{algorithm}

We describe the two assignment steps more precisely. Suppose that
$j_r$ is the boundary type, and consider the residual supplies
immediately before the boundary bundle is formed. Put
\[
 R
 :=
 t-
 \sum_{h=1}^{r-1}
 \overline u_{k,j_h}b_{j_h}
 \qquad\text{and}\qquad
 h_r
 :=
 \left\lfloor
 \frac{R}{\overline u_{k,j_r}}
 \right\rfloor.
\]
The boundary bundle is $P
 :=
 \sum_{h=1}^{r-1}b_{j_h}\mathbf e_{j_h}
 +h_r\mathbf e_{j_r},$ 
where $\mathbf e_j$ denotes a bundle containing one type-$j$ good.
Assign $P$ to one agent in $A'$ and remove that agent from $A'$.
After this assignment, the residual supplies of
$j_1,\ldots,j_{r-1}$ are zero, and $b_{j_r}$ is decreased by $h_r$.

Next, define $q_r
 :=
 \left\lfloor
 \frac{t}{\overline u_{k,j_r}}
 \right\rfloor.$
 If $q_r>0$, assign $q_r$ copies of type $j_r$ to each of $\lambda
 :=
 \min\left\{
 |A'|,
 \left\lfloor
 \frac{b_{j_r}}{q_r}
 \right\rfloor
 \right\}$ 
agents in $A'$. If $\lambda>0$, record these identical assignments as
one batch. Remove their recipients from $A'$ and decrease $b_{j_r}$ by
$\lambda q_r$. If $q_r=0$, no copy of type $j_r$ fits into an empty
bundle. In the next iteration, the remaining agents in the current
configuration block therefore receive empty small-good bundles (that is, they get assigned zero copies of every type
$j\in\mathcal S_k$, but they may still receive the large goods
specified by their regular configurations).

\begin{claim}\label{cl:target}
For every agent $i$ represented by a regular configuration $C_i$, the
rounded normalized value $s_i$ of the small-good bundle assigned to $i$
satisfies $s_i\leq t(C_i)$.
\end{claim}
\begin{proof}
Indeed, if $i$ receives a boundary bundle, then $s_i
 =
 \sum_{h=1}^{r-1}
 \overline u_{k,j_h}b_{j_h}
 +
 h_r\overline u_{k,j_r}.$
 Since $ h_r
 =
 \left\lfloor
 \frac{
 t(C_i)-
 \sum_{h=1}^{r-1}
 \overline u_{k,j_h}b_{j_h}
 }{
 \overline u_{k,j_r}
 }
 \right\rfloor,$
 we have $h_r\overline u_{k,j_r}
 \leq
 t(C_i)-
 \sum_{h=1}^{r-1}
 \overline u_{k,j_h}b_{j_h}.$
 Consequently, $s_i\leq t(C_i)$.

If $i$ belongs to a single-type batch, then it receives $q_r$ copies
of type $j_r$, where $q_r
 =
 \left\lfloor
 \frac{t(C_i)}{\overline u_{k,j_r}}
 \right\rfloor.$
 Thus $ s_i
 =
 q_r\overline u_{k,j_r}
 \leq t(C_i).$
 Finally, an agent receiving an empty small-good bundle has $s_i=0$,
which is at most $t(C_i)$ because $t(C_i)\geq0$. Therefore, $s_i\leq t(C_i)$
for every assigned agent $i$.~\end{proof}

We next verify that the boundary type required by the algorithm always
exists. Consider any stage of the algorithm for the fixed agent type
$k$. Let $\mathcal U_k$ be the set of type-$k$ agents whose small-good
bundles have not yet been assigned. For every $i\in\mathcal U_k$, let
$C_i$ denote the unique configuration $C\in\mathcal C_k$ such that
$i\in A_{k,C}$. Define $ S_{\mathrm{rem}}
 :=
 \sum_{j\in\mathcal S_k}
 \overline u_{k,j}b_j$
 and $T_{\mathrm{rem}}
 :=
 \sum_{i\in\mathcal U_k}t(C_i).$ 
Thus $S_{\mathrm{rem}}$ is the total rounded normalized value of the
reserved small goods that have not yet been assigned, whereas
$T_{\mathrm{rem}}$ is the sum of the small-good targets of all agents
whose small-good bundles have not yet been assigned.

We maintain the invariant $ S_{\mathrm{rem}}\geq T_{\mathrm{rem}}.$ 
Initially, $b_j=\overline x_{k,j}$ for every
$j\in\mathcal S_k$, and $\mathcal U_k$ contains all agents represented
by regular configurations of type $k$. Hence
\begin{align*}
 S_{\mathrm{rem}}
 &=
 \sum_{j\in\mathcal S_k}
 \overline u_{k,j}\overline x_{k,j}\geq
 \sum_{C\in\mathcal C_k}
 t(C)\widehat y(C)=
 \sum_{i\in\mathcal U_k}t(C_i)
 =
 T_{\mathrm{rem}},
\end{align*}
where the inequality follows from
\eqref{eq:small-good-supply-reconstruction}.

Suppose that the algorithm assigns small-good bundles to a set
$B\subseteq\mathcal U_k$ of agents. For every $i\in B$, let $s_i$ be
the rounded normalized value of the bundle assigned to $i$. By
Claim~\ref{cl:target}, $s_i\leq t(C_i)$ for every $i\in B$.
After this assignment, the updated quantities satisfy $S_{\mathrm{rem}}'
 =
 S_{\mathrm{rem}}-\sum_{i\in B}s_i$ and $ T_{\mathrm{rem}}'
 =
 T_{\mathrm{rem}}-\sum_{i\in B}t(C_i).$ 
Consequently, $S_{\mathrm{rem}}'-T_{\mathrm{rem}}'
 =  S_{\mathrm{rem}}-T_{\mathrm{rem}}  +  \sum_{i\in B}\bigl(t(C_i)-s_i\bigr)\geq0.$
 Thus the invariant is preserved after every individual assignment and
every batched assignment.

Suppose that the algorithm is processing an agent with positive target
$t$. Since this agent belongs to $\mathcal U_k$, the invariant gives $ S_{\mathrm{rem}}
 \geq
 T_{\mathrm{rem}}
 \geq t.$ 
Hence the residual small goods have total rounded normalized value at
least $t>0$. In particular, at least one residual small-good type exists.
Therefore, unless its first type already has value greater than $t$, the
smallest prefix whose total value reaches $t$ exists, and the boundary
type $j_r$ is well defined.

It remains to prove the lower bound on the value assigned to each agent.
Consider first an agent receiving a boundary bundle. By the minimality
of $r$, $ 0<R\leq \overline u_{k,j_r}b_{j_r}.$ Moreover, $h_r\overline u_{k,j_r}
 \leq R$ and $ h_r\overline u_{k,j_r}
 >
 R-\overline u_{k,j_r}.$ 
Therefore, the rounded normalized value $s_i$ of the boundary bundle
satisfies $t-\overline u_{k,j_r}
 <
 s_i
 \leq t.$ 
Since $j_r\in\mathcal S_k$, we have
$\overline u_{k,j_r}<1$, and hence $ t-1<s_i\leq t.$

An agent in the subsequent single-type batch receives $q_r$ copies of
type $j_r$. From the definition of $q_r$, $q_r\overline u_{k,j_r}
 \leq t
 <
 (q_r+1)\overline u_{k,j_r}.$ Consequently, $t-\overline u_{k,j_r}
 <
 q_r\overline u_{k,j_r}
 \leq t,$ 
so the value assigned to this agent is also greater than $t-1$.

Finally, suppose that an agent receives an empty small-good bundle.
Then either $t=0$ or the first residual small good has value greater
than $t$. In the latter case, $ t<\overline u_{k,j_1}<1,$ 
and therefore $0>t-1$. Thus every agent $i$ represented by a regular
configuration $C$ receives a small-good bundle of rounded normalized
value $s_i$ satisfying $t(C)-1<s_i\leq t(C)$.
In particular,
\begin{equation}\label{eq:individual-small-guarantee}
 s_i\geq(1-\delta)\mu(C)-1.
\end{equation}

Finally, we bound the number of batches created by the small-good
assignment. Consider one configuration block $A_{k,C}$ and an
iteration whose boundary type is $j_r$. The boundary bundle exhausts
the residual supplies of every type
$j_1,\ldots,j_{r-1}$. It creates one assignment batch, consisting of a
single agent.

After updating the residual supplies, suppose that $q_r>0$. The
algorithm assigns $q_r$ copies of type $j_r$ to $\lambda
 =
 \min\left\{
 |A'|,
 \left\lfloor
 \frac{b_{j_r}}{q_r}
 \right\rfloor
 \right\}$
 agents and records these identical bundles as one additional batch. If
agents remain after this assignment, then $\lambda
 =
 \left\lfloor
 \frac{b_{j_r}}{q_r}
 \right\rfloor.$
 Consequently, the updated residual supply $b_{j_r}'$ satisfies $ 0\leq b_{j_r}'<q_r.$ 
If $b_{j_r}'=0$, type $j_r$ is already exhausted. Otherwise, $ b_{j_r}'\overline u_{k,j_r}
 <
 q_r\overline u_{k,j_r}
 \leq t(C).$ 
Hence, in the next iteration, the residual copies of $j_r$ precede the
new boundary type and are all included in the next boundary bundle.
They are therefore exhausted in that iteration.

It follows that, within a fixed configuration block, every iteration
either exhausts at least one small good type, terminates the block, or
is followed by an iteration that exhausts its boundary type. Moreover,
each iteration creates at most one boundary-bundle batch and one
single-type batch, and at most one batch of empty small-good bundles is
created for the entire block. Therefore, the number of batches created
for $A_{k,C}$ is $O(|\mathcal S_k|+1)$. Summing over all regular
configuration blocks of type $k$ gives $ O\bigl(
 |\mathcal C_k|(|\mathcal S_k|+1)
 \bigr)
 =
 O_{K,\delta}(q),$
 because $|\mathcal C_k|=O_{K,\delta}(1)$ and
$|\mathcal S_k|\leq q$. The number of distinct small-good bundle
patterns is at most the number of recorded batches and is therefore
also $O_{K,\delta}(q)$.

The algorithm manipulates only residual supplies and batch
multiplicities. These quantities are computed by additions,
subtractions, comparisons, and integer divisions applied to
binary-encoded integers and rationals of polynomial encoding length.
In particular, the algorithm does not enumerate individual goods or
agents. It therefore constructs the compressed small-good assignment
in time polynomial in the binary encoding length of the instance.

The preceding procedures define the reconstructed allocation as follows.
Within every block $A_{k,C}$, both procedures use the same fixed ordering
of the agents. Each recorded batch is an interval in this ordering.
Take the common refinement of the interval partitions induced by all
large-good and small-good batches. Its number of intervals is at most
one plus the total number of their endpoints and is therefore
$O_{K,\delta}(q)$, rather than the product of the numbers of batches.
On every interval of the refinement, all agents receive the same
combined bundle.

Together with the singleton and zero-configuration assignments, these
combined batches specify an individual allocation, which we denote by
$x^{\mathrm{rec}}$. Conceptually,
$x^{\mathrm{rec}}_{i,j}$ is the number of type-$j$ goods assigned to
agent $i$. The algorithm does not enumerate all these coordinates;
instead, it stores the bundle on each nonempty interval together with
the interval's binary-encoded size. Identical bundles belonging to the
same agent type can be merged by adding their multiplicities.

\medskip
\noindent\textbf{Utility guarantee.}
Consider an agent $i\in A_{k,C}$ represented by a regular configuration
$C\in\mathcal C_k$. Define $ L(C)
 :=
 \sum_{r\in\mathcal R_\delta}
 a_r(C)(1+\delta)^r.$ 
Thus $w(C)=L(C)+\mu(C).$ 
The large goods assigned to $i$ have total rounded normalized value
exactly $L(C)$. By \eqref{eq:individual-small-guarantee}, the total
rounded normalized value of the goods assigned to $i$ is at least
 $L(C)+(1-\delta)\mu(C)-1
 \geq  (1-\delta)w(C)-1.$
Every retained regular configuration satisfies
$w(C)\geq1/(5\delta)$. Therefore, $1\leq5\delta w(C)$, and $ (1-\delta)w(C)-1
 \geq
 (1-6\delta)w(C).$
 By~\eqref{eq:geometric-rounding}, $u_{k,j}
 \geq
 \alpha_k\overline u_{k,j}$
 for every good type $j$ having positive value to type $k$. Consequently,
\begin{equation}
\label{eq:reconstructed-regular-utility}
 u_k(x_i^{\mathrm{rec}})
 \geq
 (1-6\delta)\alpha_k w(C)
 =
 (1-6\delta)v(C).
\end{equation}

Now suppose that $i\in A_{k,j}^{\mathrm{sing}}$. If
$k\in\mathcal K_+$, then $(k,j)$ is a huge singleton pair, and $ u_k(x_i^{\mathrm{rec}})
 =
 u_{k,j}
 \geq
 \alpha_k\overline u_{k,j}
 =
 v_{k,j}.$
 If $k\in\mathcal K_0$, the singleton is represented exactly, and
 $ u_k(x_i^{\mathrm{rec}})
 =
 u_{k,j}
 =
 v_{k,j}.$ 
Every agent in $A_{k,0}$ initially receives no good having positive
value to type $k$ and therefore initially has utility zero.

\medskip
\noindent\textbf{Compressed representation and completion.}
The preceding endpoint bound gives $O_{K,\delta}(q)$ regular bundle
patterns. The singleton assignments contribute at most $Kq$ additional
patterns, and the zero configurations contribute at most $K$ patterns.
Every bundle coordinate and pattern multiplicity is obtained using
arithmetic on binary-encoded integers and rationals. The partial
allocation therefore has a polynomial-size compressed representation
and can be constructed in polynomial time.

Any goods that were not used, including unreserved goods and unused
reserved goods, may now be assigned to one fixed agent. This changes at
most one existing bundle pattern and introduces at most one new pattern.
Because all utilities are nonnegative, this completion cannot decrease
any agent's utility and preserves all the lower bounds established
above. Since the reservations obey
\eqref{eq:configuration-item-capacity} and no reserved copy is used more
than once, the resulting allocation is feasible. Omit records of
multiplicity zero and merge identical bundles within each agent type.
This gives the required compressed allocation and proves the lemma.
\end{proof}

\medskip
\noindent\textbf{Selection of the final candidate.}
For every enumerated threshold vector not discarded above and every
feasible associated restricted mixed-integer program, apply
Lemma~\ref{lem:assignment-rounding} followed by
Lemma~\ref{lem:compressed-reconstruction}. Let $\mathfrak X$ be the
resulting collection of feasible reconstructed allocations. For the
correct threshold vector, the first conclusion of
Lemma~\ref{lem:assignment-rounding} supplies an enumerated tuple
$\mathbf J^\circ$ whose restricted program yields a candidate. Hence
$\mathfrak X\neq\emptyset$.

Every $x\in\mathfrak X$ is represented by a polynomial-size collection
$\mathcal P(x)$ of records $(k,P,\lambda)$, indicating that $\lambda$
agents of type $k$ receive the bundle pattern
$P=(P_1,\ldots,P_q)$. Its utility to an agent of type $k$ is $u_k(P)$.
The objectives defined in Section~\ref{sec:notation} can be evaluated
directly from this compressed representation.  For $p\neq0$,
\[
 Q_p(x)
 =
 \sum_{(k,P,\lambda)\in\mathcal P(x)}
 \lambda u_k(P)^p,
\]
and, for $p=0$,
\[
 \Lambda(x)
 =
 \sum_{(k,P,\lambda)\in\mathcal P(x)}
 \lambda\log u_k(P).
\]
The rational coefficients used in the configuration formulations
approximate the represented values $v(C)$ and $v_{k,j}$. After
reconstruction, however, an agent's actual utility $u_k(P)$ need not
equal its represented value. We therefore compute fresh rational
approximations from the actual utilities of the reconstructed bundle
patterns. Because $\mathcal P(x)$ has polynomial size, this computation
does not require enumerating the agents.

For each $x\in\mathfrak X$, apply
Lemma~\ref{lem:weighted-rational-comparison} with accuracy $\delta$ to
the utility list with entries $(\lambda,u_k(P))$ indexed by
$\mathcal P(x)$. The weights sum to $n$. Denote the resulting rational
score by $\widetilde Q_p^{\mathrm{rec}}(x)$ when $0<p<1$, its negative
by $\widetilde Q_p^{\mathrm{rec}}(x)$ when $p<0$, and the score by
$\widetilde\Lambda^{\mathrm{rec}}(x)$ when $p=0$. The lemma gives,
using $\log(1+\delta)\leq\delta$ in the logarithmic case,
\begin{align*}
 \frac{Q_p(x)}{1+\delta}
 &\leq \widetilde Q_p^{\mathrm{rec}}(x)\leq Q_p(x)
 &&\text{if }0<p<1,\\
 Q_p(x)
 &\leq \widetilde Q_p^{\mathrm{rec}}(x)
 \leq(1+\delta)Q_p(x)
 &&\text{if }p<0,\\
 \Lambda(x)-n\delta
 &\leq \widetilde\Lambda^{\mathrm{rec}}(x)\leq\Lambda(x)
 &&\text{if }p=0.
\end{align*}
For $0<p<1$, a zero-utility pattern contributes zero exactly. When
$p\leq0$, all enumerated vectors with a zero threshold were discarded;
the reconstruction therefore gives every agent positive utility.
Choose $x^{\mathrm{out}}$ to maximize
$\widetilde Q_p^{\mathrm{rec}}$ if $0<p<1$, minimize it if $p<0$, and
maximize $\widetilde\Lambda^{\mathrm{rec}}$ if $p=0$.

The polynomial-time guarantee in
Lemma~\ref{lem:weighted-rational-comparison} applies directly to these
compressed lists. The selection rule maximizes the signed rational
score in all three cases. The comparison conclusion of the lemma
therefore implies that, for every $x\in\mathfrak X$,
\begin{equation}
\label{eq:final-candidate-comparison}
 \begin{cases}
 Q_p(x^{\mathrm{out}})
 \geq Q_p(x)/(1+\delta),
 & 0<p<1,\\[2mm]
 Q_p(x^{\mathrm{out}})
 \leq(1+\delta)Q_p(x),
 & p<0,\\[2mm]
 \Lambda(x^{\mathrm{out}})
 \geq\Lambda(x)-n\delta,
 & p=0.
 \end{cases}
\end{equation}
These bounds will be used in the approximation analysis.

\subsection{Approximation analysis}

Put
$
 \beta:=1-6\delta.
$
Since $\delta\leq1/10$, we have $0<\beta<1$ and
$
 1-4\delta\geq\beta.
$

Let $x^*$ be the canonical optimal allocation, and let $Z^*$ be the
integer configuration solution induced by $x^*$ under the correct
threshold vector, as constructed in
Lemma~\ref{lem:configuration-representation}. For every $k\in[K]$ and
$i\in A_k$, let
$\widetilde u_i^*$ denote its represented utility in $Z^*$. Equation~\eqref{eq:configuration-represented-utility} gives $\widetilde u_i^*
 \geq
 (1-4\delta)u_k(x_i^*)
 \geq
 \beta u_k(x_i^*)$
 whenever $u_k(x_i^*)>0$. When $0<p<1$ and
$u_k(x_i^*)=0$, the corresponding represented utility is also zero,
so the same inequality holds. Consequently,
\begin{equation}
\label{eq:representation-beta}
 \widetilde u_i^*
 \geq
 \beta u_k(x_i^*)
 \qquad k\in[K],\ i\in A_k.
\end{equation}

Apply the coefficient-loss analysis to $Z=Z^*$. For the correct
threshold vector, let $\mathbf J^\circ$ be the tuple guaranteed by the
first conclusion of Lemma~\ref{lem:assignment-rounding}, and let
$Z^{\mathrm{corr}}$ be the resulting rounded integral vector
$\overline Z$. Let $x^{\mathrm{corr}}\in\mathfrak X$ be its reconstruction
from Lemma~\ref{lem:compressed-reconstruction}. For every $k\in[K]$ and
$i\in A_k$, let
$\widetilde u_i^{\mathrm{corr}}$ denote its represented utility in
$Z^{\mathrm{corr}}$.

If $i$ is represented by a regular configuration, equation~\eqref{eq:reconstructed-regular-utility} gives
 $ u_k(x_i^{\mathrm{corr}})
 \geq
 \beta\widetilde u_i^{\mathrm{corr}}.$ 
If $i$ is represented by a singleton pair $(k,j)$, then
equation~\eqref{eq:reconstructed-singleton-utility} gives
 $ u_k(x_i^{\mathrm{corr}})
 \geq
 v_{k,j}
 =
 \widetilde u_i^{\mathrm{corr}}
 \geq
 \beta\widetilde u_i^{\mathrm{corr}},$ 
where the last inequality uses $0<\beta\leq1$.

Finally, a zero configuration has represented utility zero, and the
completion step can only increase the actual utility of its agent.
Therefore,
\begin{equation}
\label{eq:reconstruction-beta}
 u_k(x_i^{\mathrm{corr}})
 \geq
 \beta\widetilde u_i^{\mathrm{corr}}
 \qquad k\in[K],\ i\in A_k.
\end{equation}

\medskip
\noindent\textbf{The case $0<p<1$.}
If $Q_p(x^*)=0$, then the optimal welfare is zero and the approximation
claim is immediate. Hence assume $Q_p(x^*)>0$.
 Raising~\eqref{eq:representation-beta} to the positive power $p$ and
summing over the agents gives $Q_p^{\mathrm{conf}}(Z^*)
 \geq
 \beta^p Q_p(x^*).$
 Similarly, \eqref{eq:reconstruction-beta} gives
 $ Q_p(x^{\mathrm{corr}})
 \geq
 \beta^p Q_p^{\mathrm{conf}}(Z^{\mathrm{corr}}).$ 
Applying, in order,
\eqref{eq:final-candidate-comparison},
\eqref{eq:reconstruction-beta},
\eqref{eq:positive-configuration-comparison}, and
\eqref{eq:representation-beta} yields
\begin{align*}
 Q_p(x^{\mathrm{out}})
 &\geq
 \frac{1}{1+\delta}Q_p(x^{\mathrm{corr}})
 \geq
 \frac{\beta^p}{1+\delta}
 Q_p^{\mathrm{conf}}(Z^{\mathrm{corr}})
 \geq
 \frac{\beta^p}{(1+\delta)^2}
 Q_p^{\mathrm{conf}}(Z^*)
 \geq
 \frac{\beta^{2p}}{(1+\delta)^2}
 Q_p(x^*)
 .
\end{align*}
Therefore,
\[
 \frac{\PSW_p(x^{\mathrm{out}})}
 {\PSW_p(x^*)}
 =
 \left(
 \frac{Q_p(x^{\mathrm{out}})}
 {Q_p(x^*)}
 \right)^{1/p}
 \geq
 \beta^2(1+\delta)^{-2/p}.
\]

For every $c>0$, Bernoulli's inequality gives $(1+\delta)^{-c}\geq1-c\delta.$
 Moreover, $ \beta^2=(1-6\delta)^2\geq1-12\delta.$
 Since both $\beta^2$ and $(1+\delta)^{-2/p}$ belong to $[0,1]$, it
follows that $\beta^2(1+\delta)^{-2/p}
 \geq
 1-\left(12+\frac{2}{p}\right)\delta.$ 
Hence
\begin{equation}
\label{eq:positive-final-ratio}
 \frac{\PSW_p(x^{\mathrm{out}})}
 {\PSW_p(x^*)}
 \geq
 1-\left(12+\frac{2}{p}\right)\delta.
\end{equation}

\medskip
\noindent\textbf{The case $p<0$.}
Put $a:=|p|$, so that $p=-a$. The flow test in
Subsection~\ref{subsec:transformed-objectives-positivity} guarantees the
existence of a positive allocation and hence a positive canonical
optimum. Since threshold vectors with a zero coordinate are discarded
when $p\leq0$, every reconstructed candidate used here is also positive.

Because the exponent $p=-a$ is negative,
\eqref{eq:representation-beta} implies $ Q_p^{\mathrm{conf}}(Z^*)
 \leq
 \beta^{-a}Q_p(x^*),$ 
and \eqref{eq:reconstruction-beta} implies $Q_p(x^{\mathrm{corr}})
 \leq
 \beta^{-a}Q_p^{\mathrm{conf}}(Z^{\mathrm{corr}}).$
 Applying, in order,
\eqref{eq:final-candidate-comparison},
\eqref{eq:reconstruction-beta},
\eqref{eq:negative-configuration-comparison}, and
\eqref{eq:representation-beta}, we obtain
\begin{align*}
 Q_p(x^{\mathrm{out}})
 &\leq
 (1+\delta)Q_p(x^{\mathrm{corr}})
 \leq
 (1+\delta)\beta^{-a}
 Q_p^{\mathrm{conf}}(Z^{\mathrm{corr}})
 \leq
 (1+\delta)^2\beta^{-a}
 Q_p^{\mathrm{conf}}(Z^*)
 \leq
 (1+\delta)^2\beta^{-2a}
 Q_p(x^*).
\end{align*}

Raising this inequality to the power $-1/a$ reverses its direction.
Therefore,
\[
 \frac{\PSW_p(x^{\mathrm{out}})}
 {\PSW_p(x^*)}
 \geq
 \beta^2(1+\delta)^{-2/a}.
\]
As in the positive case, $\beta^2(1+\delta)^{-2/a}
 \geq
 1-\left(12+\frac{2}{a}\right)\delta.$
 Since $a=|p|$, we conclude that
\begin{equation}
\label{eq:negative-final-ratio}
 \frac{\PSW_p(x^{\mathrm{out}})}
 {\PSW_p(x^*)}
 \geq
 1-\left(12+\frac{2}{|p|}\right)\delta.
\end{equation}

\medskip
\noindent\textbf{The case $p=0$.}
The same positivity argument applies when $p=0$.

Taking logarithms in \eqref{eq:representation-beta} gives $ \Lambda^{\mathrm{conf}}(Z^*)
 \geq
 \Lambda(x^*)+n\log\beta.$
 Similarly,~\eqref{eq:reconstruction-beta} gives $ \Lambda(x^{\mathrm{corr}})
 \geq
 \Lambda^{\mathrm{conf}}(Z^{\mathrm{corr}})+n\log\beta.$
 Applying, in order,
\eqref{eq:final-candidate-comparison},
\eqref{eq:reconstruction-beta},
\eqref{eq:nash-configuration-comparison}, and
\eqref{eq:representation-beta} gives
\begin{align*}
 \Lambda(x^{\mathrm{out}})
 &\geq
 \Lambda(x^{\mathrm{corr}})-n\delta
 \geq
 \Lambda^{\mathrm{conf}}(Z^{\mathrm{corr}})
 +n\log\beta-n\delta
 \geq
 \Lambda^{\mathrm{conf}}(Z^*)
 +n\log\beta-2n\delta
\geq
 \Lambda(x^*)
 +2n\log\beta-2n\delta.
\end{align*}
Exponentiating after dividing by $n$ yields $ \frac{\PSW_0(x^{\mathrm{out}})}
 {\PSW_0(x^*)}
 \geq
 \beta^2e^{-2\delta}.$
 Since $e^{-2\delta}\geq1-2\delta$ and
$\beta^2\geq1-12\delta$, we obtain
\begin{equation}
\label{eq:nash-final-ratio}
 \frac{\PSW_0(x^{\mathrm{out}})}
 {\PSW_0(x^*)}
 \geq
 \beta^2e^{-2\delta}
 \geq
 1-14\delta.
\end{equation}

By the choice of $\delta$ in~\eqref{eq:constant-agent-delta}, equations
\eqref{eq:positive-final-ratio},
\eqref{eq:negative-final-ratio}, and
\eqref{eq:nash-final-ratio} then imply
\[
 \PSW_p(x^{\mathrm{out}})
 \geq
 (1-\varepsilon)\PSW_p(x^*)
 =
 (1-\varepsilon)\OPT_p.
\]

\subsection{Running time}

Let $L$ denote the binary encoding length of the instance. For fixed
$K$, $p$, and $\varepsilon$, the parameters $\delta$, $B_\delta$,
$D$, $F$, and $d$ are constants. If $U_k>0$, equation~\eqref{eq:Tk} gives $ |\mathcal T_k|
 =
 2+\left\lceil
 \log_{1+\delta}
 \left(\frac{n_kU_k}{u_k^{\min}}\right)
 \right\rceil
 =
 O_\delta(L).$
 If $U_k=0$, then $|\mathcal T_k|=1$. Hence the number of threshold
vectors is $\prod_{k=1}^K|\mathcal T_k|
 =
 O_\delta(L^K),$
 which is polynomial because $K$ is fixed.

For each threshold vector, the rounded values in
\eqref{eq:geometric-rounding} can be computed in polynomial time using
exact comparisons with integer powers of the fixed rational base
$1+\delta$. The exponents have magnitude $O_\delta(L)$, and the resulting
rational numbers have polynomial encoding length. By~\eqref{eq:retained},
the number $|\mathcal C|$ of retained regular configurations is
$O_{K,\delta}(1)$.

For every threshold vector, Section~\ref{sec:solving-configuration-formulation} enumerates at most
$
 (q+1)^{Kd}
$
support tuples. Each resulting restricted mixed-integer program has
$O(Kq+|\mathcal C|+K)$ variables and constraints, polynomially bounded
coefficient lengths, and at most $ |\mathcal C|+K+Kd=O_{K,\delta}(1)$
 integer variables. It can therefore be solved in polynomial time by
fixed-dimensional mixed-integer programming. The basic feasible
solution and the integral solution of the reduced system~\eqref{eq:reduced-assignment-system} can also
be computed in polynomial time; by Claim~\ref{clm:fractional-assignment-bound}, that system contains at
most $F=2D=O_{K,\delta}(1)$ fractional coordinates.

The reconstruction algorithms in Section~\ref{sec:compressed-reconstruction} produce
$O_{K,\delta}(q)$ assignment batches. Thus, the reconstructed
allocation can be computed in polynomial time. By
Lemma~\ref{lem:weighted-rational-comparison}, the same holds for its
directed rational objective approximation, without expanding the
binary-encoded agent or good multiplicities.

Consequently, the total running time is $\left(\prod_{k=1}^K|\mathcal T_k|\right)
 (q+1)^{Kd}
 \operatorname{poly}_{K,p,\delta}(L)
 =
 \operatorname{poly}_{K,p,\varepsilon}(L),$
 where $q\leq L$. Since the enumeration contains a correct threshold
vector and the tuple $\mathbf J^\circ$ guaranteed by
Lemma~\ref{lem:assignment-rounding}, the approximation analysis applies
to the best candidate returned by the algorithm.
\end{proof}

\section{Conclusion}
The exact algorithms and hardness results establish sharp item-type
boundaries for computing an optimal allocation: hardness starts at two
item types for every fixed finite rational $p<1$, $p\neq0$, and at three
for egalitarian welfare. Nash welfare is NP-hard with three item types;
its complexity with two item types remains unsettled by our results,
already for explicitly listed agents and binary item multiplicities.
The hardness results also apply to rational-threshold decision, whereas
the polynomial exact-allocation results need not: even the one-item-type
case $p=-1/2$ contains \textsc{Sqrt-Sum} as a threshold problem.

For every fixed number of item types, we give an FPTAS for every fixed
$p\in\mathbb Q\cup\{-\infty\}$. For a fixed number of agent types and
an unrestricted number of item types, we give a PTAS for every fixed
finite rational $p<1$. Both algorithms run in the compact input model and
return compressed allocations. The one-agent-type strong-hardness
reduction and its quantitative gap rule out an FPTAS in the latter
setting unless $\mathsf{P}=\mathsf{NP}$.

The exact algorithms for one item type and for egalitarian welfare with
two item types follow classical resource-allocation and scheduling
methods. Our proofs make their compact-input implementation and output
explicit. Appendix~\ref{sec:exact-egalitarian-appendix} similarly records
the known exact egalitarian algorithm when both numbers of types are
fixed, as a consequence of Goemans and Rothvo{\ss}~\cite{GoemansR20}.

Our work leaves several important problems open for future research. The first is to complete the picture of approximability by settling the complexity of maximizing Nash welfare with two item types. A second natural question is whether the problem can be solved exactly in polynomial time when both the number of agent types and the number of item types are constant; so far, this has been answered affirmatively only for egalitarian welfare. Finally, a promising direction is to study the computation of fair allocations under standard fairness criteria, including envy-freeness (EF) and its relaxations EF1 and EFX, proportionality (PROP) and PROP1, and the maximin share (MMS). In particular, it would be interesting to determine whether computing an allocation that satisfies such a criterion is tractable or $\np$-hard when the number of item types is small, especially when fairness is combined with efficiency requirements such as Pareto optimality or maximum Nash welfare. For instance, \cite{azizchore23} gives an algorithm for computing an EF allocation with two item types whose running time is polynomial in $n$ and $m$; an algorithm that is polynomial in $n$ and $\log m$, and hence in the input size when item multiplicities are encoded in binary, would be highly desirable.

\paragraph{Use of generative AI.}
The authors used ChatGPT iteratively to refine and improve an initial draft of this paper. The authors have verified all proofs and accept full responsibility for the accuracy and correctness of the results presented in the paper.

\bibliographystyle{plain}
	\bibliography{ref}

\appendix

\section{Proofs of the Technical Lemmas in Sections~\ref{sec:notation} and~\ref{sec:hardness}}\label{sec:tech-lemmas}

\subsection{Proof of Lemma~\ref{lem:directed-rational-approximation}}
\begin{proof}
The case $p>0$ and $v=0$ is immediate, so suppose that $v>0$, and let
$B$ be its encoding length. First let $p=a/d\neq0$, where the nonzero
integer $a$ and the positive integer $d$ are fixed. We have
$2^{-B}\leq v\leq2^B$. Thus, for a sufficiently large integer constant
$c$ depending only on $p$, the dyadic numbers
$l_0:=2^{-cB}$ and $u_0:=2^{cB}$ satisfy
$0<l_0\leq v^p\leq u_0$.
Set $\xi:=\eta/(1+\eta)$. Starting with the interval
$[l_0,u_0]$, repeatedly bisect the current interval and retain
a half containing $v^p$, until the interval has width at most
$\xi l_0$. Each comparison is exact: for a positive rational
midpoint $z$, compare $z^d$ with the rational number $v^a$, also when
$a<0$. After $O_p(B+\log(1/\eta))$ iterations we obtain $[l,u]$ with $ l\leq v^p\leq u,$ 
 $ u-l\leq\xi l_0\leq\xi v^p.$ 
For $p>0$, output $l$. Then
$l\geq(1-\xi)v^p=v^p/(1+\eta)$.
For $p<0$, output $-u$. Then
$u\leq(1+\xi)v^p\leq(1+\eta)v^p$.
The endpoints and the intermediate powers have encoding length
polynomial in $B+\log(1/\eta)$. The exact stopping test also uses the
encoding of $\eta$, so the complete computation has the claimed bit
complexity.

For $p=0$, write $v=2^e y$ with $e\in\mathbb Z$, $1\leq y<2$, and
$|e|=O(B)$. For a rational $t\in[1,2]$, put
$z=(t-1)/(t+1)\in[0,1/3]$. The identity $ \log t=2\sum_{\ell=0}^{\infty}\frac{z^{2\ell+1}}{2\ell+1}$
 shows that the rational partial sum $ P_N(t):=2\sum_{\ell=0}^{N-1}\frac{z^{2\ell+1}}{2\ell+1}$
 satisfies, for every integer $N\geq1$,
\begin{align*}
 0\leq\log t-P_N(t)
 =2\sum_{\ell=N}^{\infty}\frac{z^{2\ell+1}}{2\ell+1}
 \leq\frac{2}{2N+1}
       \sum_{\ell=N}^{\infty}z^{2\ell+1}
 =\frac{2z^{2N+1}}{(2N+1)(1-z^2)}
 \leq\frac{3}{4(2N+1)}\,9^{-N}
 \leq9^{-N}.
\end{align*}
Put $h:=9^{-N}$. The preceding bounds give
\[
 \log y\in[P_N(y),P_N(y)+h],
 \qquad
 \log2\in[P_N(2),P_N(2)+h].
\]
Since $\log v=\log y+e\log2$, it follows that
\(
 \log v\in[L_N,U_N],
\)
where
\[
 L_N:=P_N(y)+eP_N(2)+\min\{e,0\}h,
 \qquad
 U_N:=P_N(y)+eP_N(2)+(1+\max\{e,0\})h.
\]
Indeed, if $e\geq0$, multiplication by $e$ preserves
the endpoint order; if $e<0$, it reverses that order.
Both endpoints are rational, and $ U_N-L_N=(1+|e|)9^{-N}.$ 
Choose $N$ so that this width is at most $\eta/2$,
and output the lower endpoint $L_N$. Since
\[
 \log(1+\eta)=\int_0^\eta\frac{dt}{1+t}
 \geq\frac{\eta}{1+\eta}\geq\frac\eta2,
\]
this output satisfies~\eqref{eq:prelim-rational-logarithmic}.
Only $O(1+\log(B+1)+\log(1/\eta))$ terms are required. Their exact
rational evaluation and the test for choosing $N$ take time polynomial
in the encoding lengths of $v$ and $\eta$. The output has polynomial
encoding length as well.
\end{proof}

\subsection{Proof of Lemma~\ref{lem:two-type-positive-support}}

\begin{proof}
For $y=0$, the function in~\eqref{eq:two-type-positive-support} is $ F(x,0)=\frac{a^{1-p}x^p}{p}-x.$
 For real $x>0$, its derivative is $(a/x)^{1-p}-1$.
Consequently, its unique maximum on $[0,\infty)$ is attained at
$x=a$, with value $\gamma a$.

Now let $y\geq1$. The partial derivative with respect to $x$ is $ \alpha(\alpha x+\beta y)^{p-1}-1
 \leq \alpha\beta^{p-1}-1
 =\left(\frac aH\right)^{\gamma}-1<0,$
 where $H>a$. Thus, for fixed $y\geq1$, the maximum is attained
uniquely at $x=0$. At $x=0$ the function is
$Hy^p/p-Cy$, whose derivative on $[1,\infty)$ is at most $ H-C=(1-p)(a-C)<0.$ 
Its maximum on that interval is attained uniquely at $y=1$ and
equals $H/p-C=\gamma a$. This proves the inequality and its
equality cases.

It remains to bound the loss at the other integer pairs. Suppose that
$y=0$ and $x>0$, with $x\neq a$. Every real number $t$ between $x$ and
$a$ satisfies $1\leq t\leq Z$. Hence
\[
 -\frac{\partial^2F}{\partial x^2}(t,0)
 =(1-p)a^{1-p}t^{p-2}
 =(1-p)\frac{a^{1-p}t^p}{t^2}
 \geq\frac{1-p}{Z^2},
\]
because $a^{1-p}\geq1$, $t^p\geq1$, and $t^2\leq Z^2$.
Since $\frac{\partial F}{\partial x}(a,0)=0$, Taylor's theorem gives
\begin{equation}\label{e1-1}
 F(a,0)-F(x,0)
 \geq\frac{1-p}{2Z^2}(x-a)^2
 \geq\frac{1-p}{2Z^2},
\end{equation}
where the last inequality uses that $x$ and $a$ are distinct integers.
At $(0,0)$, the loss is $\gamma a\geq1-p$.

Suppose that $y\geq1$ and $x\geq1$. For fixed $y$, the preceding partial derivative bound gives
\[
 \frac{\partial F}{\partial x}(t,y)
 \leq \left(\frac aH\right)^\gamma-1
 \qquad\text{for every }t\in[0,x].
\]
Since $F(0,y)\leq\gamma a$, integrating this inequality yields
\begin{align*}
 \gamma a-F(x,y)
 &\geq F(0,y)-F(x,y)=-\int_0^x\frac{\partial F}{\partial x}(t,y)\,dt\geq x\left[1-\left(\frac aH\right)^\gamma\right]\geq 1-\left(\frac aH\right)^\gamma,
\end{align*}
where the last step uses $x\geq1$ and $H>a$.

For every $z\in(0,1)$ and $\gamma>0$, $1-z^\gamma\geq\min\{1,\gamma\}(1-z).$ Indeed, if $\gamma\geq1$, then $z^\gamma\leq z$; if
$0<\gamma<1$, concavity gives
$z^\gamma\leq1+\gamma(z-1)$.
Applying this inequality with $z=a/H$, we obtain
\begin{align}\label{e1-2}
 \gamma a-F(x,y)
 &\geq \min\{1,\gamma\}\left(1-\frac aH\right)=\min\{1,\gamma\}\frac{p(C-a)}{H}=\min\{p,1-p\}\frac{C-a}{H}\geq\frac{\min\{p,1-p\}}{Z}.
\end{align}
Here we used $H-a=p(C-a)$, $\gamma=(1-p)/p$,
$C-a\geq1$, and $H\leq Z$.
Finally, suppose that $x=0$ and $y\geq2$. Since $ F(0,t)=\frac{Ht^p}{p}-Ct,$
 for every $t\geq1$ we have $\frac{\partial F}{\partial t}(0,t)
 =Ht^{p-1}-C
 \leq H-C
 =-(1-p)(C-a),$
 where we used $p<1$ and hence $t^{p-1}\leq1$.
Since $F(0,1)=\gamma a$, integration gives
\begin{align}\label{e1-3}
 \gamma a-F(0,y)
 &=F(0,1)-F(0,y)=-\int_1^y\frac{\partial F}{\partial y}(0,t)\,dt\geq(y-1)(1-p)(C-a)\geq1-p,
\end{align}
because $y-1\geq1$ and $C-a\geq1$. Equations~\eqref{e1-1},~\eqref{e1-2} and~\eqref{e1-3} imply
\eqref{eq:two-type-positive-gap}.
\end{proof}

\subsection{Proof of Lemma~\ref{lem:two-type-negative-support}}
\begin{proof}
The assumptions imply $0<H<a$. For $y=0$ and real $x>0$,
\[
 G(x,0)=\frac{a^{r+1}x^{-r}}r+x,
 \qquad
 \frac{\partial G}{\partial x}(x,0)
 =1-\left(\frac ax\right)^{r+1}.
\]
This function has its unique minimum at $x=a$, with value
$\gamma a$.

For $y\geq1$ and $x\geq0$, the partial derivative with respect
to $x$ satisfies
\[
 1-\alpha(\alpha x+\beta y)^{-r-1}
 \geq1-\alpha\beta^{-r-1}
 =1-\left(\frac Ha\right)^\gamma>0.
\]
Thus the minimum for fixed $y\geq1$ is attained uniquely at
$x=0$. At $x=0$ the function is $Hy^{-r}/r+Cy$, whose
derivative on $[1,\infty)$ is at least $ C-H=(r+1)(C-a)>0.$ 
Its minimum is therefore attained uniquely at $y=1$, with
value $H/r+C=\gamma a$.

For the quantitative bound, first suppose that $y=0$ and
$x\neq a$. On the interval between the positive integers $x$
and $a$,
\[
 \frac{\partial^2G}{\partial t^2}(t,0)
 =(r+1)a^{r+1}t^{-r-2}
 \geq Z^{-\lceil r\rceil-2}.
\]
Taylor's theorem gives an excess of at least
$1/(2Z^{\lceil r\rceil+2})$.
If $y\geq1$ and $x\geq1$, increasing $x$ from zero causes an
excess of at least $1-\left(\frac Ha\right)^\gamma
 \geq1-\frac Ha
 =\frac{r(C-a)}a\geq\frac rZ,$ 
because $\gamma>1$. If $x=0$ and $y\geq2$, increasing $y$
from one causes an excess of at least
$(r+1)(C-a)\geq1$. This proves
\eqref{eq:two-type-negative-gap}.
\end{proof}

\subsection{Proof of Lemma~\ref{lem:two-type-rational-between-powers}}
\begin{proof}
Put $R:=\frac{R_0+R_1}{2},$  $\Delta:=R_1-R_0,$ and $ B:=1+\lceil R_1\rceil.$ 
We will compute a rational approximation to the positive root
$\alpha:=R^{1/s}$. Since $B\geq1$ and $B^s\geq B>R$, we have
$0<\alpha<B$. The root $\alpha$ is used only in the analysis;
the algorithm does not need to evaluate it.

Start with the rational interval $[\ell,u]=[0,B]$, which contains
$\alpha$. At each step, compute its midpoint $ m:=\frac{\ell+u}{2}$ and compare $m^s$ with $R$. If $m^s<R$, replace $\ell$ by $m$;
if $m^s>R$, replace $u$ by $m$. If $m^s=R$, return $z=m$,
which already satisfies the required inequalities.
Because $t\mapsto t^s$ is strictly increasing on $[0,\infty)$,
each update retains $\alpha$ and halves the interval's width.

All comparisons are exact rational comparisons. Specifically,
if $m=a/b$ and $R=c/d$, with positive denominators, then comparing
$m^s$ with $R$ is equivalent to comparing the integers
$da^s$ and $cb^s$.

Continue until $u-\ell\leq
 \eta:=\frac{\Delta}{2sB^{s-1}},$
 and return the rational midpoint $z=(\ell+u)/2$.
Since $\alpha\in[\ell,u]$, we have $ |z-\alpha|
 \leq\frac{u-\ell}{2}
 \leq\frac{\Delta}{4sB^{s-1}}.$ 
The derivative of $t^s$ on $[0,B]$ is at most $sB^{s-1}$.
The mean value theorem therefore gives $|z^s-R|
 =|z^s-\alpha^s|
 \leq sB^{s-1}|z-\alpha|
 \leq\frac{\Delta}{4}.$
 Since $R$ is the midpoint of $[R_0,R_1]$, it follows that $R_0
 <R-\frac{\Delta}{4}
 \leq z^s
 \leq R+\frac{\Delta}{4}
 <R_1.$

Finally, let $L$ be the total binary encoding length of $R_0$
and $R_1$. The integer $B$ has $O(L)$ bits, and the positive
rational number $\Delta$ satisfies $\Delta\geq2^{-O(L)}$.
Thus the number of bisection steps is at most
\[
 \left\lceil
 \log_2\left(\frac{B}{\eta}\right)
 \right\rceil
 =
 \left\lceil
 \log_2\left(\frac{2sB^s}{\Delta}\right)
 \right\rceil
 =O_s(L).
\]
Every endpoint and midpoint is a dyadic rational with
$O_s(L)$ bits. Since $s$ is fixed, the integer powers and
comparisons at each step take polynomial time. This proves
both the running-time and output-length bounds.
\end{proof}

\section{An Illustrative Example for the proof of Lemma~\ref{lem:fixed-item-universal-partition}}\label{sec:example-fixed-item-universal-partition}
Let $N=27$ and $s=3$. We first construct the partition independently
of the histogram. At each nonterminal level, choose
$a_\ell\in\{2,3\}$ with the same parity as $N_\ell$, and set
$N_{\ell+1}=(N_\ell-a_\ell)/2$. This gives
\[
\begin{array}{c|c|c|l}
 \text{Level }\ell & N_\ell & a_\ell & \text{Parts created}\\ \hline
 0 & 27 & 3 & 1,1,1\\
 1 & 12 & 2 & 2,2\\
 2 & 5  & 3 & 4,4,4\\
 3 & 1  & 1 & 8
\end{array}
\]
The construction terminates at level $3$, since $N_3=1<s$.
Thus the fixed partition is $27=1+1+1+2+2+4+4+4+8.$

Now consider the histogram $h=(13,8,6)$, with categories labeled
$A,B,C$. We assign categories to the parts as follows. At level
$\ell$, each unit of the current histogram represents $2^\ell$
units of the original histogram.

\begin{enumerate}
\item
At level $0$, the current histogram is $(13,8,6)$ and $a_0=3$.
Only the first entry is odd. Remove one unit from category $A$,
leaving $(12,8,6)$. To remove exactly three units, remove a further
pair from category $B$, leaving $(12,6,6)$. Assign the three
size-$1$ parts to $A,B,B$, respectively. Halving the remaining
histogram gives $ (12,6,6)/2=(6,3,3).$

\item
At level $1$, the current histogram is $(6,3,3)$ and $a_1=2$.
Remove one unit from each of the odd entries, corresponding to
categories $B$ and $C$. Assign the two size-$2$ parts to these
categories. Halving the remaining histogram gives $ (6,2,2)/2=(3,1,1).$

\item
At level $2$, the current histogram is $(3,1,1)$ and $a_2=3$.
All three entries are odd, so remove one unit from each category.
Assign the three size-$4$ parts to $A,B,C$, respectively.
Halving the remaining histogram gives $ (2,0,0)/2=(1,0,0).$

\item
At level $3$, the current histogram is $(1,0,0)$.
This is the terminal level, so assign the single size-$8$ part
to category $A$.
\end{enumerate}

The parts assigned to each category recover its original
multiplicity:
\[
\begin{aligned}
 A &: 1+4+8=13,\\
 B &: 1+1+2+4=8,\\
 C &: 2+4=6.
\end{aligned}
\]
The part sizes depend only on $N$ and $s$, whereas their category
labels depend on the histogram. The proof shows that the same
fixed partition can represent every histogram on at most three
categories with total multiplicity $27$.

\section{Exact Egalitarian Welfare with Few Agent and Item Types}
\label{sec:exact-egalitarian-appendix}

For completeness, we give the proof of
Proposition~\ref{prop:known-exact-egalitarian}. The exact polynomial-time
algorithm is a direct corollary of Theorem~12 of Goemans and
Rothvo{\ss}~\cite{GoemansR20}; the role of the following construction is to
express the allocation constraints and a prescribed minimum utility in the
form required by their theorem. Throughout this appendix, both the number
$K$ of agent types and the number $q$ of item types are fixed, and both
agent and item multiplicities are encoded in binary.

For a set $X\subseteq\mathbb Z^d$, its \emph{integer cone} is
\[
 \operatorname{int.cone}(X)
 :=
 \left\{
  \sum_{x\in X}\lambda_x x:
  \lambda_x\in\mathbb Z_{\geq0},\quad
  \bigl|\{x\in X:\lambda_x>0\}\bigr|<\infty
 \right\}.
\]
The empty sum is permitted. For a rational polyhedron $P$, let
$\operatorname{enc}(P)$ denote the binary encoding length of an inequality
description of $P$. We use the following bounded-polytope form of the
integer-cone theorem.

\begin{theorem}[Goemans--Rothvo{\ss}, Theorem~12~\cite{GoemansR20}]
\label{thm:GR-union}
Fix positive integers $d$ and $h$. Given bounded rational polytopes
$P_1,\ldots,P_h\subseteq\mathbb R^d$ and a rational polyhedron
$Q\subseteq\mathbb R^d$, one can decide whether
\[
 \operatorname{int.cone}
 \left(
  \bigcup_{r=1}^h(P_r\cap\mathbb Z^d)
 \right)
 \cap Q
 \neq\varnothing
\]
in time
\[
 \left(\sum_{r=1}^h\operatorname{enc}(P_r)\right)^{2^{O(d+h)}}
 \operatorname{enc}(Q)^{O(1)}.
\]
In the affirmative case, the algorithm returns integer points
$x\in P_r\cap\mathbb Z^d$ and nonnegative integral multiplicities
$\lambda_{r,x}$ such that
\[
 \sum_{r=1}^h\sum_{x\in P_r\cap\mathbb Z^d}\lambda_{r,x}x\in Q.
\]
The returned multiplicities have polynomial binary encoding length, and
the number of pairs $(r,x)$ with $\lambda_{r,x}>0$ is at most
$2^{2(d+h)+1}$.
\end{theorem}

\begin{proof}[Proof of Proposition~\ref{prop:known-exact-egalitarian}]
We use the notation and compressed representation from
Section~\ref{sec:notation}. In particular, $n\geq1$. Delete agent types
of multiplicity zero, so every retained $n_k$ is positive. If $m=0$, every
agent receives the empty bundle and the optimum egalitarian welfare is
zero. Listing the empty bundle with multiplicity $n_k$ for each agent
type $k$ gives the required compressed allocation. Henceforth, assume
$m>0$ and delete item types of multiplicity zero; thus $q\geq1$ and
$m_j\geq1$ for every $j\in[q]$.

\medskip
\noindent\textbf{Integral utilities and the search interval.}
Write
$u_{k,j}=a_{k,j}/d_{k,j}$, where
$a_{k,j}\in\mathbb Z_{\geq0}$ and $d_{k,j}\in\mathbb Z_{\geq1}$, and
define
\begin{equation}
\label{eq:exact-esw-scaling}
 D:=\prod_{k=1}^K\prod_{j=1}^q d_{k,j},
 \qquad
 \widehat u_{k,j}:=D u_{k,j},
 \qquad
 \widehat u_k(c):=\sum_{j=1}^q\widehat u_{k,j}c_j.
\end{equation}
All scaled item utilities are nonnegative integers. The binary encoding
length of $D$ is at most one plus the sum of the binary encoding lengths
of the input denominators, so this scaling takes polynomial time and
produces numbers of polynomial encoding length. Scaling every utility by
the same positive factor $D$ multiplies the minimum utility of every
allocation by $D$ and preserves its optimality.

Put
\begin{equation}
\label{eq:exact-esw-value-bound}
 U:=\max\left\{1,\max_{k\in[K],\,j\in[q]}\widehat u_{k,j}\right\},
 \qquad
 V:=\max\left\{1,\max_{k\in[K]}
                    \sum_{j=1}^q m_j\widehat u_{k,j}\right\}.
\end{equation}
Every utility in the scaled instance is an integer between zero and $V$.
Moreover, $V\leq q m_{\max}U$, so $\log(V+1)$ is polynomially bounded
by the binary encoding length of the original instance. The optimum
scaled egalitarian welfare is therefore an integer in $\{0,\ldots,V\}$.

\medskip
\noindent\textbf{Feasibility of a minimum-utility threshold.}
Let $e_k$ be the $k$th unit vector of $\mathbb Z^K$. For an integer
$T\in\{0,\ldots,V\}$ and each agent type $k\in[K]$, define the
tagged configuration polytope
\begin{equation}
\label{eq:exact-esw-polytopes}
 P_k(T):=
 \left\{
  (c,y)\in\mathbb R^{q+K}:
  y=e_k,\quad
  0\leq c_j\leq m_j\ (j\in[q]),\quad
  \widehat u_k(c)\geq T
 \right\}.
\end{equation}
This is a bounded rational polytope, possibly empty. Its integer points
are precisely the configurations giving scaled utility at least $T$ to
one agent of type $k$, with the last $K$ coordinates recording that
agent's type. The configuration bounds allow an empty bundle when
$T=0$ and do not exclude any bundle in a feasible allocation. Define the
singleton target polyhedron
\begin{equation}
\label{eq:exact-esw-target}
 Q^{\mathrm{alloc}}:=\{(\boldsymbol m,\boldsymbol n)\}
 \subseteq\mathbb R^{q+K}.
\end{equation}
There is an allocation with scaled egalitarian welfare at least $T$ if
and only if
\begin{equation}
\label{eq:exact-esw-cone-feasibility}
 \operatorname{int.cone}
 \left(
  \bigcup_{k=1}^K(P_k(T)\cap\mathbb Z^{q+K})
 \right)
 \cap Q^{\mathrm{alloc}}
 \neq\varnothing.
\end{equation}

To prove the forward implication, let $x$ be such an allocation. For
each agent $i\in A_k$, the vector $(x_i,e_k)$ belongs to
$P_k(T)\cap\mathbb Z^{q+K}$. Summing these integer points over all
agents gives $(\boldsymbol m,\boldsymbol n)$: the first $q$ coordinates
express the item-allocation equalities, and the remaining $K$
coordinates count the agents of each type. Grouping repeated points
gives the required integer-cone representation.

Conversely, suppose~\eqref{eq:exact-esw-cone-feasibility} holds, and let
$\lambda_{k,c}\in\mathbb Z_{\geq0}$ be the multiplicity of the tagged
point $(c,e_k)$ in a corresponding decomposition. Comparing the last
$K$ coordinates with the target gives $ \sum_c\lambda_{k,c}=n_k$ for  $k\in[K]$, 
and comparing the first $q$ coordinates gives $\sum_{k=1}^K\sum_c\lambda_{k,c}c_j=m_j
 \qquad(j\in[q]).$ 
Thus, for every $k$, assigning configuration $c$ to
$\lambda_{k,c}$ agents of type $k$ assigns a bundle to every agent and
uses exactly the available multiplicity of every item type. Every used
configuration satisfies $\widehat u_k(c)\geq T$, so the resulting
allocation has scaled egalitarian welfare at least $T$. The
type--configuration pairs and their multiplicities already constitute
its compressed representation.

Apply Theorem~\ref{thm:GR-union} with dimension $d=q+K$, with $h=K$
polytopes, and with $Q=Q^{\mathrm{alloc}}$. Since $K$ and $q$ are fixed,
this decides~\eqref{eq:exact-esw-cone-feasibility} in polynomial time
and, when feasible, returns a corresponding compressed allocation.

\medskip
\noindent\textbf{Exact optimization and reconstruction.}
Threshold zero is feasible, since all utilities are nonnegative and all
items can be assigned to an arbitrary agent. Threshold feasibility is
monotone: if an allocation attains a threshold $T$, it attains every
smaller nonnegative threshold. Binary search on
$\{0,\ldots,V\}$ therefore finds the largest feasible threshold $T^*$
using $O(\log(V+1))$ integer-cone feasibility tests. Invoke the same
algorithm at $T^*$ to obtain a decomposition, omit zero multiplicities,
and merge repeated type--configuration pairs. The resulting compressed
allocation attains the exact optimum of the scaled instance, because
every attainable minimum utility is an integer. Its welfare in the
original instance is exactly $T^*/D$.

This argument also covers zero utilities. In particular, if an agent
type values all available items at zero, every positive threshold is
infeasible and the algorithm returns an optimal allocation of welfare
zero. If all item utilities are zero, the convention in
\eqref{eq:exact-esw-value-bound} gives $V=1$ and the same search returns
$T^*=0$. Item types that are worthless to every agent may be retained in
the polytopes; the target still ensures that all their copies are
allocated.

\medskip
\noindent\textbf{Bit complexity and compressed output.}
Let $L$ denote the binary encoding length of the original instance.
For each tested $T$, the coefficients in
\eqref{eq:exact-esw-polytopes} have polynomial encoding length in $L$,
as do the coordinates of the singleton target in
\eqref{eq:exact-esw-target}. Equalities can be written as pairs of
inequalities. Consequently,
\[
 \sum_{k=1}^K\operatorname{enc}(P_k(T))=\operatorname{poly}(L),
 \qquad
 \operatorname{enc}(Q^{\mathrm{alloc}})=\operatorname{poly}(L).
\]
The exponent in Theorem~\ref{thm:GR-union} depends only on $K$ and $q$,
and the number of feasibility tests is polynomial in $L$. The overall
running time is therefore $\operatorname{poly}_{K,q}(L)$.

The returned decomposition has at most
$2^{2(q+2K)+1}$ nonzero type--configuration pairs. Every configuration
coordinate lies between zero and the corresponding $m_j$, and the
type-count equalities imply $\lambda_{k,c}\leq n_k\leq n$ for every
pair. Thus every configuration and multiplicity has polynomial binary
encoding length. The algorithm outputs these pairs directly, without
expanding either the binary-encoded agent population or the item
multisets. This proves the claimed exact polynomial-time algorithm and
the polynomial size of its compressed output.
\end{proof}

\end{document}